\documentclass[onecolumn,draftclsnofoot]{IEEEtran}
\usepackage{amsmath}
\usepackage{graphicx}
\usepackage{subcaption}
\usepackage{amsthm}
\usepackage{float}
\usepackage{mathrsfs}
\usepackage{verbatim}
\usepackage{epstopdf}
\usepackage{amssymb}
\usepackage{amsfonts}
\usepackage{color}
\usepackage{cite}
\usepackage{cancel}
\usepackage[shortlabels]{enumitem}
\usepackage[breaklinks=true,letterpaper=true,colorlinks=false,bookmarks=false]{hyperref}
\usepackage{algorithm}
\usepackage{algpseudocode}
\usepackage{arydshln}
\usepackage[official]{eurosym}
\usepackage{comment}
\usepackage{amsmath,amssymb,amsthm,mathrsfs,amsfonts,dsfont}
\usepackage[shortlabels]{enumitem}
\DeclareMathOperator*{\argmin}{arg\,min}
\DeclareMathOperator*{\argmax}{arg\,max}
\newcommand\independent{\protect\mathpalette{\protect\independenT}{\perp}}
\def\independenT#1#2{\mathrel{\rlap{$#1#2$}\mkern2mu{#1#2}}}

\newtheorem{theorem}{Theorem}
\newtheorem{proposition}{Proposition}
\newtheorem{remark}{Remark}
\newtheorem{corollary}{Corollary}[theorem]
\newtheorem{example}{Example}
\newtheorem{lemma}{Lemma}

\newtheorem{definition}{Definition}

\usepackage{tikz}
\usetikzlibrary{dsp,chains}
\usepackage{pgfplots}
\pgfplotsset{compat = newest}
\usetikzlibrary{shapes.misc, positioning}
\usetikzlibrary{pgfplots.external}
\usetikzlibrary{external}

\pgfplotscreateplotcyclelist{nano mark style}{%
dashed, every mark/.append style={solid, fill=gray!30}, mark=*, line width=1.1pt\\%
dashed, every mark/.append style={solid, fill=gray!30}, mark=square*, line width=1.1pt\\%
dashed, every mark/.append style={solid, fill=gray!30}, mark=triangle*, mark size=2.5pt, line width=1.1pt\\%
dashed, every mark/.append style={solid, fill=gray!30},mark=diamond*, mark size=2.5pt, line width=1.1pt\\%
dashed, every mark/.append style={solid, fill=gray!30},mark=x, mark size=3pt, line width=1.1pt\\%
}

\pgfplotscreateplotcyclelist{nano line style}{%
color = black, line width=2pt\\%
densely dashed, color = black, line width=2pt\\%
densely dotted, color = black, line width=2pt\\%
loosely dashed, color = black, line width=2pt\\%
}
\begin{document}
\title{Information-Theoretic Privacy-Preserving Schemes Based On Perfect Privacy}
\author{Borzoo Rassouli$^1$ and Deniz G\"und\"uz$^2$\\
\small{$^1$ Nokia Bell Labs, Munich, Germany}\\
\small{$^2$ Department of Electrical and Electronic Engineering, Imperial College London, London, UK}\\
{\tt\small borzoo.rassouli@nokia-bell-labs.com}, {\tt\small  d.gunduz@imperial.ac.uk}
}

\maketitle
\begin{abstract}
Consider a pair of random variables $(X,Y)$ distributed according to a given joint distribution $p_{XY}$. A curator wishes to maximally disclose information about $Y$, while limiting the information leakage incurred on $X$. Adopting mutual information to measure both utility and privacy of this information disclosure, the problem is to maximize $I(Y;U)$, subject to $I(X;U)\leq\epsilon$, where $U$ denotes the released random variable and $\epsilon$ is a given privacy threshold. Two settings are considered, where in the first one, the curator has access to $(X,Y)$, and hence, the optimization is over $p_{U|XY}$, while in the second one, the curator can only observe $Y$ and the optimization is over $p_{U|Y}$. In both settings, the utility-privacy trade-off is investigated from theoretical and practical perspective. More specifically, several privacy-preserving schemes are proposed in these settings based on generalizing the notion of statistical independence. Moreover, closed-form solutions are provided in certain scenarios. Finally, convexity arguments are provided for the utility-privacy trade-off as functionals of the joint distribution $p_{XY}$.
\end{abstract}
% \begin{IEEEkeywords}
% Privacy, perfect privacy, non-private information, mutual information.
% \end{IEEEkeywords}
\section{Introduction}
% With the explosion of machine learning algorithms, and their applications in many areas of science, technology, and governance, data is becoming an extremely valuable asset. However, with the growing power of machine learning algorithms in learning individual behavioral 
% patterns from diverse data sources, privacy is becoming a major concern, calling for strict regulations on data ownership and distribution. On the other hand, many recent examples of de-anonymization attacks on publicly available anonymized data ({\color{black}e.g.,}\!\!\cite{Nar},\cite{Ding}) show that regulation alone will not be sufficient to limit access to private data. An alternative approach, also considered in this paper, is to process the data at the time of release, such that no private information is leaked, called \textit{perfect privacy}. Assuming that the joint distribution of the observed data and the latent variables that should be kept private is known, an information-theoretic study is carried out in this paper to characterize the fundamental limits on perfect privacy. 

Consider a situation in which Alice wants to release some  \textit{useful} information about herself to Bob, represented by random variable $Y$, and she receives some utility from this disclosure of information. At the same time, she wishes to conceal from Bob some \textit{private} information which depends on $Y$, represented by $X$. To this end, instead of letting Bob have a direct access to $Y$, a \textit{privacy-preserving mapping}/\textit{data release mechanism}\footnote{The terms "mapping", "mechanism", "scheme" and "algorithm" are used interchangeably in this paper.} is applied, whereby a distorted version of $Y$, denoted by $U$, is revealed to Bob. In this context, privacy and utility are competing goals: The more $Y$ is distorted by the privacy-preserving mapping, the less information can Bob infer about $X$, but also the less the utility that can be obtained. This trade-off is the very result of the dependencies between $X$ and $Y$. 

Stated formally in a general context, consider a triplet of random variables $(X,Y,W)$, distributed according to the given/known joint probability mass function (pmf) $p_{XYW}$. Let $X$ denote the \textit{private/sensitive data} that the user/curator wants to conceal, $Y$ denote the \textit{public/useful data} the user wishes to disclose, and $W$ denote the \textit{observable data} that the curator observes, which can be regarded as a noisy version of $(X,Y)$. Assume that the privacy-preserving mapping takes $W$ as input, and maps it to the \textit{released data}, denoted by $U$. In this scenario, $(X,Y)-W-U$ form a Markov chain, and the privacy-preserving mapping is captured by the conditional distribution $p_{U|W}$. We assume that all the alphabets/supports $\mathcal{X,Y,W}$ are finite sets.

As stated, for the triplet $(X,Y,W)$, a privacy-preserving mapping can be constructed by obtaining pmfs $p_{U|W}(\cdot|w)$ for each $w\in\mathcal{W}$ such that $p_{U|W}$ meets certain conditions corresponding to the utility/privacy requirements. Equivalently, this can be carried out by obtaining $\mathcal{U}$, $p_{W|U}(\cdot|u),\ \forall u\in\mathcal{U}$, such that $p_W$ is preserved in the joint distribution $p_{WU}$ and those conditions are met. We call the former approach \textit{forward construction/model} and the latter one \textit{backward construction/model}.

In this paper, we are solely interested in the special cases of \textit{full data observation} and \textit{public data observation} which refer to the settings in which the privacy-preserving mapping has direct access to both the private and public data (i.e., $W=(X,Y)$) and only to the public data (i.e., $W=Y$), respectively. 

By adopting mutual information as the measure of both \textit{utility} and \textit{privacy} (i.e., $I(Y;U)$, and $I(X;U)$, respectively), the optimal utility-privacy (U-P) trade-off in the public data observation model is defined as
\begin{equation}\label{def}
g_{\epsilon}(X,Y)\triangleq\max_{\substack{p_{U|Y}:\\X-Y-U\\I(X;U)\leq\epsilon}}I(Y;U),
\end{equation}
and in the full data observation model, the optimal U-P trade-off can be formulated as
\begin{equation}\label{Geps}
   G_\epsilon(X,Y)\triangleq \max_{\substack{p_{U|XY}:\\I(X;U)\leq\epsilon}} I(Y;U),
\end{equation}
where the effective range of $\epsilon$ is $[0,I(X;Y)]$ by noting that both (\ref{def}) and (\ref{Geps}) have the upper bound of $H(Y)$ which is attained by setting $U\triangleq Y$, which in turn results in $I(X;U)=I(X;Y)$.\footnote{ That the maximums in (\ref{def}) and (\ref{Geps}) exist follows from standard arguments in real analysis (compactness, continuity) as in \cite{RG-JSAIT}. Furthermore, when $(X,Y)\sim p_{XY}$, the quantities $g_\epsilon(X,Y)$ and $G_\epsilon(X,Y)$ are written interchangeably as $g_\epsilon(p_{XY})$ and $G_\epsilon(p_{XY})$, respectively.}

Perfect privacy \cite{RG-JSAIT} refers to the stringent constraint of $X\independent U$, i.e., $\epsilon=0$ in the U-P trade-off. Assume that we have an algorithm which satisfies this constraint. In other words, once applied to $W$ (with $W=(X,Y)$ or $Y$ depending on the model involved), this algorithm releases $U$ such that $I(X;U)=0$. Select an arbitrary conditional pmf $p_{Z|XY}$, and construct a private-public pair as $(Z,Y)\sim\sum_xp_{XY}(x,\cdot)p_{Z|XY}(\cdot|x,\cdot)$. Applying the algorithm in this new context (with $W=(Z,Y)$ or $Y$), we get $U'$ such that $I(Z;U')=0$. The utility obtained is a lower bound on the original optimal U-P trade-off at $\epsilon=I(X;U')$, and by changing $p_{Z|XY}$ and repeating this process, we can sweep the whole range of $\epsilon$, i.e., $[0,I(X;Y)]$. This simple observation is the basis of the privacy-preserving schemes proposed in this paper.

The information-theoretic view of privacy has gained increasing attention recently, with an incomplete list of related literature being \cite{Calmon1,Makhdoumi,SG19,info7010015,Issa,RG-JSAIT,Shkel,Zamani,RRG,Liao,Cuff,RGTIFS,Wang}. In \cite{Calmon1}, a general statistical inference framework is proposed to capture the loss of privacy in legitimate transactions of data. In \cite{Makhdoumi}, the privacy-utility trade-off under the \textit{log-loss} cost function is considered, called as the \textit{privacy funnel}, which is closely related to the \textit{information bottleneck} introduced in \cite{Tishby} (see also \cite{SG19}). In \cite{Calmon2}, sharp bounds on the optimal privacy-utility trade-off for the privacy funnel are derived, and an alternative characterization of the perfect privacy condition (also studied in \cite{Berger} in a different context) is proposed. Measuring both the privacy and the utility in terms of mutual information, perfect privacy is fully characterized in \cite{Shahab1} for the binary case. 

The current paper contributes to this context as follows.
\begin{itemize}
    \item Upper and lower bounds on $G_0(X,Y)$ are proposed and their tightness is investigated.
    \item Based on the aforementioned bounds, and the simplex method \cite{murty}, a closed-from solution for $G_0(X,Y)$ is derived when $X$ is binary, or $(|\mathcal{X}|,|\mathcal{Y}|)=(3,2)$. In this context, it is shown that for fixed conditional pmf $p_{Y|X}$, the optimal privacy-preserving mapping does not depend on $p_X$, which is of practical interest when the curator is unaware of the distribution of the private data.
    \item It is shown that for fixed $p_X$, $G_\epsilon(X,Y)$ is concave in $p_{Y|X}$, and for fixed $p_{Y|X}$, both $G_0(X,Y)$ and $g_0(X,Y)$ are convex in $p_X$.
    \item Based on the lower bound on $G_0(X,Y)$, a privacy-preserving scheme is presented as a lower bound on the optimal U-P trade-off in the case of full data observation.
    \item When $X$ is binary, an algorithmic lower bound on $g_0(X,Y)$ is presented, which is optimal when $|\mathcal{Y}|=3$.
    \item Based on this algorithm, a privacy-preserving scheme is presented as a lower bound on the optimal U-P trade-off in the case of public data observation.
\end{itemize}
% In \cite{Makhdoumi}, the privacy-utility trade-off under the \textit{log-loss} cost function is considered, called as the \textit{privacy funnel}, which is closely related to the \textit{information bottleneck} introduced in \cite{Tishby} (see also \cite{SG19}). In \cite{Calmon2}, sharp bounds on the optimal privacy-utility trade-off for the privacy funnel are derived, and an alternative characterization of the perfect privacy condition (also studied in \cite{Berger} in a different context) is proposed. Measuring both the privacy and the utility in terms of mutual information, perfect privacy is fully characterized in \cite{Shahab1} for the binary case. Furthermore, a new quantity is introduced to capture the amount of private information about the latent variable $X$ carried by the useful data $Y$.

% We study the information theoretic perfect privacy in this paper.

\textbf{Notations.} Random variables are denoted by capital letters ($X,Y$, etc.), their realizations by lower case letters ($x,y$, etc.), and their alphabets by capital letters in calligraphic font ($\mathcal{X},\mathcal{Y}$, etc.). 
Matrices,  and vectors are denoted by bold capital and bold lower case letters, respectively. 
The rank of matrix $\mathbf{A}$ is denoted by $\textnormal{rank}(\mathbf{A})$. 
For integers $m, n$, we have the discrete interval $[m:n]\triangleq\{m, m+1,\ldots,n\}$ if $m\leq n$, and $\emptyset$ (the empty set), otherwise. The set $[1:n]$ is written in short as $[n]$. Given two positive integers $a,b$, $a$ modulo $b$ is abbreviated as $a\textnormal{ mod }b$, and $\mathcal{S}\textnormal{ mod b}$ denotes $\{x\textnormal{ mod }b|\ \forall x \in \mathcal{S}\}$.
Given two pmfs $p,q$, the Kullback–Leibler divergence from $q$ to $p$ is defined as\footnote{We assume that $p$ is absolutely continuous with respect to $q$, i.e., $q(x)=0$ implies $p(x)=0$, otherwise, $D(p||q)\triangleq\infty$.} $D(p||q)\triangleq\sum_{x}p(x)\log_2\frac{p(x)}{q(x)}$. {\color{black}All the logarithms in this paper are to the base of 2.} For $0\leq t\leq 1$, $\Bar{t}\triangleq 1-t$ and $H_b(t)\triangleq-t\log t-\Bar{t}\log \Bar{t}$ denotes the binary entropy function (with the convention $0\log 0\triangleq 0$). For an event $\mathcal{E}$, the indicator $\mathds{1}_{\{\mathcal{E}\}}$ is one when $\mathcal{E}$ occurs, and zero, otherwise. The domain of function $f$ is denoted by $\textnormal{dom}(f)$, and throughout the paper, if there are more than one candidate for $\argmin_{x\in\textnormal{dom}(f)}f(x)$, one is selected arbitrarily.
% Given two positive integers $a,b$, we have $a$ modulo $b$ written in short as $a\textnormal{ mod }b$. We have that $d_{\textnormal{TV}}$, $\lfloor\cdot\rfloor$, and $\lceil\cdot\rceil$ denote the total variation distance, the floor, and the ceiling operators, respectively. 
For a real number $x$, we define $(x)^+\triangleq\max\{0,x\}$. Define the support of a given pair $(X,Y)\sim p_{XY}$ as $ \textnormal{supp}(X,Y)\triangleq\{(x,y)\in\mathcal{X}\times\mathcal{Y}|p_{XY}(x,y)>0\}.$
Finally, throughout this paper, we encounter summations of the form $\sum_{u\in\mathcal{U}}(\cdot)$, i.e., summation over the elements of $\mathcal{U}$. If $\mathcal{U}$ happens to be the empty set, this summation is defined as zero. 
% \footnote{Due to space constraints, some of the proofs are provided in the extended online version \cite{rassouli2017perfect}. {\color{red}REMOVE THIS FOR THE ARXIV VERSION.}}

\section{Preliminaries}

% In this context, we assume that $p_X(x),p_Y(y),p_W(w)>0,\forall (x,y,w)\in\mathcal{X}\times\mathcal{Y}\times\mathcal{W}$, since otherwise the alphabets could have been modified accordingly. 
% This equivalently means that the corresponding probability vectors $\mathbf{p}_X,\mathbf{p}_Y, \mathbf{p}_W$ are in the interior of their corresponding probability simplices, i.e., $\mathcal{P}(\mathcal{X}),\mathcal{P}(\mathcal{Y}),\mathcal{P}(\mathcal{W})$, respectively.

% The following proposition states the necessary and sufficient condition for the feasibility of perfect privacy.
% \begin{proposition}
% Perfect privacy is feasible for $(X,Y,W)\in\mathcal{X}\times\mathcal{Y}\times\mathcal{W}$ if and only if
% \begin{equation}\label{Fullshart}
% \textnormal{dim}\bigg(\textnormal{Null}(\mathbf{P}_{X|W})\backslash\textnormal{Null}(\mathbf{P}_{Y|W})\bigg)\neq 0.
% \end{equation}
% \end{proposition}
% \begin{proof}
% The proof is a simple generalization of \cite[Theorem 4]{Berger}, by noting that both $X-W-U$ and $Y-W-U$ form Markov chains. In other words, we have $X\independent U$ and $Y\not\independent U$ if and only if there exists a vector $\mathbf{v}$ in $\mathcal{P}(\mathcal{W})$, such that a change in $\mathbf{p}_W$ along $\mathbf{v}-\mathbf{p}_W$ changes $\mathbf{p}_Y$ (in line with $Y\not\independent U$), while keeps $\mathbf{p}_X$ unchanged, i.e., $X\independent U$. Equivalently, there exists a vector $\mathbf{v}'\in\textnormal{Null}(\mathbf{P}_{X|W})$ such that $\mathbf{v}'\not\in\textnormal{Null}(\mathbf{P}_{Y|W})$, which is equivalent to (\ref{Fullshart}).
% \end{proof}

Throughout the paper, the U-P plane refers to the 2-dimensional plane, in which the horizontal and vertical axes denote the privacy-leakage $I(X;U)$ and utility $I(Y;U)$, respectively. Furthermore, we say that a point $P=(P_x,P_y)$ is achievable on this plane if there exists a joint distribution $p_{XY}\cdot p_{U|XY}$ (in the case of full data observation), or $p_{XY}\cdot p_{U|Y}$ (in the case of public data observation), such that $I(X;U)=P_x$, and $I(Y;U)=P_y$. 

\begin{remark}\label{concavity}
When $\epsilon\in[0,I(X,Y)]$, both $g_\epsilon(X,Y)$ and $G_\epsilon(X,Y)$ are concave and strictly increasing functions of $\epsilon$ that lie above the lines connecting $(0,g_0(X,Y))$, and $(0,G_0(X,Y))$ to $(I(X;Y),H(Y))$, respectively. Furthermore, they both lie under the line connecting $(0,H(Y|X))$ to $(I(X;Y),H(Y))$. Also, for an optimal mapping in (\ref{def}) or (\ref{Geps}), we have $I(X;U)=\epsilon$. Finally, both $g_0(X,Y)$ and $G_0(X,Y)$ can be obtained via a linear program (LP).
\end{remark}
\begin{proof}
The concavity of $g_\epsilon(X,Y)$ in $\epsilon$ is shown in \cite[lemma 2]{info7010015}. That of $G_\epsilon(X,Y)$ follows similarly\footnote{It can be shown via example that the claim of strict concavity is too strong for these two curves.}. We also have
\begin{align}
   g_\epsilon(X,Y)&\leq G_\epsilon(X,Y)\label{byd}\\
   &\leq \max_{\substack{p_{U|X,Y}:\\I(X;U)\leq\epsilon}} I(X,Y;U)\nonumber\\
   &=\max_{\substack{p_{U|X,Y}:\\I(X;U)\leq\epsilon}}I(X;U)+I(Y;U|X)\nonumber\\
   &\leq \epsilon + H(Y|X)\label{ful1},\ \epsilon\in[0,I(X;Y)],
\end{align}
where (\ref{byd}) is by definition. This means that both $g_\epsilon(X,Y)$ and $G_\epsilon(X,Y)$ lie under the line connecting $(0,H(Y|X))$ to $(I(X;Y),H(Y))$. Unless in the degenerate case of $X\independent Y$ (which results in $\epsilon\in[0,I(X;Y)]=\{0\}$), we have $H(Y|X)<H(Y)$, and hence, both $g_\epsilon(X,Y)$ and $G_\epsilon(X,Y)$, being concave functions, must be strictly increasing in $\epsilon$. As a result, they lie above the lines connecting $(0,g_0(X,Y))$, and $(0,G_0(X,Y))$ to $(I(X;Y),H(Y))$, respectively.

By noting that in maximizing a convex functional, the maximum occurs at an extreme point, we conclude that for any maximizer in (\ref{def}) or (\ref{Geps}), we have $I(X;U)=\epsilon$, when $\epsilon\in[0,I(X;Y)]$.

As shown in \cite{RG-JSAIT}, $g_0(X,Y)$ is obtained via an LP. More specifically, in the Markov chain $X-Y-U$, if we consider the backward model, the conditional pmf $p_{Y|U}(\cdot|u)$ must belong to a convex polytope determined by the condition $X\independent U$, i.e., $p_X(\cdot)=\sum_yp_{X|Y}(\cdot|y)p_{Y|U}(y|u)$. Since $H(Y|U=u)$ is a concave functional of $p_{Y|U}(\cdot|u)$, and its minimum occurs at an extreme point, we first obtain the extreme points of the aforementioned convex polytope, and what remains is to allocate proper weights, i.e., $p_U(\cdot)$, such that $H(Y|U)$ is minimized subject to $p_Y(\cdot)=\sum_u p_{Y|U}(\cdot|u)p_U(u)$, which is an LP. Similarly, by considering the Markov chain $X-(X,Y)-U$, $G_0(X,Y)$ can be obtained through an LP.
\end{proof}
% For a positive integer $n$ and a set of points $\mathcal{P}\triangleq\{(x_i,y_i)\}_{i=1}^n$ in $\mathds{R}^2$, define
% \begin{equation*}
%     \mathcal{F}_{[\mathcal{P}]}\triangleq\{f(x)|\textnormal{ dom}(f)=[\min_ix_i,\max_ix_i], f\textnormal{ is concave },f(x_i)\geq y_i,\ \forall i\in[n]\}
% \end{equation*}
% as the set of all concave functions $f(x)$ with domain $[\min_ix_i,\max_ix_i]$, for which $f(x_i)\geq y_i,\ \forall i\in[n]$.
\begin{definition}\label{UCE} The upper concave envelope of a set of points $\mathcal{P}\triangleq\{(x_i,y_i)\}_{i=1}^n$ in $\mathds{R}^2$ is defined as\footnote{The \textit{lower convex envelope} of $\mathcal{P}$ is the negative of the upper concave envelope of $\mathcal{P}^-\triangleq\{(x_i,y_i)|(x_i,-y_i)\in\mathcal{P},\ \forall i\in[n]\}$.}
\begin{equation*}
    uce_{[\mathcal{P}]}(\cdot)\triangleq \inf \{f(x)|\textnormal{ dom}(f)=[\min_ix_i,\max_ix_i], f\textnormal{ is concave },f(x_i)\geq y_i,\ \forall i\in[n]\}.
\end{equation*}
\end{definition}
% and the non-decreasing upper concave envelope of $\mathcal{P}$ is
% \begin{equation*}
%     uce^*_{[\mathcal{P}]}(\cdot)\triangleq \inf \left\{f(x)\in\mathcal{F}_{[\mathcal{P}]}\ |\ f(x)\leq f(y),\ \forall x,y\in \textnormal{dom}(f)\right\}.
% \end{equation*}
% \end{definition}
% Let $(x^*,y^*)\triangleq\argmax_{(x,y)}\{y|(x,y)\in\mathcal{P}\}$, and if there are more than one maximizer, select one arbitrarily. It can be readily verified that $uce^*_{[\mathcal{P}]}(x)$ coincides $uce_{[\mathcal{P}]}(x)$ for $x\in[\min_ix_i,x^*]$, and remains constant at $y^*$ when $x\in[x^*,\max_ix_i]$. As a special case, if the $y$ coordinate corresponding to $\max_ix_i$ is maximum, we have $uce^*(\cdot)=uce(\cdot).$

Figure \ref{fig_uce} provides an example of the upper concave envelope (solid line) of a set of points (filled circles).
\begin{figure}
    \centering
    \includegraphics{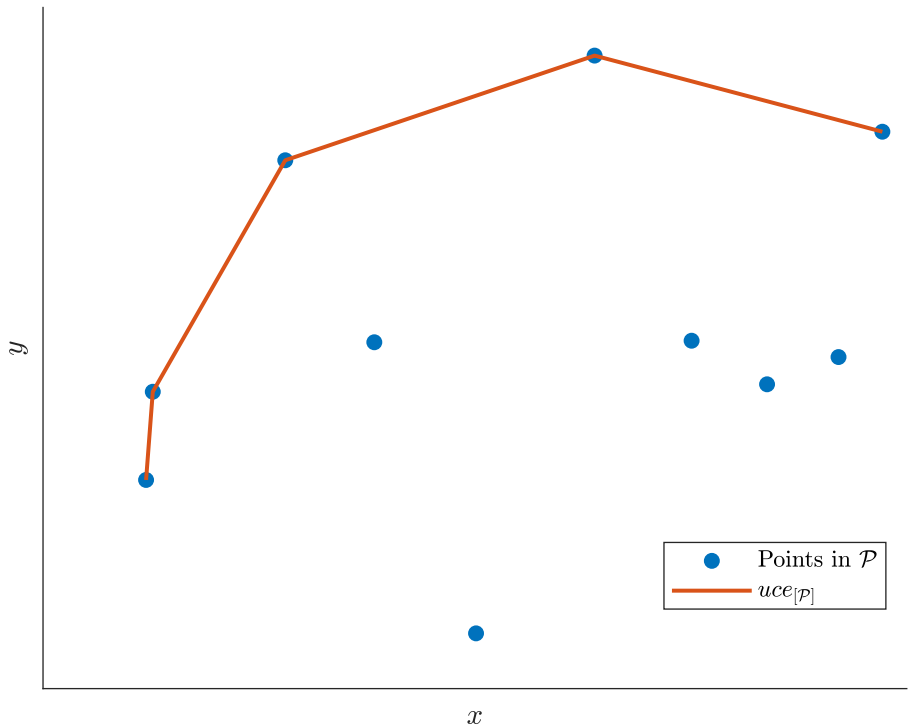}
    \caption{The upper concave envelope of a set of points.}
    \label{fig_uce}
\end{figure}
\begin{remark}
From Remark \ref{concavity} and Definition \ref{UCE}, it is obvious that if the points in $\mathcal{P}$ are achievable on the U-P plane, 
$\{(x,uce_{[\mathcal{P}]}(x))\ |\ \forall x\in[\min_ix_i,\max_ix_i]\}$
is also achievable, and hence, $uce_{[\mathcal{P}]}(\cdot)$ serves as a lower bound on the optimal utility-privacy trade-off\footnote{Since $(I(X;Y),H(Y))$ is an achievable point, we can include it in $\mathcal{P}$. Therefore, the resulting $uce_{[\mathcal{P}]}(\cdot)$ will be non-decreasing.}.   
\end{remark}

The privacy-preserving algorithms in this paper find achievable points based on a generalization of statistical independence. The sole purpose of the following definition is to simplify the explanation of this generalization.
\begin{definition} For a pair of random variables  $(X,U)\sim p_{XU}$, and $k\in[|\mathcal{X}|-1]$, if
\begin{equation}
    \bigg|\bigg\{x\in\mathcal{X}\ \bigg|\ p_{X|U}(x|u)=p_X(x),\ \forall u\in\mathcal{U}\bigg\}\bigg|\geq k,
\end{equation}
or equivalently, 
\begin{equation}\label{fed2}
    \bigg|\bigg\{x\in\mathcal{X}\ \bigg|\ p_{U|X}(\cdot|x)=p_U(\cdot)\bigg\}\bigg|\geq k,
\end{equation}
we say that $X$ is \textit{at least $k$-independent of $U$}, which is denoted by $X\stackrel{k}\independent U$.
\end{definition}
Note that i) $X\stackrel{k}\independent U\Longrightarrow X\stackrel{k-1}\independent U,\ k\in[2:|\mathcal{X}|-1]$, ii) $X\stackrel{|\mathcal{X}|-1}\independent U\Longrightarrow X\independent U$, and iii) this is an asymmetric notion, i.e., $X\stackrel{k}\independent U\not\Longrightarrow U\stackrel{k}\independent X$. Furthermore, if $X-U-Z$ form a Markov chain, from (\ref{fed2}), we conclude that $X\stackrel{k}\independent U$ results in $X\stackrel{k}\independent Z$.

% {\color{red} Two approaches: Merging ($I(Y;U)\to 0$), and $k$-independence ($I(X;U)\to 0$). The latter can also be applied to $Y$ to result in $I(Y;U)\to 0$.}
\section{Full data observation model}
In this Section, we assume that the curator has access to both $X$ and $Y$, i.e., $W=(X,Y)$, and propose an achievable scheme, i.e., a lower bound on $G_\epsilon(X,Y)$, which is defined in (\ref{Geps}). To this end, we find achievable points on the U-P trade-off and propose their upper concave envelope as the lower bound.

% While Remark \ref{concavity} states the concavity of $G_\epsilon(X,Y)$ (and also $g_\epsilon(X,Y)$) as a function of $\epsilon$, the following Proposition investigates its concavity as a functional of the joint distribution $p_{XY}$.

The following Lemma is central to the analysis in this Section.
\begin{lemma}\label{lemcardi}
For an optimal $p_{U^*|XY}$ in the evaluation of $G_0(X,Y)$, we must have
\begin{equation}\label{chahar}
    |\{y\in\mathcal{Y}|p(x,y|u^*)>0\}|= 1,\ \forall (x,u^*) \in\mathcal{X}\times\mathcal{U}^*,
\end{equation}
which results in
\begin{equation}\label{entfunction}
    H(Y|X,U^*)=0,
\end{equation}
where $U^*\in\mathcal{U}^*$ is induced by the optimal mapping $p_{U^*|XY}$. In other words, for any $(x,u^*)\in\mathcal{X}\times\mathcal{U}^*$, there must exist $y_{x,u^*}\in\mathcal{Y}$ such that
\begin{equation}\label{sheshomi}
    p(x,y_{x,u^*}|u^*)=p(x),\ p(x,y|u^*)=0,\ \forall y\in\mathcal{Y}\backslash\{y_{x,u^*}\},
\end{equation}
which results in a lower bound on the cardinality of $|\mathcal{U}^*|$ as
\begin{equation*}
    |\mathcal{U}^*|\geq\max_{x\in\mathcal{X}}|\{y\in\mathcal{Y}|p(y|x)>0\}|.
\end{equation*}
\end{lemma}
\begin{proof}
    The proof is provided in Appendix \ref{app22}.\footnote{This Lemma is also given in \cite[Lemma 5]{Shkel}. Since it was also independently provided in \cite[Lemma 3]{arxv4} by the authors of the current manuscript, it is mentioned here.}
\end{proof}
Lemma \ref{lemcardi} is exemplified in Figure \ref{fig1} in which $(X,Y)\in\{x_1,x_2,x_3\}\times\{y_1,y_2\}$. Let $p_i\triangleq p(x_i),i\in[3]$. As (\ref{sheshomi}) requires, for each realization of $U$, there is exactly one link to subgroup $i$ of nodes, i.e., $\{(x_i,y_j)\}_{j=1}^2$ with transition probability $p_i$.

\begin{figure}
\centering
\begin{tikzpicture}[scale=0.8, transform shape]
  \node[dspnodeopen, minimum width=4pt,  dsp/label=left] (X_1) {$x_1y_1$};    
  \node[dspnodeopen, minimum width=4pt, below=1cm of X_1, dsp/label=left] (X_2) {$x_1y_2$};
  \node[dspnodeopen, minimum width=4pt, below=2cm of X_2, dsp/label=left] (X_3) {$x_2y_1$};
  \node[dspnodeopen, minimum width=4pt, below=1cm of X_3, dsp/label=left] (X_4) {$x_2y_2$};
  \node[dspnodeopen, minimum width=4pt, below=2cm of X_4, dsp/label=left] (X_5) {$x_3y_1$};
  \node[dspnodeopen, minimum width=4pt, below=1cm of X_5, dsp/label=left] (X_6) {$x_3y_2$};

    \node[dspnodeopen, minimum width=4pt, right=3cm of X_1, dsp/label=right] (U_1) {$u_1$};    
  \node[dspnodeopen, minimum width=4pt, below=1cm of U_1, dsp/label=right] (U_2) {$u_2$};
  \node[dspnodeopen, minimum width=4pt, below=1cm of U_2, dsp/label=right] (U_3) {$u_3$};
   \node[dspnodeopen, minimum width=4pt, below=1cm of U_3, dsp/label=right] (U_4) {$u_4$};
   \node[dspnodeopen, minimum width=4pt, below=1cm of U_4, dsp/label=right] (U_5) {$u_5$};
   \node[dspnodeopen, minimum width=4pt, below=1cm of U_5, dsp/label=right] (U_6) {$u_6$};
   \node[dspnodeopen, minimum width=4pt, below=1cm of U_6, dsp/label=right] (U_7) {$u_7$};
   \node[dspnodeopen, minimum width=4pt, below=1cm of U_7, dsp/label=right] (U_8) {$u_8$};
    
  \draw[line width=1.2pt] (U_1) to node[near end, inner sep=1pt, above,sloped] {}(X_1);
  \draw[line width=1.2pt] (U_1) to node[near end, inner sep=1pt, above,sloped] {}(X_3);
  \draw[line width=1.2pt] (U_1) to node[near end, inner sep=1pt, above,sloped] {}(X_5);
  \draw[line width=1.2pt] (U_2) to node[near end, inner sep=1pt, above,sloped] {}(X_1);
  \draw[line width=1.2pt] (U_2) to node[near end, inner sep=1pt, above,sloped] {}(X_3);
  \draw[line width=1.2pt] (U_2) to node[near end, inner sep=1pt, above,sloped] {}(X_6);
  \draw[line width=1.2pt] (U_3) to node[near end, inner sep=1pt, above,sloped] {}(X_1);
  \draw[line width=1.2pt] (U_3) to node[near end, inner sep=1pt, above,sloped] {}(X_4);
  \draw[line width=1.2pt] (U_3) to node[near end, inner sep=1pt, above,sloped] {}(X_5);
  \draw[line width=1.2pt] (U_4) to node[near end, inner sep=1pt, above,sloped] {}(X_1);
  \draw[line width=1.2pt] (U_4) to node[near end, inner sep=1pt, above,sloped] {}(X_4);
  \draw[line width=1.2pt] (U_4) to node[near end, inner sep=1pt, above,sloped] {}(X_6);
  \draw[line width=1.2pt] (U_5) to node[near end, inner sep=1pt, above,sloped] {}(X_2);
  \draw[line width=1.2pt] (U_5) to node[near end, inner sep=1pt, above,sloped] {}(X_3);
  \draw[line width=1.2pt] (U_5) to node[near end, inner sep=1pt, above,sloped] {}(X_5);
  \draw[line width=1.2pt] (U_6) to node[near end, inner sep=1pt, above,sloped] {}(X_2);
  \draw[line width=1.2pt] (U_6) to node[near end, inner sep=1pt, above,sloped] {}(X_3);
  \draw[line width=1.2pt] (U_6) to node[near end, inner sep=1pt, above,sloped] {}(X_6);
  \draw[line width=1.2pt] (U_7) to node[near end, inner sep=1pt, above,sloped] {}(X_2);
  \draw[line width=1.2pt] (U_7) to node[near end, inner sep=1pt, above,sloped] {}(X_4);
  \draw[line width=1.2pt] (U_7) to node[near end, inner sep=1pt, above,sloped] {}(X_5);
  \draw[line width=1.2pt] (U_8) to node[near end, inner sep=1pt, above,sloped] {}(X_2);
  \draw[line width=1.2pt] (U_8) to node[near end, inner sep=1pt, above,sloped] {}(X_4);
  \draw[line width=1.2pt] (U_8) to node[near end, inner sep=1pt, above,sloped] {}(X_6);

\end{tikzpicture}
\caption{An illustrative representation of Lemma \ref{lemcardi} for $(X,Y)\in\{x_1,x_2,x_3\}\times\{y_1,y_2\}.$ If $u_k$ is connected to pair $(x_i,y_j)$, we have $p(x_i,y_j|u_k)=p(x_i),\ (i,j,k)\in[3]\times[2]\times[8]$.}
\label{fig1}
\end{figure}
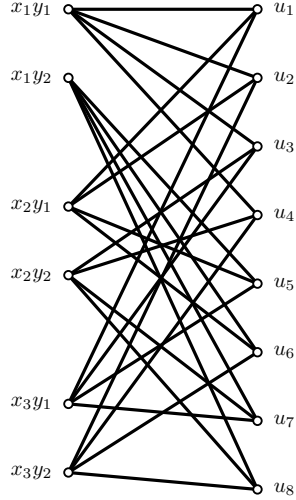
\begin{remark} 
If $Y$ is not a function of $X$, we have $G_0(X,Y)>0$ by \cite[Theorem 4]{RG-JSAIT} which is resulted by a $U$ such that $X\independent U$ and $H(Y|X,U)=0$ according to Lemma \ref{lemcardi}. Obviously, if $Y$ is a function of $X$, we can select $U$ as an arbitrary singleton. This observation provides an alternative proof for the functional representation lemma \cite[p. 626]{Elgamal}.
    
\end{remark}
In order to propose a lower bound on $G_\epsilon(X,Y)$, we start with $\epsilon = 0$, which is equivalent to $X\independent U$. It is already known that $G_0(X,Y)$ can be obtained via an LP whose dimension is the total number of extreme points of the convex polytope stated earlier in Remark \ref{concavity}. However, according to Lemma \ref{lemcardi}, these extreme points are already known. They are all the conditional pmfs $p_{XY|U}$ that satisfy the property in (\ref{sheshomi}). As a result, the dimension of the LP involved in evaluating $G_0(X,Y)$ is $\prod_{x\in\mathcal{X}}|\{y\in\mathcal{Y}|p(y|x)>0\}|\leq |\mathcal{Y}|^{|\mathcal{X}|}$. Note that even assuming a polynomial time complexity for the algorithm used for solving the LP, unless $|\mathcal{X}|$ and $|\mathcal{Y}|$ are small or the matrix of joint distribution $\mathbf{P}_{XY}$ is sparse, the problem becomes computationally intractable in terms of time and space. Therefore, a tractable method is desirable.

In what follows, an algorithm (Algorithm 1) is proposed in Theorem \ref{thbound5} that provides a lower bound on $G_0(X,Y)$. This algorithm is proved to be optimal in Theorem \ref{th2} when $X$ is binary or $(|\mathcal{X}|,|\mathcal{Y}|)=(3,2)$. Finally, building upon Algorithm 1, Proposition \ref{Prop1} presents Algorithm 2 which produces a privacy-preserving mapping as a lower bound on $G_\epsilon(X,Y),\ \epsilon\in[0,I(X;Y)]$.

% Lemma \ref{lemcardi} implies that in the optimal setting, for any $u^*\in\mathcal{U}^*$, and any $x\in\mathcal{X}$, there exists $y_{x,u^*}\in\mathcal{Y}$, such that $p(x,y_{x,u^*}|u^*)=p(x)$, and $p(x,y|u^*)=0,\ \forall y\in\mathcal{Y}\backslash\{y_{x,u^*}\}$. Therefore, in the achievable scheme, it makes sense to build a mapping $p_{U|X,Y}$, such its corresponding $p_{X,Y|U}$ satisfies this condition. To this end, we start with the backward model, i.e., $p_{X,Y|U}$, by imposing that i) for all the realizations $u$ of $U$, the condition in lemma \ref{lemcardi} must be satisfied, ii) the joint pmf $p_{X,Y}$ must be preserved in $p_{X,Y,U}$. The results are provided in the following Theorem.  
\begin{theorem}\label{thbound5}
For a given pair $(X,Y)\sim p_{XY}$, we have
\begin{align}
   G_0(X,Y)&\geq \left(H(Y)-\left(1-\sum_y\min_x p(y|x)\right)\min\{H(X),\log|\mathcal{Y}|\}\right)^+.\label{G0lowerbound}
\end{align}
% {\color{black}and $\mathcal{L}^c(X\to Y)$ is the maximal cost leakage defined in \cite{Issa} as
% \begin{equation*}
%     \mathcal{L}^c(X\to Y)\triangleq-\log\sum_y\min_x p(y|x).
% \end{equation*}
% }
\end{theorem}
\begin{proof}
{\color{black} 
% A mapping $p_{U|{X,Y}}$ is built in an iterative way, such that $p_{X|U}(\cdot|u)=p_{X}(\cdot),\ \forall u\in\mathcal{U}$ (and hence, $X\independent U$), and $H(Y|U)\leq \left(1-\sum_y\min_x p(y|x)\right)\cdot H(X)$ as follows. 
Define the index set
\begin{equation}
    \mathcal{I}\triangleq\{y\in\mathcal{Y}|p(y|x)>0,\ \forall x\in\mathcal{X}\},
\end{equation}
and relabel the elements of $\mathcal{Y}=\{y_1,y_2,\ldots,y_{|\mathcal{Y}|}\}$ such that the first $|\mathcal{I}|$ elements belong to $\mathcal{I}$. The algorithm proceeds as follows. First, $|\mathcal{I}|$ mass points for $U$ are created, which are denoted by $u_j (j\in[|\mathcal{I}|])$, each having $p_U(u_j)=\min_x p_{Y|X}(y_j|x)$, respectively, such that for all $x\in\mathcal{X}$, we have $p(x,y_k|u_j)=p(x)$, if $j=k$, and $0$, otherwise ($j,k\in[|\mathcal{I}|]$). It is evident that thus far, the posterior $p_X(\cdot|u_j)$ remains the same as the prior $p_X(\cdot)$, which is in line with the condition of $X\independent U$. Moreover, these mass points are such that $p_{Y|U}(y_j|u_j)=1$, resulting in $H(Y|U=u_j)=0,\ j\in[|\mathcal{Y}|]$. Afterwards, the iterations begin. In each iteration $i$, a mass point $u_{i+|\mathcal{I}|}$ is created such that $p_X(\cdot|u_{i+|\mathcal{I}|})=p_X(\cdot)$, and the conditional pmf of the pair $(X,Y)$ conditioned on $\{U=u_{i+|\mathcal{I}|}\}$ has the same mass probabilities as in $p_X(\cdot)$ resulting in $H(X,Y|U=u_{i+|\mathcal{I}|})=H(X)$. Hence, $H(Y|U=u_{i+|\mathcal{I}|})\leq H(X),\ i\geq 1$ (note that we also have the trivial upper bound $H(Y|U=u_{i+|\mathcal{I}|})\leq\log|\mathcal{Y}|$). The procedure is provided in Algorithm 1, 
% {\color{red}where in Step 7, if there are multiple minimizers, one is selected arbitrarily\footnote{Although it can be shown that the algorithm can be improved by a careful selection of minimizers, this is not needed in this discussion. In any case, note that Step 7 of the algorithm needs to be consistent with Step 6.}. 
and the algorithm terminates at some iteration $N$, where %$a_N(x,y)=0,\ \forall (x,y)\in\mathcal{X}\times\mathcal{Y}$,
$N\leq |\textnormal{supp}(X,Y)|-|\mathcal{I}|$, which is further tightened in Remark \ref{rem44}.}
\begin{algorithm}\label{algor}
\caption{A lower bound on $G_0(X,Y).$}
\begin{algorithmic}[1]
{\color{black}\Function{Algorithm1}{$p_{Y|X}$}
\State $p(u_j|x,y)=\frac{\min_{x\in\mathcal{X}} p(y|x)}{p(y|x)}\cdot\mathds{1}_{\{y=y_j\}},\ \forall j\in[|\mathcal{I}|], (x,y)\in\textnormal{supp}(X,Y)$
\State $a_1(x,y)= p(y|x)-\min_{x}p(y|x),\ \forall (x,y)\in\textnormal{supp}(X,Y)$
\State i = 1
\While{$\max_{x,y}a_i(x,y)\neq 0$}
\State $a_i^*= \min_{x,y}\{a_i(x,y)|a_i(x,y)>0\}$
\State $(x^*,y^*)=\argmin_{x,y}\{a_i(x,y)|a_i(x,y)>0\}$
\State $f_i(x) = \argmin_{y}\{a_i(x,y)|a_i(x,y)>0\},\ \forall x\in\mathcal{X}$
\State $p(u_{i+|\mathcal{I}|}|x,y)=\frac{a_i^*}{p(y|x)}\cdot\left(\mathds{1}_{\{a_i(x,y^*)>0, y=y^*\}}+\mathds{1}_{\{a_i(x,y^*)=0, y=f_i(x)\}}\right),\ \forall (x,y)\in\textnormal{supp}(X,Y)$
\State $a_{i+1}(x,y)= a_i(x,y)-a_i^*\cdot\left(\mathds{1}_{\{a_i(x,y^*)>0, y=y^*\}}+\mathds{1}_{\{a_i(x,y^*)=0, y=f_i(x)\}}\right),\ \forall (x,y)\in\textnormal{supp}(X,Y)$
\State $i = i + 1$
\EndWhile
\State \Return $p_{U|X,Y}$
\EndFunction}
\end{algorithmic}
\end{algorithm}

% Define $a_1(x,y)\triangleq p_{Y|X}(y|x),\ \forall (x,y)\in\mathcal{X}\times\mathcal{Y}$. For each iteration $i\geq 1$, we have
% \begin{align}
%     a_i^*(x)&\triangleq \min_{y\in\mathcal{Y}}\{a_i(x,y)|a_i(x,y)>0\},\ \forall x\in\mathcal{X}\\
%     f_i(x) &\triangleq \argmin_{y\in\mathcal{Y}}\{a_i(x,y)|a_i(x,y)>0\},\ \forall x\in\mathcal{X}\\
%     a_i^*&\triangleq \min_{x\in\mathcal{X}}a_i^*(x)\\
%     %x_i^*&\triangleq \argmin_{x\in\mathcal{X}}a_i^*(x)\\
%     p_{U|X,Y}(u_i|x,y)&\triangleq\frac{a_i^*}{p_{Y|X}(y|x)}\cdot\mathds{1}_{\{y=f_i(x)\}},\ \forall (x,y)\in\mathcal{X}\times\mathcal{Y}\label{condp}\\
%     a_{i+1}(x,y)&\triangleq a_i(x,y)-p_U(u_i)\cdot\mathds{1}_{\{y=f_i(x)\}},
% \end{align}

{\color{black}
The rationale behind this algorithm becomes clear by considering the backward construction as follows. Let each realization $(x,y)$ of $(X,Y)$ be denoted by a node. Arrange these nodes in a long column vector as in Figure \ref{fig1}. In this configuration, we divide the nodes into $|\mathcal{X}|$ subgroups of nodes: The first subgroup of nodes is $\{(x_1,y_i)|i\in[|\mathcal{Y}|]\}$, the second subgroup is $\{(x_2,y_i)|i\in[|\mathcal{Y}|]\}$, and so on. Obviously, the sum of the mass probabilities of the nodes in the $i$-th subgroup is $p_X(x_i),\ \forall i\in[|\mathcal{X}|]$. Therefore, if in this construction, from each mass point (or node) $u$, there is one connection/link to only one of the nodes in the first subgroup with transition probability $p_X(x_1)$, one link to only one of the nodes in the second subgroup with transition probability $p_X(x_2)$, and so on, which is what Lemma \ref{lemcardi} implies, we have $p_{X|U}(\cdot|u)=p_X(\cdot)$, and $H(Y|U=u)\leq H(X,Y|U=u)=H(X)$. However, if in the first $|\mathcal{I}|$ realizations of $U$, the links that connect each $u\in\{u_1,\ldots,u_{|\mathcal{I}|}\}$ to $|\mathcal{X}|$ subgroups arrive at nodes that have the same second coordinate, i.e., $y$, we get $H(Y|U=u)=0$ for these $|\mathcal{I}|$ realizations of $U$. Also, if a node $u$ is to be connected to $|\mathcal{X}|$ nodes sharing the same second coordinate, e.g., $\{(x_i,y)\}_{i=1}^{|\mathcal{X}|}$ for some $y\in\mathcal{Y}$, we must have $p(u)\leq p(y|x_i),\ \forall i\in[|\mathcal{X}|]$. This is needed to guarantee that the requirement $p(x_i,y|u)=p(x_i), i\in[|\mathcal{X}|]$ does not violate the preservation of $p_{XY}$ in $p_{XYU}$.\footnote{Since otherwise, we have $p(u)> p(y|x_j)$, for some $j\in[|\mathcal{X}|]$. This results in $p_{XYU}(x_j,y,u)=p(u)p(x_j,y|u)=p(u)p(x_j)>p(x_j,y)$, which results in $p(x_j,y)$ induced by $p_{XYU}$ being greater than the original $p_{XY}(x_j,y)$.} Therefore, we set $p(u)$ equal to its maximum allowable value, i.e., $\min_xp(y|x)$. The aforementioned procedure is captured in step 2 of the algorithm by making the convention $\frac{0}{0}\cdot 0\triangleq 0$. Subsequently, the event containing the first $|\mathcal{I}|$ realizations of $U$ occurs with probability of $\sum_{y\in\mathcal{I}}\min_xp(y|x)=\sum_{y\in\mathcal{Y}}\min_xp(y|x)$, which results in $H(Y|U)\leq (1-\sum_y\min_xp(y|x))\min\{H(X),\log|\mathcal{Y}|\}$. 

The concern in this backward construction is to preserve the original distribution $p_{XY}$ in the resulting joint pmf $p_{XYU}$. Since in the construction of $U$, it is known from our impositions that if $p(x,y|u)\neq 0$, for some $(x,y)\in\textnormal{supp}(X,Y)$, then we must have $p(x,y|u) = p(x)$, we observe that the preservation of $p_{XY}$ boils down to that of the conditional pmf $p_{Y|X}$. In other words, denoting the set of all realizations $u$ that have a link to $(x,y)$ by $\mathcal{U}_{x,y}$, i.e., $\mathcal{U}_{x,y}\triangleq\{u\in\mathcal{U}|p(x,y|u)\neq 0\},\ \forall (x,y)\in\textnormal{supp}(X,Y),$
the preservation of $p_{XY}$ is equivalent to 
\begin{align}
p(x,y)&=\sum_{u\in\mathcal{U}}p(x,y,u)\nonumber\\
&=\sum_{u\in\mathcal{U}_{x,y}}p(x,y|u)p(u)\nonumber\\
&=p(x)\sum_{u\in\mathcal{U}_{x,y}}p(u),\ \forall (x,y)\in\textnormal{supp}(X,Y),\nonumber
\end{align}
which is in turn equivalent to
\begin{equation*}
    p(y|x)=\sum_{u\in\mathcal{U}_{x,y}}p(u),\ \forall (x,y)\in\textnormal{supp}(X,Y).
\end{equation*}
Therefore, we only need to make sure that the mass probabilities of all the nodes $u$ that are connected to $(x,y)$ sum up to $p(y|x)$. To this end, we harness a waterfilling-like procedure, in which the water levels denote the remaining probabilities which need to be "filled". In step 3, the water levels are set as $a_1(x,y)$ by taking into account the assignment in step 2. In other words, for each node $(x,y)$, the amount of $\min_x p(y|x)$ has already been filled by the links from $u_j,\ \forall j\in[|\mathcal{I}|]$ in step 2. At each iteration $i$, node $u_{i+|\mathcal{I}|}$ is created to fill the minimum water level denoted by $a_i^*$ in step 6. A/the minimizer is denoted by $(x^*,y^*)$ in step 7. Note that in this step and step 8, if there are multiple minimizers, one is selected arbitrarily\footnote{Although at the expense of making the algorithm more complicated, one can propose a better selection (in terms of lowering $H(Y|U)$), we do not discuss it here.}. In step 8, $f_i(x)$ denotes a/the minimum non-zero water level in subgroup $x$ at iteration $i$. We create $u_{i+|\mathcal{I}|}$, and set $p(u_{i+|\mathcal{I}|})\triangleq a_i^*$, and connect this node to $|\mathcal{X}|$ nodes, each belonging to one subgroup, with the transition probability of $p(x_1)$ for the link to subgroup 1, $p(x_2)$ for the link to subgroup 2, and so on.
In doing so, we take this intuition into account that points with common $y$-coordinates are desirable, as this allocation is in line with lowering $H(Y|U)$. Hence, in each subgroup $x$ ($x\in\mathcal{X}$), if the water level corresponding to $(x,y^*)$, i.e., $a_i(x,y^*)$, is non-zero, this point is selected, otherwise, the point corresponding to a/the minimum water level of this subgroup is selected, i.e., $(x,f_i(x))$. This is given in step 9 of the algorithm, and in step 10, the water levels are updated. 

Since in step 3 (prior to the iterations), the water levels of at least $|\mathcal{I}|$ nodes are filled, and at each iteration, at least one water level gets filled, i.e., $a_{i+1}(x^*,y^*)$ becomes zero (which occurs in step 10), the algorithm terminates after at most $|\textnormal{supp}(X,Y)|-|\mathcal{I}|$ iterations.
With this $p_{U|XY}$, we get $X\independent U$, and $H(Y|U)\leq (1-\sum_y\min_xp(y|x))\min\{H(X),\log|\mathcal{Y}|\}$, which proves the lower bound in (\ref{G0lowerbound}).
% In this algorithm, for any $u\in\mathcal{U}$, we have $p_{X|U}(\cdot|u)=p_X(\cdot)$. Therefore, we have $X\independent U$. Furthermore, we also have that the mass probabilities of $(X,Y)|{U=u}$ are those of $X,\ \forall u\in\mathcal{U}$. Hence, $H(X,Y|U)=H(X)$, which results in 
% \begin{align*}
%     G_0(X,Y)\geq H(Y)- (1-\sum_y\min_xp(y|x))H(X),
% \end{align*}
% which is 
}
\end{proof}
The following example clarifies the steps in Algorithm 1.
\begin{example}\label{ex.1}
Let $(X,Y)\in\{x_1,x_2,x_3\}\times\{y_1,y_2,y_3\}$ be distributed according to the joint pmf $\mathbf{P}_{XY}=\mathbf{P}_{Y|X}\mathbf{p}_X$ as
\begin{equation*}
    \mathbf{P}_{X,Y}=\begin{bmatrix}0.2&0.4&0.6\\0.5&0.2&0.3\\0.3&0.4&0.1\end{bmatrix}\begin{bmatrix}p_1\\p_2\\p_3\end{bmatrix},
\end{equation*}
where $p_i\triangleq p_X(x_i)$, and column $i$ of $\mathbf{P}_{Y|X}$ represents $p_{Y|X}(\cdot|x_i)$, $\forall i\in[3]$. The reason for representing the mass probabilities of $X$ as parameters, i.e., $p_i$'s, rather than numerical values is this interesting property the design of a privacy-preserving mapping via Algorithm 1 does not depend on $p_X$, which is elaborated further in Remark \ref{Rem1}.

Figure \ref{fig11:subim1} illustrates step 2 of the algorithm. On the left hand side of this figure, the elements of $\mathcal{X}\times\mathcal{Y}$ are arranged into 3 ($=|\mathcal{X}|$) subgroups in a column, and their corresponding probabilities are shown on their left side. In this example, we have $\mathcal{I}=\mathcal{Y}$, hence, we create 3 ($=|\mathcal{I}|$) realizations of $U$, denoted by $u_i,i\in[3]$, with the corresponding probabilities of $\min_xp(y_i|x)$, which are shown on the right side of these points. Afterwards, each $u_i$ is connected to $x_1y_i, x_2y_i, x_3y_i$, with transition probabilities of $p_i$, respectively. This is equivalent to step 2 of the algorithm.

In Figure \ref{fig11:subim2}, we have the same set of mass points $x_iy_j$'s whose mass probabilities have been updated by taking into account Figure \ref{fig11:subim1}. In other words, each mass point $x_iy_j$ has the remaining probability of $p(x_i,y_j)-p(x_i,y_j|u_j)p(u_j)$ ($=p_ip(y_j|x_i)-p_ip(u_j)=p_ia_1(x_i,y_j)$, where $a_1(\cdot,\cdot)$ is defined in step 3) to be filled with other realizations of $U$. These remaining probabilities are shown on the left side of $x_iy_j$'s. Iteration 1 starts, and $u_4$ is created, whose aim is to fill the minimum (non-zero) remaining probability , which is that of $x^*y^*$( $=x_3y_2$ in this example). This $u_4$ is connected to $x_3y_2$, and $x_1y_2$ (whose $y$-coordinate is in common with $x_3y_2$), and $x_2y_1$, which has the minimum (non-zero) water level in the subgroup of $x_2$ (since the water level of $x_2y_2$ is zero). These links are created bearing in mind that any connection to subgroup $i$ has the transition probability of $p_i,\ i\in[3].$

Taking into account the connections in Figure \ref{fig11:subim2}, the remaining probabilities are again updated in Figure \ref{fig11:subim3}, shown on the left side of $x_iy_j$'s. Iteration 2 starts, and realization $u_5$ is created in a similar way.

Again, taking into account the connections in Figure \ref{fig11:subim3}, the remaining probabilities are updated in Figure \ref{fig11:subim4}, shown on the left side of $x_iy_j$'s. Iteration 3 starts, and realization $u_6$ is created.

The remaining probabilities are updated in Figure \ref{fig11:subim5}
where we are left with only one non-zero probability in each subgroup, i.e., $0.2$. In iteration 4, which is the last one, $u_7$ is created to fill all the remaining water levels, and the algorithm terminates after 4 iterations.

Finally, the output of this algorithm is shown in Figure \ref{fig12:image}, where the transition probabilities in Figure \ref{fig12:subim1} represent $p_{XY|U}$, and those in Figure \ref{fig12:subim2} represent $p_{U|XY}$. From Figure \ref{fig12:subim1}, it is obvious that $p_{X|U}(\cdot|u)=p_X(\cdot),\ \forall u\in\mathcal{U}$, and hence, $X\independent U$. Also, $H(Y|U=u)=0,\ \forall u\in\{u_1,u_2,u_3\}$, and $H(Y|U=u)\leq H(X), \forall u\in\{u_4,u_5,u_6,u_7\}$. Therefore, $G_0(X,Y)\geq I(Y;U)\geq \left(H(Y)-\sum_{u_4}^{u_7}p(u)H(X)\right)^+=\left(H(Y)-0.5H(X)\right)^+.$\footnote{Note that in this example, $H(X)\leq\log|\mathcal{Y}|=\log 3.$}
\begin{figure}

\begin{subfigure}{0.5\textwidth}
\centering
\begin{tikzpicture}[scale=0.8, transform shape]
  \node[dspnodeopen, minimum width=4pt,  dsp/label=left] (X_1) {$p_1\times0.2:x_1y_1$};    
  \node[dspnodeopen, minimum width=4pt, below=1cm of X_1, dsp/label=left] (X_2) {$p_1\times0.5:x_1y_2$};
  \node[dspnodeopen, minimum width=4pt, below=1cm of X_2, dsp/label=left] (X_3) {$p_1\times0.3:x_1y_3$};
  \node[dspnodeopen, minimum width=4pt, below=2cm of X_3, dsp/label=left] (X_4) {$p_2\times0.4:x_2y_1$};
  \node[dspnodeopen, minimum width=4pt, below=1cm of X_4, dsp/label=left] (X_5) {$p_2\times0.2:x_2y_2$};
  \node[dspnodeopen, minimum width=4pt, below=1cm of X_5, dsp/label=left] (X_6) {$p_2\times0.4:x_2y_3$};
  \node[dspnodeopen, minimum width=4pt, below=2cm of X_6, dsp/label=left] (X_7) {$p_3\times0.6:x_3y_1$};
  \node[dspnodeopen, minimum width=4pt, below=1cm of X_7, dsp/label=left] (X_8) {$p_3\times0.3:x_3y_2$};
  \node[dspnodeopen, minimum width=4pt, below=1cm of X_8, dsp/label=left] (X_9) {$p_3\times0.1:x_3y_3$};
  
  \node[dspnodeopen, minimum width=4pt, right=3cm of X_1, dsp/label=right] (U_1) {$u_1:0.2$};    
  \node[dspnodeopen, minimum width=4pt, below=5cm of U_1, dsp/label=right] (U_2) {$u_2:0.2$};
  \node[dspnodeopen, minimum width=4pt, below=5cm of U_2, dsp/label=right] (U_3) {$u_3:0.1$};
  
  \draw[line width=1.2pt] (X_1) to node[near end, inner sep=1pt, above,sloped]{\footnotesize{$p_1$}} (U_1);
\draw[line width=1.2pt] (X_2) to node[near end, inner sep=1pt, above, sloped]{\footnotesize{$p_1$}} (U_2);
\draw[line width=1.2pt] (X_3) to node[near end, inner sep=1pt, above, sloped]{\footnotesize{$p_1$}} (U_3);
\draw[line width=1.2pt] (X_4) to node[near end, inner sep=1pt, above,sloped]{\footnotesize{$p_2$}} (U_1);
\draw[line width=1.2pt] (X_5) to node[near end, inner sep=1pt, above, sloped]{\footnotesize{$p_2$}} (U_2);
\draw[line width=1.2pt] (X_6) to node[near end, inner sep=1pt, above, sloped]{\footnotesize{$p_2$}} (U_3);
\draw[line width=1.2pt] (X_7) to node[near end, inner sep=1pt, above,sloped]{\footnotesize{$p_3$}} (U_1);
\draw[line width=1.2pt] (X_8) to node[near end, inner sep=1pt, above, sloped]{\footnotesize{$p_3$}} (U_2);
\draw[line width=1.2pt] (X_9) to node[near end, inner sep=1pt, above, sloped]{\footnotesize{$p_3$}} (U_3);

\end{tikzpicture}
\caption{}
\label{fig11:subim1}
\end{subfigure}
\begin{subfigure}{0.5\textwidth}
\centering
\begin{tikzpicture}[scale=0.8, transform shape]
  \node[dspnodeopen, minimum width=4pt,  dsp/label=left] (X_1) {$p_1\times0:x_1y_1$};    
  \node[dspnodeopen, minimum width=4pt, below=1cm of X_1, dsp/label=left] (X_2) {$p_1\times0.3:x_1y_2$};
  \node[dspnodeopen, minimum width=4pt, below=1cm of X_2, dsp/label=left] (X_3) {$p_1\times0.2:x_1y_3$};
  \node[dspnodeopen, minimum width=4pt, below=2cm of X_3, dsp/label=left] (X_4) {$p_2\times0.2:x_2y_1$};
  \node[dspnodeopen, minimum width=4pt, below=1cm of X_4, dsp/label=left] (X_5) {$p_2\times0:x_2y_2$};
  \node[dspnodeopen, minimum width=4pt, below=1cm of X_5, dsp/label=left] (X_6) {$p_2\times0.3:x_2y_3$};
  \node[dspnodeopen, minimum width=4pt, below=2cm of X_6, dsp/label=left] (X_7) {$p_3\times0.4:x_3y_1$};
  \node[dspnodeopen, minimum width=4pt, below=1cm of X_7, dsp/label=left] (X_8) {$p_3\times0.1:x_3y_2$};
  \node[dspnodeopen, minimum width=4pt, below=1cm of X_8, dsp/label=left] (X_9) {$p_3\times0:x_3y_3$};
  
  \node[dspnodeopen, minimum width=4pt, right=3cm of X_5, dsp/label=right] (U_4) {$u_4:0.1$};

  \draw[line width=1.2pt] (X_2) to node[near end, inner sep=1pt, above,sloped]{\footnotesize{$p_1$}} (U_4);
\draw[line width=1.2pt] (X_4) to node[near end, inner sep=1pt, above, sloped]{\footnotesize{$p_2$}} (U_4);
\draw[line width=1.2pt] (X_8) to node[near end, inner sep=1pt, above, sloped]{\footnotesize{$p_3$}} (U_4);

\end{tikzpicture}
\caption{}
\label{fig11:subim2}
\end{subfigure}
\par\bigskip\par\bigskip\par\bigskip\par\bigskip
\begin{subfigure}{0.5\textwidth}
\centering
\begin{tikzpicture}[scale=0.8, transform shape]
  \node[dspnodeopen, minimum width=4pt,  dsp/label=left] (X_1) {$p_1\times0:x_1y_1$};    
  \node[dspnodeopen, minimum width=4pt, below=1cm of X_1, dsp/label=left] (X_2) {$p_1\times0.2:x_1y_2$};
  \node[dspnodeopen, minimum width=4pt, below=1cm of X_2, dsp/label=left] (X_3) {$p_1\times0.2:x_1y_3$};
  \node[dspnodeopen, minimum width=4pt, below=2cm of X_3, dsp/label=left] (X_4) {$p_2\times0.1:x_2y_1$};
  \node[dspnodeopen, minimum width=4pt, below=1cm of X_4, dsp/label=left] (X_5) {$p_2\times0:x_2y_2$};
  \node[dspnodeopen, minimum width=4pt, below=1cm of X_5, dsp/label=left] (X_6) {$p_2\times0.3:x_2y_3$};
  \node[dspnodeopen, minimum width=4pt, below=2cm of X_6, dsp/label=left] (X_7) {$p_3\times0.4:x_3y_1$};
  \node[dspnodeopen, minimum width=4pt, below=1cm of X_7, dsp/label=left] (X_8) {$p_3\times0:x_3y_2$};
  \node[dspnodeopen, minimum width=4pt, below=1cm of X_8, dsp/label=left] (X_9) {$p_3\times0:x_3y_3$};
  
  \node[dspnodeopen, minimum width=4pt, right=3cm of X_5, dsp/label=right] (U_5) {$u_5:0.1$};

  \draw[line width=1.2pt] (X_3) to node[near end, inner sep=1pt, above,sloped]{\footnotesize{$p_1$}} (U_5);
\draw[line width=1.2pt] (X_4) to node[near end, inner sep=1pt, above, sloped]{\footnotesize{$p_2$}} (U_5);
\draw[line width=1.2pt] (X_7) to node[near end, inner sep=1pt, above, sloped]{\footnotesize{$p_3$}} (U_5);

\end{tikzpicture}
\caption{}
\label{fig11:subim3}
\end{subfigure}
\begin{subfigure}{0.5\textwidth}
\centering
\begin{tikzpicture}[scale=0.8, transform shape]
  \node[dspnodeopen, minimum width=4pt,  dsp/label=left] (X_1) {$p_1\times0:x_1y_1$};    
  \node[dspnodeopen, minimum width=4pt, below=1cm of X_1, dsp/label=left] (X_2) {$p_1\times0.2:x_1y_2$};
  \node[dspnodeopen, minimum width=4pt, below=1cm of X_2, dsp/label=left] (X_3) {$p_1\times0.1:x_1y_3$};
  \node[dspnodeopen, minimum width=4pt, below=2cm of X_3, dsp/label=left] (X_4) {$p_2\times0:x_2y_1$};
  \node[dspnodeopen, minimum width=4pt, below=1cm of X_4, dsp/label=left] (X_5) {$p_2\times0:x_2y_2$};
  \node[dspnodeopen, minimum width=4pt, below=1cm of X_5, dsp/label=left] (X_6) {$p_2\times0.3:x_2y_3$};
  \node[dspnodeopen, minimum width=4pt, below=2cm of X_6, dsp/label=left] (X_7) {$p_3\times0.3:x_3y_1$};
  \node[dspnodeopen, minimum width=4pt, below=1cm of X_7, dsp/label=left] (X_8) {$p_3\times0:x_3y_2$};
  \node[dspnodeopen, minimum width=4pt, below=1cm of X_8, dsp/label=left] (X_9) {$p_3\times0:x_3y_3$};
  
  \node[dspnodeopen, minimum width=4pt, right=3cm of X_5, dsp/label=right] (U_6) {$u_6:0.1$};

  \draw[line width=1.2pt] (X_3) to node[near end, inner sep=1pt, above,sloped]{\footnotesize{$p_1$}} (U_6);
\draw[line width=1.2pt] (X_6) to node[near end, inner sep=1pt, above, sloped]{\footnotesize{$p_2$}} (U_6);
\draw[line width=1.2pt] (X_7) to node[near end, inner sep=1pt, above, sloped]{\footnotesize{$p_3$}} (U_6);

\end{tikzpicture}
\caption{}
\label{fig11:subim4}
\end{subfigure}

\caption{An illustrative representation of Algorithm 1.}
\label{fig11:image}
\end{figure}
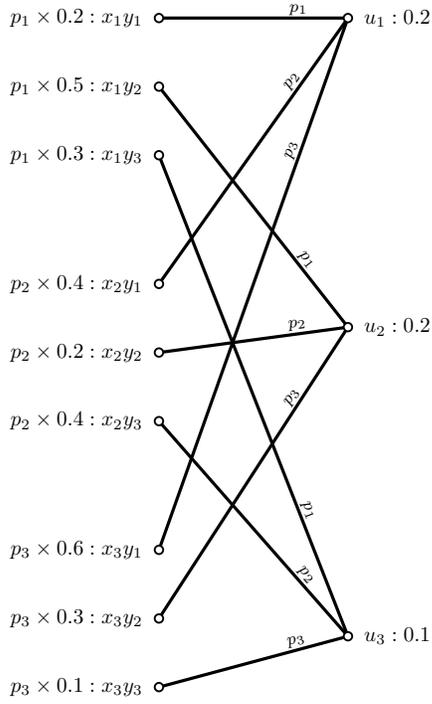
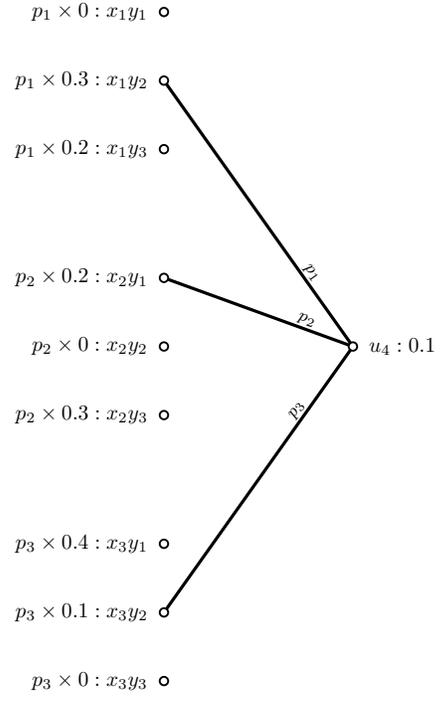
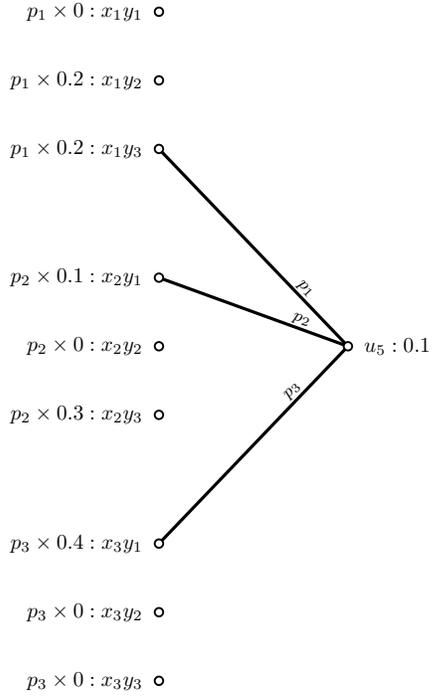
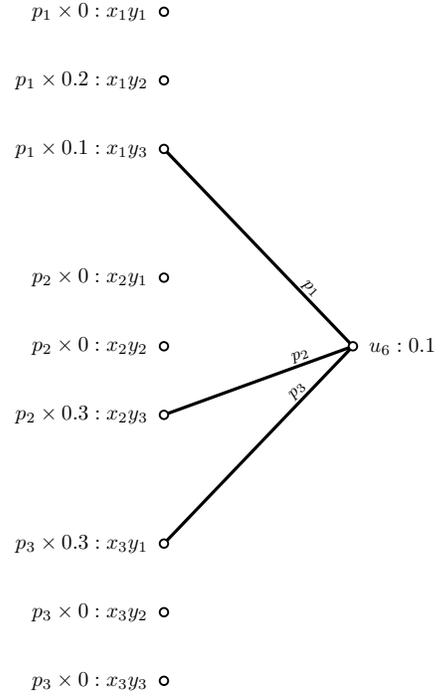
%%%%%%%%%%%%%%%%%%%%%%%%%
\begin{figure}[ht]\ContinuedFloat
\centering
    \begin{subfigure}{0.5\textwidth}
\centering
\begin{tikzpicture}[scale=0.8, transform shape]
  \node[dspnodeopen, minimum width=4pt,  dsp/label=left] (X_1) {$p_1\times0:x_1y_1$};    
  \node[dspnodeopen, minimum width=4pt, below=1cm of X_1, dsp/label=left] (X_2) {$p_1\times0.2:x_1y_2$};
  \node[dspnodeopen, minimum width=4pt, below=1cm of X_2, dsp/label=left] (X_3) {$p_1\times0:x_1y_3$};
  \node[dspnodeopen, minimum width=4pt, below=2cm of X_3, dsp/label=left] (X_4) {$p_2\times0:x_2y_1$};
  \node[dspnodeopen, minimum width=4pt, below=1cm of X_4, dsp/label=left] (X_5) {$p_2\times0:x_2y_2$};
  \node[dspnodeopen, minimum width=4pt, below=1cm of X_5, dsp/label=left] (X_6) {$p_2\times0.2:x_2y_3$};
  \node[dspnodeopen, minimum width=4pt, below=2cm of X_6, dsp/label=left] (X_7) {$p_3\times0.2:x_3y_1$};
  \node[dspnodeopen, minimum width=4pt, below=1cm of X_7, dsp/label=left] (X_8) {$p_3\times0:x_3y_2$};
  \node[dspnodeopen, minimum width=4pt, below=1cm of X_8, dsp/label=left] (X_9) {$p_3\times0:x_3y_3$};
  
  \node[dspnodeopen, minimum width=4pt, right=3cm of X_5, dsp/label=right] (U_7) {$u_7:0.2$};

  \draw[line width=1.2pt] (X_2) to node[near end, inner sep=1pt, above,sloped]{\footnotesize{$p_1$}} (U_7);
\draw[line width=1.2pt] (X_6) to node[near end, inner sep=1pt, above, sloped]{\footnotesize{$p_2$}} (U_7);
\draw[line width=1.2pt] (X_7) to node[near end, inner sep=1pt, above, sloped]{\footnotesize{$p_3$}} (U_7);

\end{tikzpicture}
\caption{}
\label{fig11:subim5}
\end{subfigure}
\caption{An illustrative representation of Algorithm 1 (cont.).}
\end{figure}
%%%%%%%%%%%%%%%%%%%%%%%%%%%%
\begin{figure}

\begin{subfigure}{0.5\textwidth}
\centering
\begin{tikzpicture}[scale=0.8, transform shape]
  \node[dspnodeopen, minimum width=4pt,  dsp/label=left] (X_1) {$p_1\times0.2:x_1y_1$};    
  \node[dspnodeopen, minimum width=4pt, below=1cm of X_1, dsp/label=left] (X_2) {$p_1\times0.5:x_1y_2$};
  \node[dspnodeopen, minimum width=4pt, below=1cm of X_2, dsp/label=left] (X_3) {$p_1\times0.3:x_1y_3$};
  \node[dspnodeopen, minimum width=4pt, below=2cm of X_3, dsp/label=left] (X_4) {$p_2\times0.4:x_2y_1$};
  \node[dspnodeopen, minimum width=4pt, below=1cm of X_4, dsp/label=left] (X_5) {$p_2\times0.2:x_2y_2$};
  \node[dspnodeopen, minimum width=4pt, below=1cm of X_5, dsp/label=left] (X_6) {$p_2\times0.4:x_2y_3$};
  \node[dspnodeopen, minimum width=4pt, below=2cm of X_6, dsp/label=left] (X_7) {$p_3\times0.6:x_3y_1$};
  \node[dspnodeopen, minimum width=4pt, below=1cm of X_7, dsp/label=left] (X_8) {$p_3\times0.3:x_3y_2$};
  \node[dspnodeopen, minimum width=4pt, below=1cm of X_8, dsp/label=left] (X_9) {$p_3\times0.1:x_3y_3$};
  
  \node[dspnodeopen, minimum width=4pt, right=4.5cm of X_1, dsp/label=right] (U_1) {$u_1:0.2$};    
  \node[dspnodeopen, minimum width=4pt, below=2cm of U_1, dsp/label=right] (U_2) {$u_2:0.2$};
  \node[dspnodeopen, minimum width=4pt, below=2cm of U_2, dsp/label=right] (U_3) {$u_3:0.1$};
  \node[dspnodeopen, minimum width=4pt, below=2cm of U_3, dsp/label=right] (U_4) {$u_4:0.1$};
  \node[dspnodeopen, minimum width=4pt, below=2cm of U_4, dsp/label=right] (U_5) {$u_5:0.1$};
  \node[dspnodeopen, minimum width=4pt, below=2cm of U_5, dsp/label=right] (U_6) {$u_6:0.1$};
  \node[dspnodeopen, minimum width=4pt, below=2cm of U_6, dsp/label=right] (U_7) {$u_7:0.2$};
 
\draw[line width=1.2pt] (X_1) to node[near end, inner sep=1pt, above]{\footnotesize{$p_1$}} (U_1);
\draw[line width=1.2pt] (X_2) to node[near end, inner sep=1pt, above, sloped]{\footnotesize{$p_1$}} (U_2);
\draw[line width=1.2pt] (X_3) to node[near end, inner sep=1pt, above, sloped]{\footnotesize{$p_1$}} (U_3);
\draw[line width=1.2pt] (X_4) to node[near end, inner sep=1pt, above]{\footnotesize{$p_2$}} (U_1);
\draw[line width=1.2pt] (X_5) to node[near end, inner sep=1pt, above, sloped]{\footnotesize{$p_2$}} (U_2);
\draw[line width=1.2pt] (X_6) to node[near end, inner sep=1pt, above, sloped]{\footnotesize{$p_2$}} (U_3);
\draw[line width=1.2pt] (X_7) to node[near end, inner sep=1pt, above,sloped]{\footnotesize{$p_3$}} (U_1);
\draw[line width=1.2pt] (X_8) to node[near end, inner sep=1pt, above, sloped]{\footnotesize{$p_3$}} (U_2);
\draw[line width=1.2pt] (X_9) to node[near end, inner sep=1pt, above, sloped]{\footnotesize{$p_3$}} (U_3);

\draw[line width=1.2pt] (X_2) to node[near end, inner sep=1pt, above,sloped]{\footnotesize{$p_1$}} (U_4);
\draw[line width=1.2pt] (X_4) to node[near end, inner sep=1pt, above, sloped]{\footnotesize{$p_2$}} (U_4);
\draw[line width=1.2pt] (X_8) to node[near end, inner sep=1pt, above, sloped]{\footnotesize{$p_3$}} (U_4);

  \draw[line width=1.2pt] (X_3) to node[near end, inner sep=1pt, above,sloped]{\footnotesize{$p_1$}} (U_5);
\draw[line width=1.2pt] (X_4) to node[near end, inner sep=1pt, above, sloped]{\footnotesize{$p_2$}} (U_5);
\draw[line width=1.2pt] (X_7) to node[near end, inner sep=1pt, above, sloped]{\footnotesize{$p_3$}} (U_5);

  \draw[line width=1.2pt] (X_3) to node[near end, inner sep=1pt, above,sloped]{\footnotesize{$p_1$}} (U_6);
\draw[line width=1.2pt] (X_6) to node[near end, inner sep=1pt, above, sloped]{\footnotesize{$p_2$}} (U_6);
\draw[line width=1.2pt] (X_7) to node[near end, inner sep=1pt, above, sloped]{\footnotesize{$p_3$}} (U_6);

  \draw[line width=1.2pt] (X_2) to node[near end, inner sep=1pt, above,sloped]{\footnotesize{$p_1$}} (U_7);
\draw[line width=1.2pt] (X_6) to node[near end, inner sep=1pt, above, sloped]{\footnotesize{$p_2$}} (U_7);
\draw[line width=1.2pt] (X_7) to node[near end, inner sep=1pt, above, sloped]{\footnotesize{$p_3$}} (U_7);
\end{tikzpicture}
\caption{Backward construction: $p_{X,Y,U}=p_{X,Y|U}\cdot p_{U}$.}
\label{fig12:subim1}
\end{subfigure}
\begin{subfigure}{0.5\textwidth}
\centering
\begin{tikzpicture}[scale=0.8, transform shape]
  \node[dspnodeopen, minimum width=4pt,  dsp/label=left] (X_1) {$p_1\times0.2:x_1y_1$};    
  \node[dspnodeopen, minimum width=4pt, below=1cm of X_1, dsp/label=left] (X_2) {$p_1\times0.5:x_1y_2$};
  \node[dspnodeopen, minimum width=4pt, below=1cm of X_2, dsp/label=left] (X_3) {$p_1\times0.3:x_1y_3$};
  \node[dspnodeopen, minimum width=4pt, below=2cm of X_3, dsp/label=left] (X_4) {$p_2\times0.4:x_2y_1$};
  \node[dspnodeopen, minimum width=4pt, below=1cm of X_4, dsp/label=left] (X_5) {$p_2\times0.2:x_2y_2$};
  \node[dspnodeopen, minimum width=4pt, below=1cm of X_5, dsp/label=left] (X_6) {$p_2\times0.4:x_2y_3$};
  \node[dspnodeopen, minimum width=4pt, below=2cm of X_6, dsp/label=left] (X_7) {$p_3\times0.6:x_3y_1$};
  \node[dspnodeopen, minimum width=4pt, below=1cm of X_7, dsp/label=left] (X_8) {$p_3\times0.3:x_3y_2$};
  \node[dspnodeopen, minimum width=4pt, below=1cm of X_8, dsp/label=left] (X_9) {$p_3\times0.1:x_3y_3$};
  
  \node[dspnodeopen, minimum width=4pt, right=4.5cm of X_1, dsp/label=right] (U_1) {$u_1:0.2$};    
  \node[dspnodeopen, minimum width=4pt, below=2cm of U_1, dsp/label=right] (U_2) {$u_2:0.2$};
  \node[dspnodeopen, minimum width=4pt, below=2cm of U_2, dsp/label=right] (U_3) {$u_3:0.1$};
  \node[dspnodeopen, minimum width=4pt, below=2cm of U_3, dsp/label=right] (U_4) {$u_4:0.1$};
  \node[dspnodeopen, minimum width=4pt, below=2cm of U_4, dsp/label=right] (U_5) {$u_5:0.1$};
  \node[dspnodeopen, minimum width=4pt, below=2cm of U_5, dsp/label=right] (U_6) {$u_6:0.1$};
  \node[dspnodeopen, minimum width=4pt, below=2cm of U_6, dsp/label=right] (U_7) {$u_7:0.2$};
  \draw[line width=1.2pt] (X_1) to node[near end, inner sep=1pt, above]{\footnotesize{$1$}} (U_1);
\draw[line width=1.2pt] (X_2) to node[near end, inner sep=1pt, above, sloped]{\footnotesize{$\frac{2}{5}$}} (U_2);
\draw[line width=1.2pt] (X_3) to node[near end, inner sep=1pt, above, sloped]{\footnotesize{$\frac{1}{3}$}} (U_3);
\draw[line width=1.2pt] (X_4) to node[near end, inner sep=1pt, above]{\footnotesize{$\frac{1}{2}$}} (U_1);
\draw[line width=1.2pt] (X_5) to node[near end, inner sep=1pt, above, sloped]{\footnotesize{$1$}} (U_2);
\draw[line width=1.2pt] (X_6) to node[near end, inner sep=1pt, above, sloped]{\footnotesize{$\frac{1}{4}$}} (U_3);
\draw[line width=1.2pt] (X_7) to node[near end, inner sep=1pt, above,sloped]{\footnotesize{$\frac{1}{3}$}} (U_1);
\draw[line width=1.2pt] (X_8) to node[near end, inner sep=1pt, above, sloped]{\footnotesize{$\frac{2}{3}$}} (U_2);
\draw[line width=1.2pt] (X_9) to node[near end, inner sep=1pt, above, sloped]{\footnotesize{$1$}} (U_3);
  \draw[line width=1.2pt] (X_3) to node[near end, inner sep=1pt, above,sloped]{\footnotesize{$\frac{1}{3}$}} (U_4);
\draw[line width=1.2pt] (X_4) to node[near end, inner sep=1pt, above, sloped]{\footnotesize{$\frac{1}{4}$}} (U_4);
\draw[line width=1.2pt] (X_8) to node[near end, inner sep=1pt, above, sloped]{\footnotesize{$\frac{1}{3}$}} (U_4);
  \draw[line width=1.2pt] (X_3) to node[near end, inner sep=1pt, above,sloped]{\footnotesize{$\frac{1}{3}$}} (U_5);
\draw[line width=1.2pt] (X_4) to node[near end, inner sep=1pt, above, sloped]{\footnotesize{$\frac{1}{4}$}} (U_5);
\draw[line width=1.2pt] (X_7) to node[near end, inner sep=1pt, above, sloped]{\footnotesize{$\frac{1}{6}$}} (U_5);
\draw[line width=1.2pt] (X_2) to node[near end, inner sep=1pt, above,sloped]{\footnotesize{$\frac{1}{5}$}} (U_6);
\draw[line width=1.2pt] (X_6) to node[near end, inner sep=1pt, above, sloped]{\footnotesize{$\frac{1}{4}$}} (U_6);
\draw[line width=1.2pt] (X_7) to node[near end, inner sep=1pt, above, sloped]{\footnotesize{$\frac{1}{6}$}} (U_6);

  \draw[line width=1.2pt] (X_2) to node[near end, inner sep=1pt, above,sloped]{\footnotesize{$\frac{2}{5}$}} (U_7);
\draw[line width=1.2pt] (X_6) to node[near end, inner sep=1pt, above, sloped]{\footnotesize{$\frac{1}{2}$}} (U_7);
\draw[line width=1.2pt] (X_7) to node[near end, inner sep=1pt, above, sloped]{\footnotesize{$\frac{1}{3}$}} (U_7);
\end{tikzpicture}
\caption{Forward construction: $p_{X,Y,U}=p_{X,Y}\cdot p_{U|X,Y}$}
\label{fig12:subim2}
\end{subfigure}
\caption{The output of Algorithm 1.}
\label{fig12:image}
\end{figure}
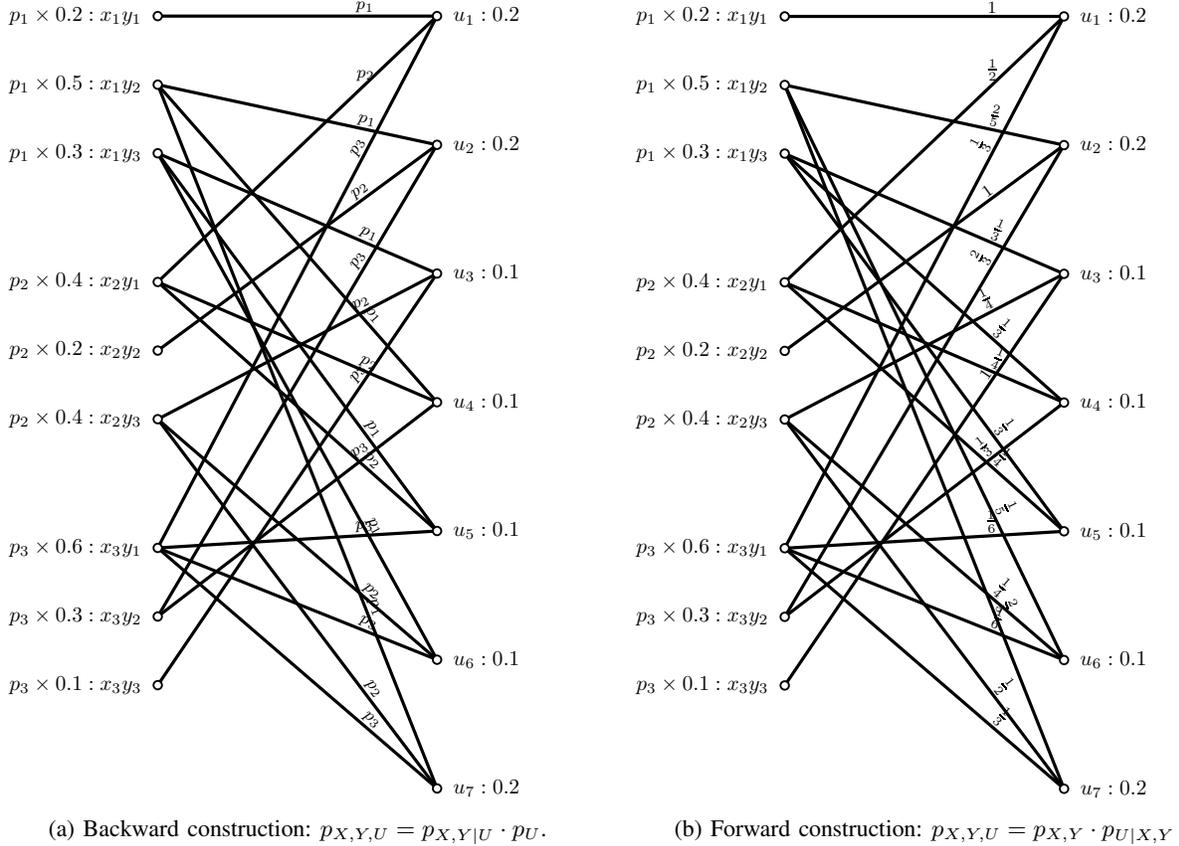
\end{example}

\begin{remark}\label{Rem1}
(\textbf{$p_X$-invariance}) An advantage of the achievable scheme in Algorithm 1 is that it does not depend on the distribution of the private data, i.e., $p_X(\cdot)$. In other words, the privacy-preserving mapping $p_{U|{XY}}$ obtained via Algorithm 1 can be derived regardless of the knowledge about $p_X(\cdot)$, as long as $p_{Y|X}$ is given. This can be verified by the fact that none of the 14 steps of Algorithm 1 rely on the knowledge of $p_X$.\footnote{Note that in the explanation of Algorithm 1, we indeed made use of $p_X$, but this should not be confusing, since that explanation is about the backward construction $p_{XY|U}$.} This is the reason that in Example 1, the mass probabilities of $X$ are given only as parameters $p_1,p_2,p_3$, and as it can be verified in Figure \ref{fig12:subim2}, the mapping $p_{U|XY}$ does not depend on a specific choice of them. This feature of Algorithm 1 is not only of practical interest (e.g., when the distribution of the private data is unknown or difficult to estimate), but also helpful in theory, as used in Corollary \ref{cor3}. Finally, it is important to emphasize that for a fixed $p_{Y|X}$, it is the proposed privacy-preserving mapping $p_{U|XY}$ that is $p_X$-invariant, not the resulting utility, i.e., $I(Y;U)$.
\end{remark}
\begin{remark}\label{rem44}(\textbf{Number of iterations})
In the explanation of Algorithm 1, it is stated that since at each iteration of the algorithm, at least one water level is filled, and the algorithm terminates after all these levels are filled, the number of iterations is upper bounded by the number of non-zero remaining probabilities prior to the start of the iterations, which is at most $|\textnormal{supp}(X,Y)|-|\mathcal{I}|$. While this is correct, we observe that, as in Figure \ref{fig11:subim4}, in the very last iteration we have exactly $|\mathcal{X}|$ non-zero and equal water levels, one in each subgroup, that are filled together in one iteration. This is a direct consequence of the fact that at each step of producing a new realization for $U$ in the algorithm, i) each subgroup of nodes has the same amount of total water levels, and ii) each subgroup of nodes undergoes the same amount of decrement in water levels. As a result, in the last iteration of the algorithm, we are left with $|\mathcal{X}|$ equal (non-zero) remaining water levels to be filled at once with the last realization of $U$. Therefore, the algorithm terminates after $N$ iterations with $N\leq|\textnormal{supp}(X,Y)|-|\mathcal{I}|-|\mathcal{X}|+1$. Moreover, since at each iteration, we get a realization for $U$, and we already have $|\mathcal{I}|$ realizations before the iterations start,  we have $|\mathcal{U}|\leq N+|\mathcal{I}|=|\textnormal{supp}(X,Y)|-|\mathcal{X}|+1$.
\end{remark}
The following Lemma is needed to obtain an upper bound on $G_0(X,Y)$ in the sequel.
\begin{lemma}\label{lf}
    Let $f(X)$ be a function of $X\sim p$ such that it has at least two realizations. We have
    \begin{equation}
        H(f(X))\geq H_b(\min_xp(x)).
    \end{equation}
\end{lemma}
\begin{proof}
       The proof is provided in Appendix \ref{app1}.
\end{proof}
{\color{black}
\begin{theorem}\label{th2}
For a given pair $(X,Y)\sim p_{XY}$, we have
\begin{equation}\label{G0upperbound}
    G_0(X,Y)\leq H(Y)-\left(1-\sum_y\min_xp(y|x)\right)H_b\left(\min_xp_X(x)\right).
\end{equation}
\end{theorem}}
{\color{black}
\begin{proof}
If $Y$ is a singleton, i.e., $|\mathcal{Y}|= 1$, we have $G_0(X,Y)=0$ and (\ref{G0upperbound}) follows, since $\min_xp(y|x)=1$. Therefore, in what follows, we assume that $|\mathcal{Y}|\geq 2$.

From Lemma \ref{lemcardi}, in an optimal mapping $p_{U|XY}$, for any $(x,u)\in\mathcal{X}\times\mathcal{U}$, there exists exactly one $y_{x,u}\in\mathcal{Y}$ such that $p(x,y_{x,u},u)>0$, and we have $p(x,y_{x,u}|u)=p(x)$. For any $y\in\mathcal{Y}$, let $\mathcal{U}_y\triangleq\{u\in\mathcal{U}|p(x,y,u)>0, \forall x\in\mathcal{X}\}$ be the set of realizations of $U$ which are connected to pairs $(x_1,y),(x_2,y),\ldots,(x_{|\mathcal{X}|},y)$. Define $\Tilde{\mathcal{U}}\triangleq\cup_{y\in\mathcal{Y}}\mathcal{U}_y$. Since $\mathcal{U}_{y'}\cap\mathcal{U}_{y''}=\emptyset$ when $y'\neq y''$, it is immediate that $H(Y|U=u)=0,\ \forall u\in\Tilde{\mathcal{U}}$. Since $Y$ conditioned on $\{U=u\},\ \forall u\not\in\Tilde{\mathcal{U}}$ is a function of $X$ which has at least two realizations, from Lemma \ref{lf}, we get $H(Y|U=u)\geq H_b\left(\min_xp_X(x)\right),\ \forall u\not\in\Tilde{\mathcal{U}}$.

For an arbitrary $y\in\mathcal{Y}$, we have
\begin{align}
      p(x)p(y|x)&=\sum_{u\in\mathcal{U}}p(x,y,u)\nonumber\\
    &=\sum_{u\in\mathcal{U}_y}p(x,y,u)+\sum_{u\not\in\mathcal{U}_y}p(x,y,u)\nonumber\\
    &\geq\sum_{u\in\mathcal{U}_y}p(u)p(x,y|u)\nonumber\\
    &=p(x)\textnormal{Pr}\{U\in\mathcal{U}_y\}\label{fl},\ \forall x\in\mathcal{X},
\end{align}
where (\ref{fl}) follows from having $p(x,y|u)=p(x),\ \forall u\in\mathcal{U}_y$. Hence, we have $\textnormal{Pr}\{U\in\mathcal{U}_y\}\leq\min_xp(y|x),\ \forall y\in\mathcal{Y}$. Noting that $\Tilde{\mathcal{U}}$ is the union of disjoint sets $\mathcal{U}_y,\forall y\in\mathcal{Y}$, we get
\begin{align}
    \textnormal{Pr}\{U\in\tilde{\mathcal{U}}\}&=\sum_y\textnormal{Pr}\{U\in\mathcal{U}_y\}\nonumber\\
    &\leq\sum_y\min_xp(y|x).\label{pubound}
\end{align}
We can write
\begin{align}
    H(Y|U)&=\sum_{u\in\mathcal{U}}p(u)H(Y|U=u)\nonumber\\
    &=\sum_{u\in\tilde{\mathcal{U}}}p(u)H(Y|U=u)+\sum_{u\not\in\tilde{\mathcal{U}}}p(u)H(Y|U=u)\nonumber\\
    &\geq\sum_{u\in\tilde{\mathcal{U}}}p(u)\cdot 0+\sum_{u\not\in\tilde{\mathcal{U}}}p(u)H_b(\min_xp_X(x))\nonumber\\
    &=\left(1-\textnormal{Pr}\{U\in\tilde{\mathcal{U}}\}\right)H_b(\min_xp_X(x))\nonumber\\
    &\geq\left(1-\sum_y\min_xp(y|x)\right)H_b(\min_xp_X(x)),\label{fl2}
\end{align}
where (\ref{fl2}) follows from (\ref{pubound}). This proves (\ref{G0upperbound}).

%%%%%%%%%%%%%%%%%%%%%%%
% If $X$ is a function of $Y$, from (\ref{Gup}), we have $G_0(X,Y)\leq H(Y|X)=H(Y)-H(X)$. Combining this with Theorem \ref{thbound5}, we get 
% \begin{align}
%     H(Y)-(1-\sum_y\min_xp(y|x))H(X)\leq G_0(X,Y)
%     \leq  H(Y)-H(X).\label{cherk}
% \end{align}
% We have $\sum_y\min_xp(y|x)=1$, if $X$ is a singleton, and $\sum_y\min_xp(y|x)=0$, otherwise. In both cases, the upper and lower bounds in (\ref{cherk}) coincide. 

\end{proof}
}
\begin{example} Consider $(X,Y)\sim p_{XY}$, in which $X$ is uniformly distributed over $[0:K-1]$, where $K>2$ is an arbitrary integer. Let $Y$ conditioned on $\{X=x\}$ be uniformly distributed on $[x:x+K-2]\textnormal{ mod }K$. Hence, $Y$ is also uniformly distributed over $[0:K-1]$. In this setting, the upper bounds in (\ref{ful1}) and (\ref{G0upperbound}) are
\begin{align*}
    H(Y|X)&=\log(K-1)\nonumber\\
    H(Y)-\left(1-\sum_y\min_xp(y|x)\right)H_b(\min_xp(x))&=\frac{K-1}{K}\log(K-1).
\end{align*}
Obviously, the bound in (\ref{G0upperbound}) is tighter in this example, and it can be readily verified that Algorithm 1 achieves it. 
\end{example}
The following lemma is needed in its following Theorem.
\begin{lemma}\label{p3s}
    For the mass probabilities $p_1,p_2,p_3$ (i.e., $p_i\geq 0,\ i\in[3],\sum_ip_i=1$), we have
    \begin{equation}\label{nbv1}
        H_b(p_i)\leq H_b(p_j)+H_b(p_k),\ i,j,k\in[3],\ j\neq k,
    \end{equation}
    where $H_b(\cdot)$ denotes the binary entropy function.
\end{lemma}
\begin{proof}
The proof is provided in Appendix \ref{app3}.
\end{proof}
\begin{theorem}\label{th3}
    When $X$ is binary, Algorithm 1 is optimal, and we have
\begin{equation}\label{G0eq}
    G_0(X,Y)=H(Y)-\left(1-\sum_y\min_xp(y|x)\right)H(X).
\end{equation}
Also, when $(|\mathcal{X}|,|\mathcal{Y}|)=(3,2)$, Algorithm 1 is optimal, and letting $\mathcal{Y}\triangleq\{y_1,y_2\}$, if we label the elements of $\mathcal{X}$ according to $x_1\triangleq\min_{x\in\mathcal{X}}p(y_1|x), x_2\triangleq\min_{x\in\mathcal{X}}p(y_2|x)$, we have
\begin{equation}\label{G0eq2}
    G_0(X,Y)=H(Y)-\left(p(y_1|x_3)-p(y_1|x_1)\right)H_b(p_1)-\left(p(y_2|x_3)-p(y_2|x_2)\right)H_b(p_2),
\end{equation}
where $p_i\triangleq p(x_i),\ i\in[3].$
\end{theorem}
\begin{proof}
   When $X$ is binary, we have $H(X)=H_b(\min_xp(x))$, and from (\ref{G0upperbound}) and Theorem \ref{thbound5}, (\ref{G0eq}) is obtained, and Algorithm 1 achieves it.

   When $(|\mathcal{X}|,|\mathcal{Y}|)=(3,2)$, we prove the optimality by the simplex method \cite{murty}. As already stated, $G_0(X,Y)$ can be obtained via an LP. The problem is to find values for $p_U(\cdot)$ in Figure \ref{fig1} such that $H(Y|U)$ is minimized and $p_{XY}(\cdot,\cdot)=\sum_u p_{XY|U}(\cdot,\cdot|u)p_U(u)$. For $i,j,k\in[2]$, let $P_{ijk}$ denote the probability of that $u$ which is connected to $(x_1,y_i),(x_2,y_j),(x_3,y_k)$ with transition probabilities $p1,p2$ and $p_3$, respectively. For example, in Figure \ref{fig1}, we have $P_{111}=p(u_1),P_{112}=p(u_2), P_{121}=p(u_3)$, and so on. As a result, the LP minimizes $H(Y|U)$, which is
   \begin{equation}\label{obj0}
       P_{111}\cdot 0+P_{112}H_b(p_3)+P_{121}H_b(p2)+P_{122}H_b(p_1)+P_{211}H_b(p_1)+P_{212}H_b(p_2)+P_{221}H_b(p_3)+P_{222}\cdot 0
   \end{equation}
over the non-negative values of $P_{ijk}, i,j,k\in[2]$ such that
\begin{align}\label{const0}
    P_{111}+P_{112}+P_{121}+P_{122}&=p(y_1|x_1)\nonumber\\
    P_{121}+P_{122}+P_{221}+P_{222}&=p(y_2|x_2)\nonumber\\
    P_{111}+P_{121}+P_{211}+P_{221}&=p(y_1|x_3)\nonumber\\
    P_{111}+P_{112}+P_{121}+P_{122}+P_{211}+P_{212}+P_{221}+P_{222}&=1.
\end{align}
Changing the order of the variables, the simplex tableau for this LP is provided in Table \ref{table:1}. By performing Gaussian elimination (i.e., subtracting row 1 from row 3, and then subtracting the sum of row 1, row 2, and the resulting row 3 from row 4), we obtain a canonical tableau as in Table \ref{table:2}. Note that all the elements of the rightmost column are non-negative due to the initial convention $x_1\triangleq\min_{x\in\mathcal{X}}p(y_1|x), x_2\triangleq\min_{x\in\mathcal{X}}p(y_2|x)$.
\begin{table}[h!]
\centering
\begin{tabular}{||c c c c c c c c c||} 
 \hline
 $P_{111}$ & $P_{222}$ & $P_{211}$ & $P_{212}$ &$ P_{112}$ & $P_{121}$ & $P_{122}$ & $P_{221}$ & \\ [0.5ex] 
 \hline\hline
1 & 0 & 0 & 0 & 1 & 1 & 1 & 0 & $p(y_1|x_1)$\\ 
 \hline
0 & 1 & 0 & 0 & 0 & 1 & 1 & 1 & $p(y_2|x_2)$\\
 \hline
 1 & 0 & 1 & 0 & 0 & 1 & 0 & 1 & $p(y_1|x_3)$ \\
 \hline
 1 & 1 & 1 & 1 & 1 & 1 & 1 & 1 & 1 \\
 \hline
\end{tabular}
\caption{Simplex tableau of order 4 and dimension 8.}
\label{table:1}
\end{table}
%%%%%%%%%%%%
\begin{table}[h!]
\centering
\begin{tabular}{||c c c c c c c c c||} 
 \hline
 $P_{111}$ & $P_{222}$ & $P_{211}$ & $P_{212}$ &$ P_{112}$ & $P_{121}$ & $P_{122}$ & $P_{221}$ & \\ [0.5ex] 
 \hline\hline
1 & 0 & 0 & 0 & 1 & 1 & 1 & 0 & $p(y_1|x_1)$\\ 
 \hline
0 & 1 & 0 & 0 & 0 & 1 & 1 & 1 & $p(y_2|x_2)$\\
 \hline
0 & 0 & 1 & 0 & -1 & 0 & -1 & 1 & $p(y_1|x_3)-p(y_1|x_1)$ \\
 \hline
 0 & 0 & 0 & 1 & 1 & -1 & 0 & -1 & $p(y_2|x_3)-p(y_2|x_2)$ \\
 \hline
\end{tabular}
\caption{Canonical form.}
\label{table:2}
\end{table}

The first four columns of the tableau in Table \ref{table:2} form a basis, and
\begin{align}\label{bfs0}
    [P_{111},P_{222},P_{211},P_{212},P_{112},P_{121},&P_{122},P_{221}]^T=\nonumber\\&[p(y_1|x_1),p(y_2|x_2),p(y_1|x_3)-p(y_1|x_1),p(y_2|x_3)-p(y_2|x_2),0,0,0,0]^T
\end{align}
is a basic feasible solution, which results in
\begin{equation}\label{bfs1}
    H(Y|U)=\left(p(y_1|x_3)-p(y_1|x_1)\right)H_b(p_1)+\left(p(y_2|x_3)-p(y_2|x_2)\right)H_b(p_2).
\end{equation}

In order to show that no other feasible solution outperforms (\ref{bfs0}), i.e., resulting in a smaller $H(Y|U)$ than (\ref{bfs1}), we proceed as follows. For any feasible solution $\Tilde{P}_{ijk},\ i,j,k\in[2]$, with some calculations, we get
\begin{align}
    H(Y|U)&=\left(p(y_1|x_3)-p(y_1|x_1)\right)H_b(p_1)+\left(p(y_2|x_3)-p(y_2|x_2)\right)H_b(p_2)\label{fs1}\\
    &\ \ \ +(H_b(p_3)+H_b(p_1)-H_b(p_2))\Tilde{P}_{112}+2H_b(p_2)\Tilde{P}_{121}\label{fs2}\\
    &\ \ \ +2H_b(p_1)\Tilde{P}_{122}+(H_b(p_3)+H_b(p_2)-H_b(p_1))\Tilde{P}_{221}\label{fs3}.
\end{align}
The RHS in (\ref{fs1}) is equal to (\ref{bfs1}). From Lemma \ref{p3s} and non-negativity of entropy, all the remaining terms in (\ref{fs2}) and (\ref{fs3}) are non-negative. As a result, no other feasible solution can produce a smaller $H(Y|U)$ than (\ref{bfs1}), which proves (\ref{G0eq2}).
\end{proof}

\begin{corollary}\label{cor3}
When $|\mathcal{X}|=2$, or $(|\mathcal{X}|,|\mathcal{Y}|)=(3,2)$, for fixed $p_{Y|X}$, the optimal $p^*_{U|X,Y}$ is $p_X$-invariant \footnote{This, however, does not hold in general.}, and $G_0(X,Y)$ is convex in $p_X$.
\end{corollary}
{\color{black}

\begin{proof}
The first part of this claim is proved by combining the optimality of Algorithm 1 in Theorem \ref{th3}, and Remark \ref{Rem1}.

In what follows, we provide two methods to prove the second part of the claim. The first method makes use of $p_X$-invariance, while the second one relies on directly inspecting $G_0(X,Y)$ in (\ref{G0eq}) and (\ref{G0eq2}).
\subsection{First method}
Fix $\lambda\in(0,1)$. Let $p_{W|XY}$, $p_{V|XY}$ and $p_{U|XY}$ be optimal solutions in the evaluation of $G_0(p_X'\cdot p_{Y|X})$, $G_0(p_X''\cdot p_{Y|X})$ and $G_0\left((\lambda p'_X+\bar{\lambda}p_X'')\cdot p_{Y|X}\right)$, respectively. From Corollary \ref{cor3}, we can set $\mathcal{W}=\mathcal{V}=\mathcal{U}$, and write
\begin{equation}\label{ch0}
 p_{W|XY}=p_{V|XY}=p_{U|XY}.    
\end{equation}
In the sequel, by $p_{XYW}$, $p_{XYV}$ and $p_{XYU}$, we are referring to $p'_X\cdot p_{Y|X}\cdot p_{W|XY}$, $p''_X\cdot p_{Y|X}\cdot p_{V|XY}$ and $(\lambda p'_X+\bar{\lambda}p''_X)\cdot p_{Y|X}\cdot p_{U|XY}$, respectively, which share the same support, denoted by $\mathcal{S}$, and induce
\begin{equation}\label{dalil1}
    p_{YU}=\lambda p_{YW}+\bar{\lambda}p_{YV}.
\end{equation}
For any tuple $(x,y,w)\in\mathcal{S}$, we have
\begin{align}
    p_W(w)&=\frac{p_{XYW}(x,y,w)}{p_{XY|W}(x,y|w)}\nonumber\\
    &=\frac{p'_X(x)p_{Y|X}(y|x)p_{W|XY}(w|x,y)}{p_{X|W}(x|w)p_{Y|XW}(y|x,w)}\nonumber\\
    &=\frac{p'_X(x)p_{Y|X}(y|x)p_{W|XY}(w|x,y)}{p'_{X}(x)}\label{chera1}\\
    &=p_{Y|X}(y|x)p_{W|XY}(w|x,y),\label{chera2}
\end{align}
where (\ref{chera1}) follows from i) having $X\independent W$ in $p_{XYW}$, and ii) having $p_{Y|XW}(y|x,w)=1$, since $(x,y,w)\in\mathcal{S}$.

From (\ref{ch0}) and (\ref{chera2}), we get that $p_U(\cdot)=p_W(\cdot)=p_V(\cdot)$, which in conjunction with (\ref{dalil1}) and convexity of $I(A;B)$ in $p_{A|B}$ for fixed $p_B$, results in
\begin{equation*}
    I(Y;U)\leq\lambda I(Y;W)+\bar{\lambda}I(Y;V).
\end{equation*}
This is equivalent to
\begin{equation*}
    G_0\left((\lambda p'_X+\bar{\lambda}p_X'')\cdot p_{Y|X}\right)\leq\lambda G_0(p'_X\cdot p_{Y|X})+\bar{\lambda}G_0(p''_X\cdot p_{Y|X}),
\end{equation*}
which completes the proof.
\subsection{Second method}
According to \cite{Anantharam}, for fixed $p_{Y|X}$, the minimum value of $\lambda$, for which $H(Y)-\lambda H(X)$ is a convex functional of $p_X$ is $\sup_{p_X}s^*(X;Y)$, where $s^*(X;Y)\triangleq\sup_{Z:Z-X-Y,Z\not\independent X}\frac{I(Z;Y)}{I(Z;X)}$ is the strong data processing coefficient. Since we have $1-\sum_y\min_xp(y|x)\geq \sup_{p_X}s^*(X;Y)$ (see \cite[Corollary 6]{Issa}), we conclude that $G_0(X,Y)$ in (\ref{G0eq}) is convex in $p_X$ for fixed $p_{Y|X}$. 

For fixed $p_{Y|X}$, denoting $\alpha_i\triangleq p(y_1|x_i),\ i\in[3]$ for simplicity, (\ref{G0eq2}) becomes
\begin{equation}
    f(p_1,p_2)\triangleq G_0(X,Y)=H_b(p_y)-(\alpha_3-\alpha_1)H_b(p_1)-(\alpha_2-\alpha_3)H_b(p_2),
\end{equation}
where $p_y\triangleq (\alpha_1-\alpha_3)p_1+(\alpha_2-\alpha_3)p_2+\alpha_3$.

To prove the convexity of $G_0$ in $p_X$, we show that $f$ is convex in $(p_1,p_2)$. By some calculations, the Hessian matrix of $f$ is obtained as
\begin{equation}
    \nabla^2f=\begin{bmatrix}
        a+b&\sqrt{ac}\\
        \sqrt{ac}&c+d
    \end{bmatrix},
\end{equation}
with $a\triangleq-\frac{(\alpha_1-\alpha_3)^2}{p_y(1-p_y)},b\triangleq\frac{\alpha_3-\alpha_1}{p_1(1-p_1)},c\triangleq-\frac{(\alpha_2-\alpha_3)^2}{p_y(1-p_y)}$ and $d\triangleq\frac{\alpha_2-\alpha_3}{p_2(1-p_2)}$.

The characteristic polynomial of $\nabla^2f$ is $\lambda^2-(a+b+c+d)\lambda +ad+b(c+d)$ whose roots determine the eigenvalues. From the initial convention, we have $\alpha_1\leq\alpha_3\leq\alpha_2$. This, in conjunction with the inequality $\alpha_i\Bar{\alpha_j}+\alpha_j\Bar{\alpha_i}\geq|\alpha_i-\alpha_j|,\ i,j\in[3]$ results in $a+b+c+d\geq 0$ and $ad+b(c+d)\geq0$, which in turn means that the eigenvalues of $\nabla^2f$ are non-negative. Therefore, $\nabla^2f$ is positive semi-definite and $f$ is convex in $(p_1,p_2)$.
\end{proof}
% \begin{corollary}
% If $Y$ is uniform on $\mathcal{Y}$, we have
% \begin{equation*}
%     G_0(X,Y)\geq \max\bigg\{\bigg(2^{-\mathcal{L}(Y\to X)}\log|\mathcal{Y}|-1\bigg)^+,\bigg(\log|\mathcal{Y}| -\left(1-2^{-\mathcal{L}^c(X\to Y)}\right)H(X)\bigg)^+\bigg\}.
% \end{equation*}
% \end{corollary}
The convexity result in Corollary \ref{cor3} is not specific to $|\mathcal{X}|=2$ or $(|\mathcal{X}|,|\mathcal{Y}|)=(3,2)$ as the second part of the following Theorem indicates.
\begin{theorem}\label{concavity2}
For given $\epsilon\geq 0$ and $p_X$, $G_\epsilon(X,Y)$ is concave in $p_{Y|X}$. Furthermore, for given $p_{Y|X}$, $G_0(X,Y)$ and $g_0(X,Y)$ are convex in $p_X$.
\end{theorem}
\begin{proof}
    The first part of the claim is proved as follows. Fix $\epsilon\geq 0$. Given two conditional pmfs $p'_{Y|X}$ and $p''_{Y|X}$, let $p_{W|XY}$ and $p_{V|XY}$ be maximizers of $G_\epsilon(p_X\cdot p'_{Y|X})$ and $G_\epsilon(p_X\cdot p''_{Y|X})$ in (\ref{Geps}), respectively. In other words, when $(X,Y)$ is distributed according to $p_X\cdot p'_{Y|X}$ (or $p_X\cdot p''_{Y|X}$), an optimal privacy-preserving mapping is $p_{W|XY}$ (or $p_{V|XY}$). In the sequel, $p_{XYW}$ and $p_{XYV}$ refer to $p_X\cdot p'_{Y|X}\cdot p_{W|XY}$ and $p_X\cdot p''_{Y|X}\cdot p_{V|XY}$, respectively. Without loss of optimality, select the alphabets $\mathcal{W}$ and $\mathcal{V}$, such that $\mathcal{W}\cap\mathcal{V}=\emptyset.$ Fix $\lambda\in(0,1)$, and let $\mathcal{U}\triangleq \mathcal{W}\cup\mathcal{V}$. Let $(X,Y)\sim p_X\cdot(\lambda p'_{Y|X}+\bar{\lambda}p''_{Y|X})$, and define the following conditional pmf
\begin{equation}\label{pudef}
    p_{U|XY}(u|x,y)\triangleq\frac{\lambda p'(y|x)p_{W|XY}(u|x,y)\cdot\mathds{1}_{\{u\in\mathcal{W}\}}+\bar{\lambda} p''(y|x)p_{V|XY}(u|x,y)\cdot\mathds{1}_{\{u\in\mathcal{V}\}}}{\lambda p'(y|x)+\bar{\lambda}p''(y|x)},
\end{equation}
for all $(u,x,y)\in\mathcal{U}\times\textnormal{supp}(X,Y)$. 

In what follows, we show that the joint pmf $p_{XYU}$ induced by (\ref{pudef}) results in $I(X;U)\leq\epsilon$ and $I(Y;U)\geq\lambda I(Y;W)+\bar{\lambda}I(Y;V)$, which completes the proof.

The construction in (\ref{pudef}) results in
\begin{align}
    p_U(u)&=\lambda p_W(u)\cdot\mathds{1}_{\{u\in\mathcal{W}\}}+\bar{\lambda}p_V(u)\cdot\mathds{1}_{\{u\in\mathcal{V}\}}\label{ph1}\\
    p_{XY|U}(x,y|u)&=p_{XY|W}(x,y|u)\cdot\mathds{1}_{\{u\in\mathcal{W}\}}+p_{XY|V}(x,y|u)\cdot\mathds{1}_{\{u\in\mathcal{V}\}},\ \forall (u,x,y)\in\mathcal{U}\times\textnormal{supp}(X,Y).\label{ph2}
\end{align}
Let $E\triangleq\mathds{1}_{\{U\in\mathcal{W}\}}$ be a binary r.v. defined as a function of $U$. From (\ref{ph1}), we have $p_E(1)=\lambda$. Also,
\begin{align}
    p_{E|X}(1|x)&=\sum_yp_{EY|X}(1,y|x)\nonumber\\
    &=\sum_yp_{Y|X}(y|x)p_{E|XY}(1|x,y)\nonumber\\
    &=\sum_y(\lambda p'(y|x)+\bar{\lambda}p''(y|x))p_{E|XY}(1|x,y)\nonumber\\
    &=\sum_y(\lambda p'(y|x)+\bar{\lambda}p''(y|x))\sum_{u\in\mathcal{W}}p_{U|XY}(u|x,y)\nonumber
\end{align}
\begin{align}
    &=\sum_y(\lambda p'(y|x)+\bar{\lambda}p''(y|x))\sum_{u\in\mathcal{W}}\frac{\lambda p'(y|x)p_{W|XY}(u|x,y)}{\lambda p'(y|x)+\bar{\lambda}p''(y|x)}\label{fyek}\\
    &=\sum_y(\lambda p'(y|x)+\bar{\lambda}p''(y|x))\frac{\lambda p'(y|x)}{\lambda p'(y|x)+\bar{\lambda}p''(y|x)}\nonumber\\
    &=\lambda,\ \forall x\in\mathcal{X}\label{fdo}
\end{align}
where (\ref{fyek}) results from (\ref{pudef}), and (\ref{fdo}) results in $E\independent X$ with $p_E(1)=\lambda$. 

We can write
\begin{align}
  p_{XYU|E}(x,y,u|e)
  &=\frac{p_U(u)p_{XY|U}(x,y|u)p_{E|U}(e|u)}{p_E(e)}\nonumber\\
  &=\frac{p_U(u)\left(p_{XY|W}(x,y|u)\cdot\mathds{1}_{\{e=1\}}+p_{XY|V}(x,y|u)\cdot\mathds{1}_{\{e=0\}}\right)}{\lambda\cdot\mathds{1}_{\{e=1\}}+\bar{\lambda}\cdot\mathds{1}_{\{e=0\}}}\label{ph3}\\
  &=p_{XYW}(x,y,u)\cdot\mathds{1}_{\{e=1\}}+p_{XYV}(x,y,u)\cdot\mathds{1}_{\{e=0\}},\ e\in\{0,1\},\label{ph4}
\end{align}
where (\ref{ph3}) and (\ref{ph4}) follow from (\ref{ph2}) and (\ref{ph1}). Therefore,
\begin{align}
    I(X;U)&=I(X;U,E)\label{gyek}\\
    &=I(X;U|E)\label{gdo}\\
    &=\lambda I(X;U|E=1)+\bar{\lambda}I(X;U|E=0)\nonumber\\
    &=\lambda I(X;W)+\bar{\lambda}I(X;V)\label{gse}\\
    &\leq\epsilon,\label{eptight}
\end{align}
where in (\ref{gyek}), $E$ is a deterministic function of $U$, and (\ref{gdo}) results from $X\independent E$. We obtain (\ref{gse}) from (\ref{ph4}). 
% \footnote{Note that the inequality in (\ref{eptight}) is tight, since we have $I(X;W)=I(X;V)=\epsilon$. Nonetheless, since this tightness is not needed in this analysis, (\ref{eptight}) is written in a general form.}

From (\ref{eptight}), we are allowed to write
\begin{align}
    G_\epsilon\left(p_X\cdot(\lambda p'_{Y|X}+\bar{\lambda}p''_{Y|X})\right)&\geq I(Y;U)\nonumber\\
    &=I(Y;U,E)\nonumber\\
    &\geq I(Y;U|E)\nonumber\\
    &=\lambda I(Y;U|E=1)+\bar{\lambda}I(Y;U|E=0)\nonumber\\
    &=\lambda I(Y;W)+\bar{\lambda}I(Y;V)\label{jahrom}\\
    &=\lambda G_\epsilon(p_X\cdot p'_{Y|X})+\bar{\lambda}G_\epsilon(p_X\cdot p''_{Y|X})\label{jahrom2},
\end{align}
where (\ref{jahrom}) follows from (\ref{ph4}). Finally, by noting that $\epsilon$ and $\lambda$ were chosen arbitrarily, (\ref{jahrom2}) proves that for fixed $p_X$, $G_\epsilon(p_X\cdot p_{Y|X})$ is concave in $p_{Y|X}$.

To prove the second part of the claim, we proceed as follows. Fix $p_{Y|X}$ and $\lambda\in(0,1)$. Given two pmfs $p_X'$ and $p_X''$, let $p_X=\lambda p_X' + \Bar{\lambda}p_X''$. Let $p^*_{U|XY}$ be an optimal mapping in the evaluation of $G_0(p_X\cdot p_{Y|X})$, which induces $U^*\sim p_{U^*}$. Let $I^q$ denote $I(Y;U)$ when $(X,Y,U)\sim q\cdot p_{Y|X}\cdot p^*_{U|XY}$. Obviously, $G_0(p_X\cdot p_{Y|X})=I^{p_X}$.

When $(X,Y,U)\sim q\cdot p_{Y|X}\cdot p^*_{U|XY}$, where $q$ is an arbitrary pmf on $\mathcal{X}$, we have
\begin{align}
    p_{U|X}(\cdot|x)&=\sum_{y}p_{Y|X}(y|x)p^*_{U|XY}(\cdot|x,y)\nonumber\\
    &=p_{U^*|X}(\cdot|x)\nonumber\\
    &=p_{U^*}(\cdot),
\end{align}
and hence, $X\independent U$. Therefore, we have by definition
\begin{equation}\label{tozih3}
    I^q\leq G_0(q\cdot p_{Y|X}).
\end{equation}
When $(X,Y,U)\sim q\cdot p_{Y|X}\cdot p^*_{U|XY}$, we have
\begin{align}
    p_{Y|U}(\cdot|\cdot)&=\sum_xq(x)p_{Y|XU}\nonumber\\
    &=\sum_xq(x)\frac{p_{Y|X}(\cdot|x)p^*_{U|XY}(\cdot|x,\cdot)}{p_{U^*}(\cdot)},
\end{align}
which means that the conditional pmf $p_{Y|U}$ derived from $(\lambda p_X'+\bar{\lambda}p_X'')\cdot p_{Y|X}\cdot p^*_{U|XY}$ is $\lambda$ times the conditional pmf $p_{Y|U}$ derived from $p_X'\cdot p_{Y|X}\cdot p^*_{U|XY}$ plus $\Bar{\lambda}$ times the conditional pmf $p_{Y|U}$ derived from $p_X''\cdot p_{Y|X}\cdot p^*_{U|XY}$.
Hence, we can write
\begin{align}
    G_0((\lambda p_X'+\bar{\lambda}p_X'')\cdot p_{Y|X})&=I^{\lambda p_X'+\bar{\lambda}p_X''}\nonumber\\
    &\leq\lambda I^{p_X'}+\Bar{\lambda}I^{p_X''}\label{tozih1}\\
    &\leq \lambda G_0(p_X'\cdot p_{Y|X})+\Bar{\lambda}G_0(p_X''\cdot p_{Y|X})\label{tozih2},
\end{align}
where (\ref{tozih1}) results from the convexity of $I(Y;U)$ in $p_{Y|U}$ for fixed $p_U$, and (\ref{tozih2}) follows from (\ref{tozih3}).

Finally, the above analysis for proving the convexity of $G_0(X,Y)$ remains valid if $X-Y-U$ form a Markov chain, and we replace $p^*_{U|XY}$ with $p^*_{U|Y}$, which is an optimal mapping in the evaluation of $g_0(p_X\cdot p_{Y|X})$. Therefore, for fixed $p_{Y|X}$, $g_0(X,Y)$ is also convex in $p_X$.
\end{proof}
\begin{example}
Let $(X,Y)\in\{x_1,x_2\}\times\{y_1,y_2\}$, in which $p\triangleq p(x_1)$, and the transition from $X$ to $Y$ follows a general \textit{binary asymmetric channel} (BAC) with cross over probabilities $\alpha,\beta \in[0,1]$, i.e., $\alpha\triangleq p(y_2|x_1)$ and $\beta\triangleq p(y_1|x_2)$. Therefore, we have $q\triangleq 
 p(y_1)=p\bar{\alpha}+\bar{p}\beta$. From Theorem \ref{th2}, we have that
\begin{align}
    G_0(X,Y)&=H_b(q)-\left(1-\min\{\alpha,\bar{\beta}\}-\min\{\beta,\bar{\alpha}\}\right)H_b(p),\nonumber\\
    &=H_b(q)-|\alpha-\bar{\beta}|H_b(p).
\end{align}
It is already known that for a given $p$, $G_0(X,Y)$ is concave in $(\alpha,\beta)$, and the maximizer is $(\alpha^*,\beta^*)=(\frac{1}{2},\frac{1}{2})$, which results in $X\independent Y$, and $H(Y)=1$.

For a given $(\alpha,\beta)$, $G_0(X,Y)$ is convex in $p$. Therefore, by setting $\frac{d}{dp}G_0(X,Y)=0$, we solve for $p^*$ as
% in 
% \begin{equation}\label{pstar}
%     \left(\frac{1}{p^*(1-\alpha-\beta)-\beta}-1\right)^{(1-\alpha-\beta)}=\left(\frac{1}{p^*}-1\right)^{\left(1-\min\{\alpha,1-\alpha\}-\min\{\beta,1-\beta\}\right)}.
% \end{equation}
% It can be verified that if $\alpha,\beta<\frac{1}{2}$, or $\alpha,\beta>\frac{1}{2}$, (\ref{pstar}) has a closed-form solution\footnote{Note tht for $\alpha=\beta=\frac{1}{2}$, we have $G_0(X,Y)=1$, irrespective of the value of $p$.} 
as
\begin{align*}
    p^*&=\frac{\beta}{\alpha+\beta}\mathds{1}_{\{\alpha<\Bar{\beta}\}}+\frac{\bar{\beta}}{\bar{\alpha}+\bar{\beta}}\mathds{1}_{\{\alpha>\Bar{\beta}\}},
\end{align*}
and when $\alpha = \bar{\beta}$, $p^*$ is arbitrary, since we get $G_0(X,Y)=H_b(\beta)=H_b(\alpha)$ irrespective of the value of $p$ due to the independence of $X$ and $Y$.

In the special case of $\alpha=\beta$, i.e., if the transition from $X$ to $Y$ is a binary symmetric channel (BSC($\alpha$)), we have that $p^*=\frac{1}{2}$, if $\alpha\neq\frac{1}{2}$, and any number in $[0,1]$ if $\alpha=\frac{1}{2}$. Therefore, we get $G_0(X,Y)\geq2\min\{\alpha,\bar{\alpha}\}$. Taking Corollary \ref{cor3} into account, this means that if for a pair $(X,Y)$ whose $p_{Y|X}$ is BSC($\alpha$), the curator is unaware of $p_X$, he can still obtain the optimal mapping and make sure that the utility of this release is at least $2\min\{\alpha,\bar{\alpha}\}$.

As another case, let $(X,Y)\in\{x_1,x_2\}\times\{y_0,y_1,y_2\}$, in which $p(x_1)\triangleq p$, and the transition from $X$ to $Y$ is a binary eraser channel with erasure probability of $e$, i.e., $\textnormal{BEC}(e)$, in which $p(y_i|x_j)=e\mathds{1}_{\{i=0\}}+\bar{e}\mathds{1}_{\{i=j\}},\ (i,j)\in[0:2]\times[2]$. Denote the entropy of a ternary random variable with mass probabilities $p_1,p_2,p_3$ by $H(p_1,p_2,p_3)$. From Theorem \ref{th2}, we have
\begin{align*}
    G_0(X,Y)&=H(Y)-(1-\sum_y\min_xp(y|x))H(X)\\
    &=H(p\bar{e},e,\bar{p}\bar{e})-\bar{e}H_b(p)\\
    &=H_b(e),
\end{align*}
and $U\triangleq \mathds{1}_{\{Y=y_0\}}$ attains it. Moreover, since this mapping is only from $Y$ to $U$, we conclude that $g_0(X,Y)=H_b(e)$. By changing the role of $X$ and $Y$, and applying the second part of Theorem \ref{th3}, we get
\begin{equation*}
    G_0(Y,X) = H_b(p)-pH_b(\Bar{p}\Bar{e})-\Bar{p}H_b(p\Bar{e}),
\end{equation*}
where $U\triangleq X\cdot\mathds{1}_{\{Y=y_0\}}+\tilde{X}\cdot\mathds{1}_{\{Y\neq y_0\}}$, in which $\Tilde{X}\in\mathcal{X}$ is a Bernoulli random variable independent of $(X,Y)$ with $p_{\tilde{X}}(x_1)=p$, achieves it. We also have $g_0(Y,X)=0$ from \cite[Corollary 2]{RG-JSAIT} by noting that the nullity of $\mathbf{P}_{Y|X}$ is zero.\footnote{With some back of the envelope calculations, it can be verified that for an $M$-ary erasure channel ($M\geq 2$), in which $(X,Y)\in\{x_1,\ldots,x_M\}\times\{y_0,y_1,\ldots,y_M\}$ and $p(y_i|x_j)=e\mathds{1}_{\{i=0\}}+\bar{e}\mathds{1}_{\{i=j\}},\ (i,j)\in[0:M]\times[M]$, we have $g_0(X,Y)=G_0(X,Y)=H(Y|X)=H_b(e)$ attained by $U\triangleq \mathds{1}_{\{Y=y_0\}}$. Furthermore, denoting the probability vector associated with $p_X$ by $\mathbf{p}$, we have $g_0(Y,X)=0$ and $G_0(Y,X)=H(X)-\sum_{i=1}^Mp_iH(\bar{e}\mathbf{p}+e\mathbf{1}_i)$ attained by $U\triangleq X\cdot\mathds{1}_{\{Y=y_0\}}+\tilde{X}\cdot\mathds{1}_{\{Y\neq y_0\}}$, where $\mathbf{1}_i$ is the $i$-th standard unit vector, and $\tilde{X}\independent (X,Y)$ is distributed according to $p_X$.}
\end{example}
% \begin{corollary}\label{cor55}
% If $X$ is binary, we have 
% \begin{equation*}
%   g_0(X,Y)\leq H(Y)-\left(1-\sum_y\min_xp(y|x)\right)H(X),
% \end{equation*}
% which is tighter than the previously known upper bound in literature, i.e., $H(Y|X)$.
% \end{corollary}
% \begin{proof}
% The above bound follows (\ref{G0eq}), and the fact that $g_0(X,Y)\leq G_0(X,Y)$ by definition. Also, the fact that this bound is tighter than $H(Y|X)$ is immediate from having $G_0(X,Y)\leq H(Y|X)$ as in (\ref{G0up}). \footnote{Alternatively, the latter can be proved by strong data processing inequality as follows.
% \begin{align}\label{tritri}
%     \frac{I(X;Y)}{H(X)}&\leq s^*(X;Y)\leq\max_{p_X}s^*(X;Y)\leq 1-2^{-\mathcal{L}^c(X\to Y)},
% \end{align}
% where the first inequality is immediate by noting that $Z\triangleq X$ is a permissible point in the search space in the definition of $s^*(X;Y)$, which results in the LHS of (\ref{tritri}), and the last inequality follows from \cite[Corollary 6]{Issa}. This results in $H(Y)-\left(1-2^{-\mathcal{L}^c(X\to Y)}\right)H(X)\leq H(Y|X)$.}
% \end{proof}

% \begin{remark}
% (On the tightness of the lower bound in Theorem \ref{thbound5}) We observed that the second term in the RHS of (\ref{G0bound}) is tight if $X$ is a deterministic function of $Y$. In what follows, it is shown that the first term in the RHS of (\ref{G0bound}) can be within a margin of 1 bit from $G_0(X,Y)$.
% \end{remark}
}

Thus far, we have presented an achievable scheme in Algorithm 1 as a lower bound on $G_0(X,Y)$.
Based on this scheme, we proceed to present a privacy-preserving algorithm as a lower bound on $G_\epsilon(X,Y)$.  On the U-P plane, the privacy restriction becomes stricter as we move from right to left. The rightmost point $(I(X;Y),H(Y))$ is achieved when there is no constraint on privacy, while the leftmost point $(0,G_0(X,Y))$ relates to the strictest privacy restriction, which is statistical independence between the private and the released data. Therefore, it makes sense to obtain achievable points on the U-P plane starting from no privacy constraint and increasing the restrictions incrementally until we reach the requirement $X\independent U$. One approach is as follows. Let $R\triangleq (I(X;Y),H(Y))$ denote the rightmost point on the U-P trade-off curve. In order to obtain an achievable point $P_1$, we impose the requirement that $X$ must be at least 1-independent of $U$, i.e., $X\stackrel{1}\independent U$. In other words, we require that $X$ must have at least one realization, call it $x_1$, whose posterior probability $p(x_1|u)$ is the same as the prior $p(x_1)$ for any realization $u$ of $U$. Letting $0$ not be an element of $\mathcal{X}$, define an auxiliary random variable $Z_1\triangleq X\cdot\mathds{1}_{\{X=x_1\}}$. The constraint of having $X$ at least 1-independent of the released data can be satisfied by designing a privacy-preserving scheme via Algorithm 1 for the new pair $(Z_1,Y)$, i.e., $p_{U_1|Z_1Y}$, which guarantees $Z_1\independent U_1$, or equivalently $X\stackrel{1}\independent U_1$. The mapping $p_{U_1|XY}$, which is induced by $p_{U_1|Z_1Y}$, results in the achievable point $P_1\triangleq(I(X;U_1),I(Y;U_1))$. To get $P_2$, we require that $X$ must be at least 2-independent of the released data. To this end select arbitrary $x_1,x_2\in\mathcal{X}$, and define $Z_2\triangleq X\cdot\mathds{1}_{\{X\in\{x_1,x_2\}\}}$, and obtain a mapping $p_{U_2|Z_2Y}$ via Algorithm 1, which satisfies $X\stackrel{2}\independent U_2$. The corresponding $p_{U_2|XY}$ produces $P_2\triangleq(I(X;U_2),I(Y;U_2))$. This procedure continues providing achievable points until we reach the constraint $X\independent U$, which is taken care of by applying Algorithm 1 to the pair $(X,Y)$. Finally, the upper concave envelope of the set of these achievable points results in a lower bound on $G_\epsilon(X,Y)$, which is formally presented in the following Proposition..
\begin{proposition}\label{Prop1}(\textbf{Privacy-preserving mapping - a lower bound on $G_\epsilon(X,Y)$}) Let $(X,Y)\sim p_{XY}$ be given. Let $k\in[|\mathcal{X}|-2]$ and $\mathcal{X}'$ be an arbitrary subset of $\mathcal{X}$ with size $k$. Without loss of generality, assume that $0\not\in\mathcal{X}$, and let $Z$ be a function of $X$ defined as
\begin{equation}\label{zkdef}
    Z\triangleq X\cdot\mathds{1}_{\{X\in\mathcal{X}'\}}.
\end{equation}
Applying the achievable scheme in Theorem \ref{thbound5}, i.e., Algorithm 1, to the pair $(Z,Y)$ results in a mapping $p_{U|ZY}$, such that $Z\independent U$, or equivalently, $X\stackrel{k}\independent U$. Calculate $p_{U|XY}$ from $p_{U|ZY}$, and set $P_{\mathcal{X}'}\triangleq(I(X;U),I(Y;U))$. Apply Algorithm 1 to $(X,Y)$ and obtain an achievable point denoted by $L$. Let $\mathcal{P}\triangleq \cup_{\mathcal{X}'\subset\mathcal{X}}\{P_{\mathcal{X}'}\}\cup\{L,R\}$ denote the set of all the achievable points obtained so far. Finally, we have the U-P trade-off $uce_{[\mathcal{P}]}(\epsilon)$ as a lower bound on $G_\epsilon(X,Y),\ \forall\epsilon\in[0,I(X;Y)]$.
\end{proposition}
For the tuple $(Z,X,Y,U)$ in Proposition \ref{Prop1}, it can be readily verified that 
\begin{align}
    p_{Z}(z)&=p_{X}(z)\cdot\mathds{1}_{\{z\neq 0\}}+(\sum_{x\in\mathcal{X}\backslash\mathcal{X}'}p_X(x))\cdot\mathds{1}_{\{z= 0\}}\nonumber\\
    p_{Y|Z}(y|z)&=p_{Y|X}(y|z)\cdot\mathds{1}_{\{z\neq 0\}}+\frac{\sum_{x\in\mathcal{X}\backslash\mathcal{X}'} p_{XY}(x,y)}{\sum_{x\in\mathcal{X}\backslash\mathcal{X}'}p_X(x)}\cdot\mathds{1}_{\{z= 0\}},\label{zcond}\\
    p_{U|XY}(u|x,y)&=p_{U|ZY}(u|x,y)\cdot\mathds{1}_{\{x\in\mathcal{X}'\}}+p_{U|ZY}(u|0,y)\cdot\mathds{1}_{\{x\not\in\mathcal{X}'\}}.\nonumber
\end{align}
% and applying Algorithm 1 in Theorem \ref{thbound5} to $(Z,Y)$ provides a mapping $p_{U|ZY}$ with utility lower bounded as
% \begin{equation}\label{LLBB}
%     I(Y;U)\geq \bigg(H(Y)-\left(1-\sum_y\min_{z}p_{Y|Z}(y|z)\right)\min\{H(Z),\log|\mathcal{Y}|\}\bigg)^+,
% \end{equation}
% which is tight when $|\mathcal{X}'|=1$, i.e., when $Z$ is binary. 
The quantity $I(Y;U)$ is obtained after $p_{U|ZY}$ is obtained in the algorithm (unless $|\mathcal{X}'|=1$, for which Theorem \ref{th2} gives a closed-form solution), while $I(X;U)$ can be obtained prior to the algorithm as 
\begin{align}
    I(X;U)&=I(X,Z;U)\label{qe1}\\
    &=I(Z;U)+I(X;U|Z)\nonumber\\
    &=I(X;U|Z)\label{qe2}\\
    &=I(X;U,Y|Z)\label{qe3}
\end{align}
\begin{align}
    &=I(X;Y|Z)+I(X;U|Z,Y)\nonumber\\
    &=I(X;Y|Z)\label{qe4}\\
    &= I(X;Y)-I(Z;Y)\label{qe5},
\end{align}
where (\ref{qe1}) follows from defining $Z$ as a function of $X$ in (\ref{zkdef}), and (\ref{qe2}) results from $Z\independent U$ in Algorithm 1. The equality in (\ref{qe3}) follows from satisfying the conditions of Lemma \ref{lemcardi} in Algorithm 1. In other words, conditioned on the event $\{Z=z\}$, $Y$ is a function of $U$, or equivalently $H(Y|Z,U)=0$. The equality in (\ref{qe4}) results from the Markov chain $X-(Z,Y)-U$, since $U$ is generated by applying Algorithm 1 to the pair $(Z,Y)$. Finally, (\ref{qe5}) follows from the fact that $Z-X-Y$ form a Markov chain.

% Setting $P_0\triangleq(I(X;Y),H(Y))$, and $P_k\triangleq(I(X;U_k),I(Y;U_k))=(I(X;Y|Z_k),I(Y;U_k)),\ k\in[|\mathcal{X}|-1]$, we obtain $|\mathcal{X}|$ points on the utility-privacy plane, whose upper concave envelope, i.e., $\mathds{U}_{[\{P_k\}_{k=0}^{|\mathcal{X}|-1}]}(\cdot)$ provides the achievable utility-privacy trade-off as a lower bound on $G_\epsilon(X,Y)$. 
% Finally, since the trade-off is a concave function, we take the upper concave envelope of these points to obtain a lower bound on $G_\epsilon(X,Y)$.

The procedure in Proposition \ref{Prop1} is computationally complex when $|\mathcal{X}|$ is large, since there are $2^{|\mathcal{X}|}-|\mathcal{X}|-2$ nonempty subsets of $\mathcal{X}$ with size at most $|\mathcal{X}|-2$. Therefore, for large $|\mathcal{X}|$, we can restrict the analysis to a fixed collection of subsets denoted by $\{\mathcal{X}_k\}$ for $k\in[|\mathcal{X}|-2]$, in which $|\mathcal{X}_k|=k$ and $\mathcal{X}_{j}\subset\mathcal{X}_k$ if $j\leq k$.
% for an arbitrary sequence of the subsets of $\mathcal{X}$, i.e., $\{\mathcal{X}_k\}_{k=1}^{|\mathcal{X}|-1}$ with $|\mathcal{X}_k|=k$. {\color{red}Therefore, a questions arises on how to select this initial sequence. Although this selection could be arbitrary, the following heuristic methods are proposed.}
Let $(Z_k,U_k)$ be the same as $(Z,U)$ in Proposition \ref{Prop1} when $\mathcal{X}'=\mathcal{X}_k,\ k\in[|\mathcal{X}|-2].$ On the U-P plane, the slope of the line connecting $R$ ($=(I(X;Y),H(Y))$) to $(I(X;U_k),I(Y;U_k))$ is
\begin{align}
    m_k&\triangleq\frac{H(Y)-I(Y;U_k)}{I(X;Y)-I(X;U_k)}\nonumber\\
    &=\frac{H(Y)-I(Y;U_k)}{I(Z_k;Y)}\label{zzk}\\
    &\leq\frac{\min\bigg\{H(Y),\left(1-\sum_y\min_{z}p_{Y|Z_k}(y|z)\right)\min\{H(Z_k),\log|\mathcal{Y}|\}\bigg\}}{I(Z_k;Y)}\label{bandebala},
\end{align}
where (\ref{zzk}) and (\ref{bandebala}) follow from (\ref{qe5}) and (\ref{G0lowerbound}), respectively. We denote the upper bound in (\ref{bandebala}) by $f(p_{XY},\mathcal{X}_k)$.

Knowing that $G_\epsilon(X,Y)$ is a concave and non-decreasing curve, a heuristic approach is to select $\mathcal{X}_1$ such that $m_1$ is minimized, and an even simpler approach would be to minimize the upper bound, i.e., $f(p_{XY},\mathcal{X}_1)$. Therefore, we set 

\begin{equation}\label{ordering0}
   x_k\triangleq \argmin_{\substack{x\in\mathcal{X}\backslash\mathcal{X}_{k-1}:\\ \mathcal{\mathcal{X}}_{k}=\mathcal{X}_{k-1}\cup\{x\}}}f(p_{XY},\mathcal{X}_k),\ \  \forall k\in[|\mathcal{X}|-2],\ \mathcal{X}_0\triangleq\emptyset.
\end{equation}
The procedures of this achievable scheme are provided in Algorithm \ref{algorithm3}.
% \begin{align}
%     I(Y;U_1)&=H(Y)-(1-\sum_y\min_{z}p_{Y|Z_1}(y|z))H(Z_1),\\
%     I(Y;U_{|\mathcal{X}|-1})&\geq(H(Y)-(1-\sum_y\min_{x}p_{Y|X}(y|x))H(X))^+\\
%     I(X;U_{|\mathcal{X}|-1})&=0.
% \end{align}
\begin{algorithm}
\caption{Privacy-preserving mapping (simplified version of Proposition \ref{Prop1})}\label{algorithm3}
\begin{algorithmic}[1]
{\color{black}\Function{Algorithm2}{$p_{X,Y}$}
\State Select the collection $\{\mathcal{X}_k\},\ k\in[|\mathcal{X}|-2]$, according to (\ref{ordering0}).
\State $p_{U|XY}\triangleq\textnormal{Algorithm1}(p_{Y|X})$ 
\State $\mathcal{P}\triangleq\{(0,I(Y;U)),(I(X;Y),H(Y))\}$
\State $Z_0\triangleq 0$ 
\State $k = 1$
\While{$I(Z_{k-1};Y)\neq I(X;Y)$}
\State $Z_k\triangleq X\cdot\mathds{1}_{\{X\in\mathcal{X}_k\}}$
\State $p_{U_k|Z_kY}\triangleq\textnormal{Algorithm1}(p_{Y|Z_k})$ 
\State $\mathcal{P}\triangleq\mathcal{P}\cup\{(I(X;U_k),I(Y;U_k))\}$
\State $k = k + 1$
\EndWhile
\State \Return $uce_{[\mathcal{P}]}(\cdot)$
\EndFunction }
\end{algorithmic}
\end{algorithm}

\begin{remark}\label{nonrem}(\textbf{Non-algorithmic U-P trade-off})
The achievable points in Proposition \ref{Prop1} are obtained after applying Algorithm 1 to each constructed pair $(Z,Y)$. More specifically, it is the $y$ coordinate of these points that are obtained after the application of the algorithm, since the $x$ coordinates are already known prior to the algorithm as in (\ref{zcond}). If we replace these $y$ coordinates with their corresponding lower bounds according to (\ref{G0lowerbound}), we obtain a new set of achievable points. Obviously, these points lie below the initial set of points, but they are obtained without the need for the algorithm. Therefore, preserving (\ref{zkdef}) and its preceding assumptions in Proposition \ref{Prop1}, we set
\begin{align}
    \Tilde{P}_{\mathcal{X}'}&\triangleq \bigg(I(X;Y)-I(X;Z)\ ,\ \left(H(Y)-\left(1-\sum_y\min_{z}p_{Y|Z}(y|z)\right)\min\{H(Z),\log|\mathcal{Y}|\}\right)^+\bigg)\\
    \tilde{L}&\triangleq \bigg(0\ ,\ \left(H(Y)-\left(1-\sum_y\min_{x}p_{Y|X}(y|x)\right)\min\{H(X),\log|\mathcal{Y}|\}\right)^+\bigg),
\end{align}
and $\Tilde{\mathcal{P}}\triangleq\cup_{\mathcal{X}'\subset\mathcal{X}}\{\tilde{P}_{\mathcal{X}'}\}\cup\{\tilde{L},R\}$. The U-P trade-off $uce_{[\tilde{\mathcal{P}}]}(\epsilon)$ is a non-algorithmic lower bound on $G_\epsilon(X,Y),\ \forall\epsilon\in[0,I(X;Y)]$. Needless to say that this can also be applied to the simplified scheme (for large $|\mathcal{X}|$) discussed in Algorithm \ref{algorithm3}.
\end{remark}
\section{Public data observation}
{\color{black}
In this section, we assume that the curator has access to only $Y$, and propose an achievable scheme, i.e., a lower bound on $g_\epsilon(X,Y)$, defined in (\ref{def}). We start with $\epsilon = 0$, i.e., $X\independent U$. An algorithm is proposed (Algorithm 3) that provides a lower bound on $g_0(X,Y)$. Afterwards, this algorithm is used to generate a privacy-preserving mapping, which results in a lower bound on $g_\epsilon(X,Y),\ \epsilon\in[0:I(X;Y)]$.

Like the previous section, we start with a simple theoretical result. 
\begin{lemma}\label{lemcardi2}(\cite[Theorem 1]{RG-JSAIT})
For an optimal mapping $p_{U^*|Y}$ in the evaluation of $g_0(X,Y)$, we have
\begin{equation}
    |\{y\in\mathcal{Y}|p(y|u^*)>0\}|\leq\textnormal{rank}(\mathbf{P}_{X|Y}),\ \forall u^*\in\mathcal{U}^*,
\end{equation}
where $\mathbf{P}_{X|Y}$ is an $|\mathcal{X}|\times|\mathcal{Y}|$ matrix with $(i,j)$-th entry equal to $p_{X|Y}(x_i|y_j)$.
\end{lemma}
Lemma \ref{lemcardi2} implies that in the evaluation of $g_0(X,Y)$, if $X$ is binary, for any $u^*\in\mathcal{U}^*$ (corresponding to an/the optimal solution), there exist at most two realizations of $Y$, denoted by $y_1,y_2\in\mathcal{Y}$, such that  $p(y_1|u^*),p(y_2|u^*)>0$. It is also obvious that if there exists only one $y_0\in\mathcal{Y}$, such that $p(y_0|u^*)>0$ (and hence, $p(y_0|u^*)=1$), the condition $X\independent U$ indicates that this $y_0$ must satisfy $p_{X|Y}(\cdot|y_0)=p_X(\cdot)$. In other words, if there exists no such $y_0$ satisfying $p_{X|Y}(\cdot|y_0)=p_X(\cdot)$, we must have $|\{y\in\mathcal{Y}|p(y|u^*)>0\}|=2, \forall u^*\in\mathcal{U}^*$ for binary $X$. Therefore, in the achievable scheme, it makes sense to build a mapping $p_{U|Y}$, such that its corresponding $p_{Y|U}$ is in line with this observation. To this end, we start with the backward model, i.e., $p_{Y|U}$, by imposing that i) for all the realizations $u$ of $U$, the condition in lemma \ref{lemcardi2} must be satisfied, ii) the pmf $p_{Y}$ must be preserved in $p_{Y,U}$. The results are provided in the following Proposition. Throughout this section, we exclude the trivial case of $X\independent Y$.

\begin{proposition}\label{ALG} (\textbf{A lower bound on $g_0(X,Y)$ for binary $X$.})
Let $\mathcal{X}\triangleq\{x_0,x_1\}$. First, if there exists a mass point $\hat{y}\in\mathcal{Y}$, for which $p(x_0|\hat{y})=p(x_0)$, we create a corresponding $u$, such that $p(u|y)=\mathds{1}_{\{y=\hat{y}\}},\ \forall y\in\mathcal{Y}$. The set of all such $\hat{y}$'s is denoted by $\mathcal{B}\triangleq\{y_1,\ldots,y_{|\mathcal{B}|}\}$.\footnote{Needless to say that if $p(x_0|y)\neq p(x_0),\ \forall y\in\mathcal{Y}$, we have $\mathcal{B}=\emptyset$, and $|\mathcal{B}|=0.$ Also, the elements of $\mathcal{Y}$ have been relabeled in accordance with the definition of $\mathcal{B}$.} Therefore, $|\mathcal{B}|$ realizations of $U$ are created according to $p(u_i|y)\triangleq\mathds{1}_{\{y=y_i\}},\ \forall i\in[|\mathcal{B}|],\ \forall y\in\mathcal{Y}$. Furthermore, we have $H(Y|U=u)=0$ and $p(x_0|u)=p(x_0),\ \forall u\in\{u_1,\ldots,u_{|\mathcal{B}|}\}$. 

Next, $\mathcal{Y}\backslash\mathcal{B}$ is considered. Note that there is no element $y$ of this set for which $p(x_0|y)=p(x_0)$. Hence, in line with lemma \ref{lemcardi2}, we create realizations of $U$, each of which connected to exactly two elements of this set. Having in mind that we require to have $p(x_0|u)=p(x_0)$ for any $u\in\mathcal{U}$, we conclude that each of these newly created $u$'s must be connected to two elements $y_0,y_0'\in\mathcal{Y}\backslash\mathcal{B}$ such that $p(x_0)$ can be written as a convex combination of $p(x_0|y_0)$ and $p(x_0|y_0')$. In other words, we must have either $p(x_0|y_0)<p(x_0)<p(x_0|y_0')$ or $p(x_0|y_0)>p(x_0)>p(x_0|y_0')$. In this light, the set $\mathcal{Y}\backslash\mathcal{B}$ is divided into disjoint sets $\mathcal{Y}_0\triangleq\{y\in\mathcal{Y}|p(x_0|y)<p(x_0)\}$, and $\mathcal{Y}_0'\triangleq\{y\in\mathcal{Y}|p(x_0|y)>p(x_0)\}$. The purpose of this division is to make sure that $p(x_0)$ can be written as a convex combination of an arbitrary element of $\mathcal{Y}_0$ and an arbitrary element of $\mathcal{Y}_0'$. Therefore, if we create a mass point (or node) $u$, and connect it to one node in $\mathcal{Y}_0$, and another node in $\mathcal{Y}_0'$, with proper weights, the posterior $p(x_0|u)$ remains the same as the prior $p(x_0)$, which is in accordance with the condition $X\independent U$. This is carried out in an iterative way, where at each iteration $i$, a mass point $u_{i+|\mathcal{B}|}$ is created that is connected only to two mass points of $\mathcal{Y}$, i.e., $y_0$ from $\mathcal{Y}_0$, and $y_0'$ from $\mathcal{Y}_0'$, with proper weights such that $p(x_0|u_{i+|\mathcal{B}|}) = p(x_0)$, i.e., $p(y_0|u_{i+|\mathcal{B}|})=f(y_0,y_0')\triangleq\frac{p(x_0|y_0')-p(x_0)}{p(x_0|y_0')-p(x_0|y_0)}$, and $p(y_0'|u_{i+|\mathcal{B}|})=\bar{f}(y_0,y_0')\triangleq 1-f(y_0,y_0')$. This results in $H(Y|U=u_{i+|\mathcal{B}|})=H_b(f(y_0,y_0'))$. Note that the selection of a pair $(y_0,y_0')\in\mathcal{Y}_0\times\mathcal{Y}_0'$ can be done arbitrarily; however, in order to minimize $H(Y|U)$ heuristically, an asymmetric selection is carried out, i.e., $y_0$ is the point whose corresponding $p(x_0|y_0)$ is the farthest from $p(x_0)$ among the points in $\mathcal{Y}_0$, whereas, $y_0'$ is the point whose corresponding $p(x_0|y_0')$ is the closest to $p(x_0)$ among the points in $\mathcal{Y}_0'$.

In order to preserve the marginal pmf of $Y$ in the resulting pair $(Y,U)$, a water filling approach is utilized, whereby at each iteration $i$, the water levels of $y_0,y_0'$ are updated. More specifically, once $y_0\in\mathcal{Y}_0$, and $y_0'\in\mathcal{Y}_0'$ are selected, the algorithm fills the water level of at least one of them. In the first iteration, the water levels are the mass probabilities $p(y_0)$ and $p(y_0')$. Knowing that $p(y_0|u_{1+|\mathcal{B}|})=f(y_0,y_0')$, and $p(y_0'|u_{1+|\mathcal{B}|})=\bar{f}(y_0,y_0')$, we need to assign a mass probability to $u_{1+|\mathcal{B}|}$ such that at least one of the conditions i) $p(u_{1+|\mathcal{B}|})f(y_0,y_0')=p(y_0), p(u_{1+|\mathcal{B}|})\bar{f}(y_0,y_0')\leq p(y_0')$ or ii) $p(u_{1+|\mathcal{B}|})\bar{f}(y_0,y_0')= p(y_0'), p(u_{1+|\mathcal{B}|})f(y_0,y_0')\leq p(y_0)$ is valid, which is equivalent to having at least one water level filled and the other one not exceeded. This results in the assignment $p(u_{1+|\mathcal{B}|})\triangleq \min\{\frac{p(y_0)}{f(y_0,y_0')},\frac{p(y_0')}{\bar{f}(y_0,y_0')}\}$. Afterwards, the water levels of $y_0$ and $y_0'$ are modified, and the algorithm moves on to the next iteration.

Since at each iteration, at least one water level corresponding to an element of $\mathcal{Y}\backslash\mathcal{B}$ is filled, and at the very last iteration, the remaining two water levels are filled at once\footnote{since otherwise, after one more iteration, we are left with a mass point $y'\in\mathcal{Y}\backslash\mathcal{B}$, such that $p(x_0|y')=p(x_0)$. This is a contradiction, since all such mass points are already contained in $\mathcal{B}$.}, the algorithm terminates after $N$ iterations for some $N\leq |\mathcal{Y}|-|\mathcal{B}|-1$, which results in $|\mathcal{U}|\leq |\mathcal{Y}|-1$. Let $f_{\textnormal{max}}\triangleq \max_{(y_0,y_0')\in\mathcal{Y}_0\times\mathcal{Y}_0'}f(y_0,y_0')$. As mentioned before, $H(Y|U=u)=0,\ \forall u\in\{u_1,\ldots,u_{|\mathcal{B}|}\}$. Moreover, since the conditional pmf of $Y$ given any realization $u\in\mathcal{U}\backslash\{u_1,\ldots,u_{|\mathcal{B}|}\}$ has two mass probabilities, i.e., $f(y_0,y_0'),\bar{f}(y_0,y_0')$ for some $(y_0,y_0')\in\mathcal{Y}_0\times\mathcal{Y}_0'$, we have $H(Y|U=u)\leq H_b(f_{\textnormal{max}}),\ \forall u\in\mathcal{U}\backslash\{u_1,\ldots,u_{|\mathcal{B}|}\}$. As a result, we get $H(Y|U)\leq (1-\sum_{y\in\mathcal{B}}p(y))H_b(f_{\textnormal{max}})$, and $I(Y;U)\geq \left(H(Y)-(1-\sum_{y\in\mathcal{B}}p(y))H_b(f_{\textnormal{max}})\right)^+\geq(H(Y)-1)^+$. The aforementioned procedures are provided in Algorithm \ref{algor0}.

\begin{algorithm}
\caption{A lower bound on $g_0(X,Y)$ for binary $X$.}\label{algor0}
\begin{algorithmic}[1]
{\color{black}\Function{Algorithm3}{$p_{XY}$}
\State $\mathcal{B}\triangleq\{y\in\mathcal{Y}|p(x_0|y)=p(x_0)\}=\{y_1,y_2,\ldots,y_{|\mathcal{B}|}\}$
\State $p(u_i|y)=\mathds{1}_{\{y=y_i\}},\ \forall i\in[|\mathcal{B}|],\ \forall y\in\mathcal{Y}$
\State $\mathcal{Y}_0\triangleq\{y\in\mathcal{Y}|p(x_0|y)<p(x_0)\}$, $\mathcal{Y}_0'\triangleq\{y\in\mathcal{Y}|p(x_0|y)>p(x_0)\}$
\State $f(y,y')\triangleq \frac{p(x_0|y')-p(x_0)}{p(x_0|y')-p(x_0|y)}, \bar{f}(y,y')\triangleq 1- f(y,y'),\ \forall (y,y')\in\mathcal{Y}_0\times\mathcal{Y}_0'$
\State $a_1(y)=p(y),\ \forall y\in\mathcal{Y}\backslash\mathcal{B}$
\State i = 1
\While{$\max_{y}a_i(y)\neq 0$}
\State $y_0=\argmin_{y\in\mathcal{Y}_0}\{p(x_0|y)|a_i(y)>0\}$, $y_0'=\argmin_{y\in\mathcal{Y}_0'}\{p(x_0|y)|a_i(y)>0\}$
\State $p(u_{i+|\mathcal{B}|}|y_0)=\frac{f(y_0,y_0')}{p(y_0)}\min\{\frac{a_i(y_0)}{f(y_0,y_0')},\frac{a_i(y_0')}{\bar{f}(y_0,y_0')}\}$
\State $p(u_{i+|\mathcal{B}|}|y_0')=\frac{\bar{f}(y_0,y_0')}{p(y_0')}\min\{\frac{a_i(y_0)}{f(y_0,y_0')},\frac{a_i(y_0')}{\bar{f}(y_0,y_0')}\}$
\State $p(u_{i+|\mathcal{B}|}|y)=0,\ \forall y\in\mathcal{Y}\backslash\{y_0,y_0'\}$
\State $a_{i+1}(y)= a_i(y)-p(u_{i+|\mathcal{B}|})\left(f(y_0,y_0')\mathds{1}_{\{y=y_0\}}+\bar{f}(y_0,y_0')\mathds{1}_{\{y=y_0'\}}\right),\ \forall y\in\mathcal{Y}\backslash\mathcal{B}$
\State $i = i + 1$
\EndWhile
\State \Return $p_{U|Y}$
\EndFunction }
\end{algorithmic}
\end{algorithm}

\end{proposition}
}
\begin{example}\label{ex.3}
Consider the pair $(X,Y)\in\{x_0,x_1\}\times\{y_1,y_2,y_3,y_4\}$, with
% \begin{equation*}
% \mathbf{P}_{X,Y}=\begin{bmatrix}
% 0.15 & 0.2 & 0.0625 & 0.05 \\ 0.35 & 0.05 & 0.0625 & 0.075
% \end{bmatrix},
% \end{equation*}
% which results in 
\begin{equation}\label{exam}
\mathbf{p}_Y=\begin{bmatrix}
\frac{1}{2}&\frac{1}{4}&\frac{1}{8}&\frac{1}{8}\end{bmatrix}^T,\ \ \ \mathbf{P}_{X|Y}=\begin{bmatrix}
0.3 & 0.8 & 0.5 & 0.4 \\ 0.7 & 0.2 & 0.5 & 0.6
\end{bmatrix}.
\end{equation}
We have $\mathcal{B}=\emptyset$, since there is no $y\in\mathcal{Y}$ for which $p(x_0|y)=p(x_0)$ ($=0.4625$). We have $\mathcal{Y}_0=\{y_1,y_4\}$, and $\mathcal{Y}_0'=\{y_2,y_3\}$. In step 6 of the algorithm, we have the first water levels as $a_1=[0.5,0.25,0.125,0.125]^T$, which is the same as the mass probabilities of $Y$. Figure \ref{fig22:image} provides an illustrative explanation of the iterations in the algorithm, where the probabilities for the pair $(X,Y)$ are according to (\ref{exam}). In the first three subfigures, the transition probabilities are from $U$ to $Y$, while in the last subfigure, it is from $Y$ to $U$. 

In the first iteration, we have $y_0=y_1, y_0'=y_3$ according to step 9. Hence, $u_1$ is created which connects to $y_1,y_3$ with transition probabilities $f(y_1,y_3)=0.1875$, and $\bar{f}(y_1,y_3)=1-f(y_1,y_3)$, respectively. We set $p(u_1)\triangleq\min\{\frac{a_1(y_1)}{f(y_1,y_3)},\frac{a_1(y_3)}{\bar{f}(y_1,y_3)}\}=0.154$ to fill the water level $a_1(y_3)$. 
% and steps 10,11, and 12 result in $p(u_1|1)=\frac{0.1875}{0.5}\min\{\frac{0.5}{0.1875},\frac{0.125}{0.8125}\}=0.0577$, $p(u_1|3)=1$, and $p(u_1|y)=0,\ \forall y\in\{2,4\}$. 
The water levels are updated, and we get $a_2=[0.4712,0.25,0,0.125]^T$ shown on the RHS of $y_i$'s in Figure \ref{fig22:subim2}. In iteration 2, considered in the same figure, we get $y_0=y_1, y_0'=y_2$. Hence, $u_2$ is created which connects to $y_1,y_2$ with transition probabilities $f(y_1,y_2)=0.675$, and $\bar{f}(y_1,y_2)=1-f(y_1,y_2)$, respectively. We set $p(u_2)\triangleq\min\{\frac{a_2(y_1)}{f(y_1,y_2)},\frac{a_2(y_2)}{\bar{f}(y_1,y_2)}\}=0.698$ to fill the water level $a_2(y_1)$. Hence, we get the update $a_3=[0,0.0231,0,0.125]^T$, which is shown in Figure \ref{fig22:subim3}.
% We have $f(1,2)=0.675$, $p(u_2|1)=0.9423$, $p(u_2|2)=0.9075$, $p(u_2|y)=0,\ \forall y\in\{3,4\}$. We get $a_3=[0,0.0231,0,\frac{1}{8}]$. 
In the last iteration, we have $y_0=y_4, y_0'=y_2$. Hence, $u_3$ is created which connects to $y_4,y_2$ with transition probabilities $f(y_4,y_2)=0.845$, and $\bar{f}(y_4,y_2)=1-f(y_4,y_2)$, respectively. We set $p(u_3)\triangleq\min\{\frac{a_3(y_4)}{f(y_4,y_2)},\frac{a_3(y_2)}{\bar{f}(y_4,y_2)}\}=0.148$ to fill the water levels $a_3(y_2)$ and $a_3(y_4)$.
% We have $f(4,2)=0.844$, $p(u_3|4)=1$, $p(u_3|2)=0.0925$, $p(u_3|y)=0,\ \forall y\in\{1,3\}$. 
Finally, we get $a_4=[0,0,0,0]^T$, and the algorithm terminates after 3 iterations. The output of the algorithm, i.e., $p_{U|Y}$, is shown in Figure \ref{fig22:subim4}, which results in a utility of $I(Y;U)=0.9063$ bits.  It is interesting to observe that this $p_{U|Y}$ actually coincides with the optimal solution obtained in \cite[Example 1]{arxv4} via linear programming. Therefore, for the $(X,Y)$ distributed according to (\ref{exam}), we have $g_0(X,Y)=0.9063$.
\begin{figure}

\begin{subfigure}{0.5\textwidth}
\centering
\begin{tikzpicture}[scale=1, transform shape]
  \node[dspnodeopen, minimum width=4pt,  dsp/label=above] (Y_1) {$y_1:0.5$};    
  \node[dspnodeopen, minimum width=4pt, below=1cm of Y_1, dsp/label=above] (Y_2) {$y_2:0.25$};
  \node[dspnodeopen, minimum width=4pt, below=1cm of Y_2, dsp/label=below] (Y_3) {$y_3:0.125$};
  \node[dspnodeopen, minimum width=4pt, below=1cm of Y_3, dsp/label=below] (Y_4) {$y_4:0.125$};

  \node[dspnodeopen, minimum width=4pt, right=2.5cm of Y_1, dsp/label=right] (U_1) {$u_1:0.154$};

  \node[dspnodeopen, minimum width=4pt, left=3.5cm of Y_2, dsp/label=left] (X_1) {};    
  \node[dspnodeopen, minimum width=4pt, below=1cm of X_1, dsp/label=left] (X_2) {};

\draw[line width=1.2pt] (X_1) to node[near end, inner sep=1pt, above,sloped]{\footnotesize{$0.3$}} (Y_1);
\draw[line width=1.2pt] (X_2) to node[near end, inner sep=1pt, above,sloped]{\footnotesize{$0.7$}} (Y_1);
\draw[line width=1.2pt] (X_1) to node[near end, inner sep=1pt, above,sloped]{\footnotesize{$0.8$}} (Y_2);
\draw[line width=1.2pt] (X_2) to node[near end, inner sep=1pt, above,sloped]{\footnotesize{$0.2$}} (Y_2);
\draw[line width=1.2pt] (X_1) to node[near end, inner sep=1pt, above,sloped]{\footnotesize{$0.5$}} (Y_3);
\draw[line width=1.2pt] (X_2) to node[near end, inner sep=1pt, above,sloped]{\footnotesize{$0.5$}} (Y_3);
\draw[line width=1.2pt] (X_1) to node[near end, inner sep=1pt, above,sloped]{\footnotesize{$0.4$}} (Y_4);
\draw[line width=1.2pt] (X_2) to node[near end, inner sep=1pt, above,sloped]{\footnotesize{$0.6$}} (Y_4);

\draw[line width=1.2pt] (Y_1) to node[near end, inner sep=1pt, above,sloped]{\footnotesize{$0.1875$}} (U_1);
\draw[line width=1.2pt] (Y_3) to node[near end, inner sep=1pt, above,sloped]{\footnotesize{$0.8125$}} (U_1);
\end{tikzpicture}\par\bigskip
\caption{Iteration 1.}
\label{fig22:subim1}
\end{subfigure}
\begin{subfigure}{0.5\textwidth}
\centering
\begin{tikzpicture}[scale=1, transform shape]
  \node[dspnodeopen, minimum width=4pt,  dsp/label=above] (Y_1) {$y_1:0.4712$};    
  \node[dspnodeopen, minimum width=4pt, below=1cm of Y_1, dsp/label=above] (Y_2) {$y_2:0.25$};
  \node[dspnodeopen, minimum width=4pt, below=1cm of Y_2, dsp/label=below] (Y_3) {$y_3:0$};
  \node[dspnodeopen, minimum width=4pt, below=1cm of Y_3, dsp/label=below] (Y_4) {$y_4:0.125$};

  \node[dspnodeopen, minimum width=4pt, right=2.5cm of Y_2, dsp/label=right] (U_2) {$u_2:0.698$};

  \node[dspnodeopen, minimum width=4pt, left=3.5cm of Y_2, dsp/label=left] (X_1) {};    
  \node[dspnodeopen, minimum width=4pt, below=1cm of X_1, dsp/label=left] (X_2) {};

\draw[line width=1.2pt] (X_1) to node[near end, inner sep=1pt, above,sloped]{\footnotesize{$0.3$}} (Y_1);
\draw[line width=1.2pt] (X_2) to node[near end, inner sep=1pt, above,sloped]{\footnotesize{$0.7$}} (Y_1);
\draw[line width=1.2pt] (X_1) to node[near end, inner sep=1pt, above,sloped]{\footnotesize{$0.8$}} (Y_2);
\draw[line width=1.2pt] (X_2) to node[near end, inner sep=1pt, above,sloped]{\footnotesize{$0.2$}} (Y_2);
\draw[line width=1.2pt] (X_1) to node[near end, inner sep=1pt, above,sloped]{\footnotesize{$0.5$}} (Y_3);
\draw[line width=1.2pt] (X_2) to node[near end, inner sep=1pt, above,sloped]{\footnotesize{$0.5$}} (Y_3);
\draw[line width=1.2pt] (X_1) to node[near end, inner sep=1pt, above,sloped]{\footnotesize{$0.4$}} (Y_4);
\draw[line width=1.2pt] (X_2) to node[near end, inner sep=1pt, above,sloped]{\footnotesize{$0.6$}} (Y_4);

\draw[line width=1.2pt] (Y_1) to node[near end, inner sep=1pt, above,sloped]{\footnotesize{$0.675$}} (U_2);
\draw[line width=1.2pt] (Y_2) to node[near end, inner sep=1pt, above,sloped]{\footnotesize{$0.325$}} (U_2);
\end{tikzpicture}\par\bigskip
\caption{Iteration 2.}
\label{fig22:subim2}
\end{subfigure}\\
\par\bigskip\par\bigskip
\begin{subfigure}{0.5\textwidth}
\centering
\begin{tikzpicture}[scale=1, transform shape]
  \node[dspnodeopen, minimum width=4pt,  dsp/label=above] (Y_1) {$y_1:0$};    
  \node[dspnodeopen, minimum width=4pt, below=1cm of Y_1, dsp/label=above] (Y_2) {$y_2:0.0231$};
  \node[dspnodeopen, minimum width=4pt, below=1cm of Y_2, dsp/label=below] (Y_3) {$y_3:0$};
  \node[dspnodeopen, minimum width=4pt, below=1cm of Y_3, dsp/label=below] (Y_4) {$y_4:0.125$};

  \node[dspnodeopen, minimum width=4pt, right=2.5cm of Y_3, dsp/label=right] (U_3) {$u_3:0.148$};

  \node[dspnodeopen, minimum width=4pt, left=3.5cm of Y_2, dsp/label=left] (X_1) {};    
  \node[dspnodeopen, minimum width=4pt, below=1cm of X_1, dsp/label=left] (X_2) {};

\draw[line width=1.2pt] (X_1) to node[near end, inner sep=1pt, above,sloped]{\footnotesize{$0.3$}} (Y_1);
\draw[line width=1.2pt] (X_2) to node[near end, inner sep=1pt, above,sloped]{\footnotesize{$0.7$}} (Y_1);
\draw[line width=1.2pt] (X_1) to node[near end, inner sep=1pt, above,sloped]{\footnotesize{$0.8$}} (Y_2);
\draw[line width=1.2pt] (X_2) to node[near end, inner sep=1pt, above,sloped]{\footnotesize{$0.2$}} (Y_2);
\draw[line width=1.2pt] (X_1) to node[near end, inner sep=1pt, above,sloped]{\footnotesize{$0.5$}} (Y_3);
\draw[line width=1.2pt] (X_2) to node[near end, inner sep=1pt, above,sloped]{\footnotesize{$0.5$}} (Y_3);
\draw[line width=1.2pt] (X_1) to node[near end, inner sep=1pt, above,sloped]{\footnotesize{$0.4$}} (Y_4);
\draw[line width=1.2pt] (X_2) to node[near end, inner sep=1pt, above,sloped]{\footnotesize{$0.6$}} (Y_4);

\draw[line width=1.2pt] (Y_2) to node[near end, inner sep=1pt, above,sloped]{\footnotesize{$0.155$}} (U_3);
\draw[line width=1.2pt] (Y_4) to node[near end, inner sep=1pt, above,sloped]{\footnotesize{$0.845$}} (U_3);
\end{tikzpicture}\par\bigskip
\caption{Iteration 3}
\label{fig22:subim3}
\end{subfigure}
\begin{subfigure}{0.5\textwidth}
\centering
\begin{tikzpicture}[scale=1, transform shape]
  \node[dspnodeopen, minimum width=4pt,  dsp/label=above] (Y_1) {$y_1:0.5$};    
  \node[dspnodeopen, minimum width=4pt, below=1cm of Y_1, dsp/label=above] (Y_2) {$y_2:0.25$};
  \node[dspnodeopen, minimum width=4pt, below=1cm of Y_2, dsp/label=below] (Y_3) {$y_3:0.125$};
  \node[dspnodeopen, minimum width=4pt, below=1cm of Y_3, dsp/label=below] (Y_4) {$y_4:0.125$};

  \node[dspnodeopen, minimum width=4pt, right=2.5cm of Y_1, dsp/label=right] (U_1) {$u_1:0.154$};
  \node[dspnodeopen, minimum width=4pt, right=2.5cm of Y_2, dsp/label=right] (U_2) {$u_2:0.698$};
  \node[dspnodeopen, minimum width=4pt, right=2.5cm of Y_3, dsp/label=right] (U_3) {$u_2:0.148$};

  \node[dspnodeopen, minimum width=4pt, left=3.5cm of Y_2, dsp/label=left] (X_1) {};    
  \node[dspnodeopen, minimum width=4pt, below=1cm of X_1, dsp/label=left] (X_2) {};

\draw[line width=1.2pt] (X_1) to node[near end, inner sep=1pt, above,sloped]{\footnotesize{$0.3$}} (Y_1);
\draw[line width=1.2pt] (X_2) to node[near end, inner sep=1pt, above,sloped]{\footnotesize{$0.7$}} (Y_1);
\draw[line width=1.2pt] (X_1) to node[near end, inner sep=1pt, above,sloped]{\footnotesize{$0.8$}} (Y_2);
\draw[line width=1.2pt] (X_2) to node[near end, inner sep=1pt, above,sloped]{\footnotesize{$0.2$}} (Y_2);
\draw[line width=1.2pt] (X_1) to node[near end, inner sep=1pt, above,sloped]{\footnotesize{$0.5$}} (Y_3);
\draw[line width=1.2pt] (X_2) to node[near end, inner sep=1pt, above,sloped]{\footnotesize{$0.5$}} (Y_3);
\draw[line width=1.2pt] (X_1) to node[near end, inner sep=1pt, above,sloped]{\footnotesize{$0.4$}} (Y_4);
\draw[line width=1.2pt] (X_2) to node[near end, inner sep=1pt, above,sloped]{\footnotesize{$0.6$}} (Y_4);

\draw[line width=1.2pt] (U_1) to node[midway, inner sep=1pt, above,sloped]{\footnotesize{$0.0577$}} (Y_1);
\draw[line width=1.2pt] (Y_3) to node[near end, inner sep=1pt, above,sloped]{\footnotesize{$1$}} (U_1);

\draw[line width=1.2pt] (U_2) to node[midway, inner sep=1pt, above,sloped]{\footnotesize{$0.9423$}} (Y_1);
\draw[line width=1.2pt] (U_2) to node[midway, inner sep=1pt, above,sloped]{\footnotesize{$0.9075$}} (Y_2);

\draw[line width=1.2pt] (U_3) to node[midway, inner sep=1pt, above,sloped]{\footnotesize{$0.0925$}} (Y_2);
\draw[line width=1.2pt] (U_3) to node[midway, inner sep=1pt, above,sloped]{\footnotesize{$1$}} (Y_4);
\end{tikzpicture}\par\bigskip
\caption{The output $p_{U|Y}.$}
\label{fig22:subim4}
\end{subfigure}
\caption{Illustration of Example \ref{ex.3}.}
\label{fig22:image}
\end{figure}
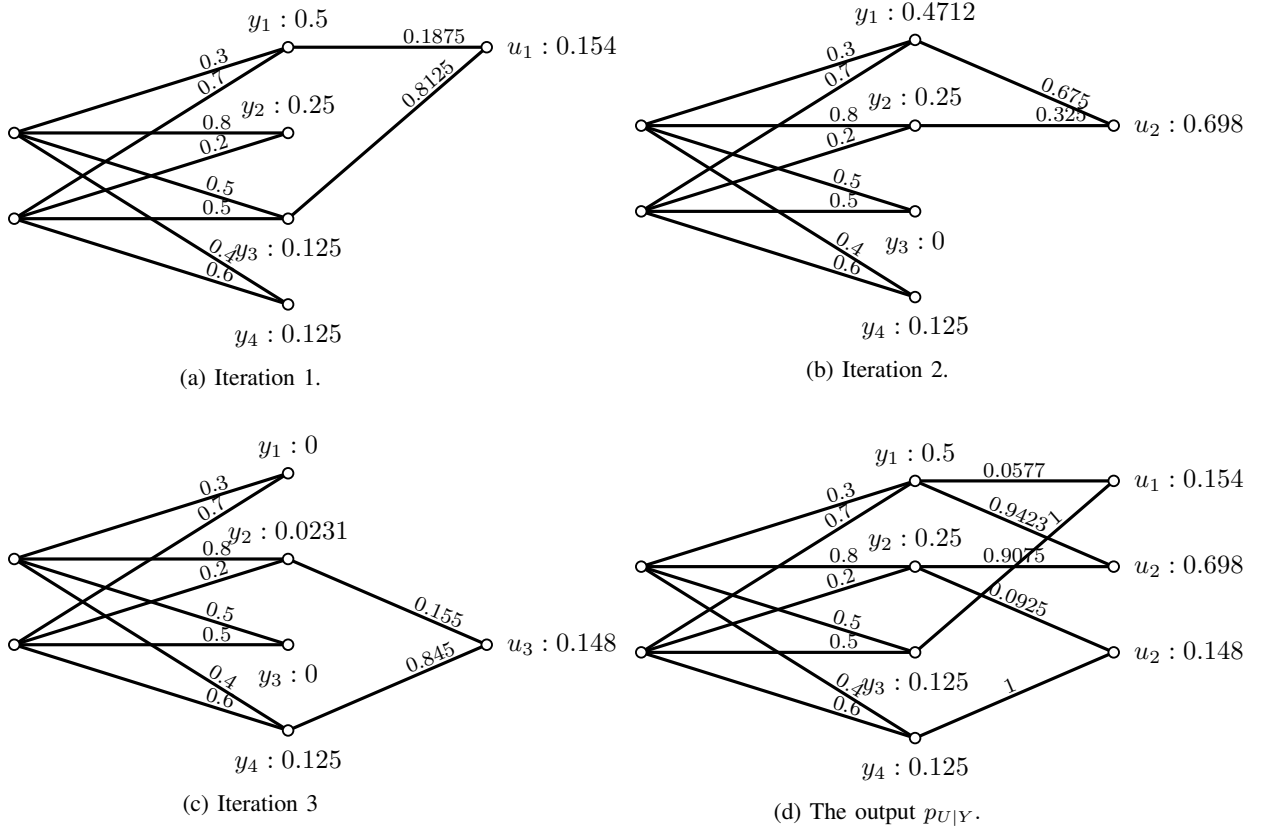
\end{example}
\begin{theorem}
    If $|\mathcal{Y}|=3$, Algorithm 3 is optimal, i.e., it achieves $g_0(X,Y)$.
\end{theorem}
\begin{proof}
    Let $\mathcal{X}\triangleq \{x_0,x_1\}$ and $\mathcal{B}\triangleq\{y\in\mathcal{Y}|p(x_0|y)=p(x_0)\}$. We have either $|\mathcal{B}|=0$ or $|\mathcal{B}|=1$, since otherwise, $X\independent Y$, and $U^*=Y$.
    
    First, assume $|\mathcal{B}|=1$. Therefore, with a proper relabling of the elements in $\mathcal{Y}$, we have $p(x_0|y_1)=p(x_0)$, $p(x_0|y_2)>p(x_0)$, and $P(x_0|y_3)<p(x_0)$. \footnote{That both $p(x_0|y_2)$ and $p(x_0|y_3)$ cannot be lower or greater than $p(x_0)$ is obvious, since otherwise, we get $p(x_0)<p(x_0)$, which is absurd.} For a mapping $p_{U|Y}$ which results in $X\independent U$, define $\mathcal{U}_i\triangleq\{u\in\mathcal{U}|p(y_i|u)>0\}$, hence, $\mathcal{U}=\cup_{i=1}^3\mathcal{U}_i$. For an optimal mapping, we must have $p(y_1|u)=1,\ \forall u\in\mathcal{U}_1$, since otherwise, we have either $|\{y\neq y_1|p(y|u)>0\}|=1$ or $2$, where the former results in $p(x_0|u)\neq p(x_0)$, which violates the condition $X\independent U$, and the latter violates Lemma \ref{lemcardi2}, which states that each $u$ must be connected to at most two realizations of $Y$. As a result $\mathcal{U}_1\cap(\mathcal{U}_2\cup\mathcal{U}_3)=\emptyset$ and $\textnormal{Pr}\{U\in\mathcal{U}_1\}=p(y_1)$. Furthermore, we have $\mathcal{U}_2=\mathcal{U}_3$, since otherwise, we get $X\not\independent U$. For any realization $u\in\mathcal{U}_2$, we must have $p(y_3|u)=1-p(y_2|u)=f(y_3,y_2)$, with $f(\cdot,\cdot)$ defined in Proposition \ref{ALG}, since otherwise, the condition $X\independent U$ is violated. Therefore, we have 
    \begin{align*}
        H(Y|U)&=\sum_{u\in\mathcal{U}_1}H(Y|U=u)+\sum_{u\in\mathcal{U}_2}H(Y|U=u)\\
        &=0+\sum_{u\in\mathcal{U}_2}H_b(f(y_3,y_2))\\
        &=(1-p(y_1))H_b(f(y_3,y_2)),
    \end{align*}
which is obtained via Algorithm 3, and we get $g_0(X,Y)=H_b(p(y_1))$.

Next, assume $|\mathcal{B}|=0$. With a proper relabling of the elements in $\mathcal{Y}$, we have either $p(x_0|y_1)>p(x_0)>p(x_0|y_i),\ i\in\{2,3\}$, or $p(x_0|y_1)<p(x_0)<p(x_0|y_i),\ i\in\{2,3\}$. We only consider the former, as the proof for the latter follows similarly. Let $\mathcal{U}_i,\ i\in[3],$ be defined as before. We have $\mathcal{U}=\mathcal{U}_1$, since otherwise, for any $u\in\mathcal{U}\backslash\mathcal{U}_1$, $p(x_0|u)<p(x_0)$, and hence, $X\not\independent U$. Furthermore, $\mathcal{U}_2\cap\mathcal{U}_3=\emptyset$, since other wise, in conjunction with $\mathcal{U}=\mathcal{U}_1$, there exists $u\in\mathcal{U}$ such that $p(y_i|u)>0,\ \forall i\in[3]$, which violates the condition in Lemma \ref{lemcardi2}. For any $i\in\{2,3\}$ and any realization $u\in\mathcal{U}_i$, we must have $p(y_i|u)=1-p(y_1|u)=f(y_i,y_1)$, since otherwise, $X\not\independent U$. Moreover, since $p(y_i)=\sum_{u\in\mathcal{U}_i}p(u)p(y_i|u)$, we get $\textnormal{Pr}\{U\in\mathcal{U}_i\}=\frac{p(y_i)}{f(y_i,y_1)},\ i\in\{2,3\}$. Finally, we have
\begin{align*}
    H(Y|U)&=\sum_{i=2}^3\sum_{u\in\mathcal{U}_i}H(Y|U=u)\\
    &=\sum_{i=2}^3\frac{p(y_i)}{f(y_i,y_1)}H_b(f(y_i,y_1)),
\end{align*}
which is attained by Algorithm 3.
\end{proof}
Based on the achievable scheme in Proposition \ref{ALG}, we can now propose a privacy-preserving as a lower bound on $g_\epsilon(X,Y)$.
\begin{proposition}\label{prpr}(\textbf{Privacy-preserving mapping - a lower bound on $g_\epsilon(X,Y)$.}) Let $(X,Y)\sim p_{XY}$ be given, and let $S=(x_1,x_2,\ldots,x_{|\mathcal{X}|-1})$ be an arbitrary ordered $(|\mathcal{X}|-1)$-tuple of the elements in $\mathcal{X}$. Set $U_0\triangleq Y$. The algorithm starts off from this point by decreasing the privacy-leakage step by step as follows. In the first step, define the binary random variable $\hat{X}_1\in\{0,1\}$ as $\hat{X}_1\triangleq\mathds{1}_{\{X=x_1\}}$. Since $\hat{X}_1$ is a function of $X$, $\hat{X}_1-X-Y-U_0$ form a Markov chain. Since $\hat{X}_1$ is binary, by applying Algorithm 3 in Proposition \ref{ALG} to the pair $(\hat{X}_1,U_0)$, $p_{U_1|U_0}$ is generated such that $X-Y-U_0-U_1$ form a Markov chain, and $\hat{X}_1\independent U_1$, or equivalently, $p_{X|U_1}(x_1|u)=p_X(x_1),\ \forall u\in\mathcal{U}_1$. Hence, $X$ is at least 1-independent of $U_1$. Set $P_1\triangleq (I(X;U_1),I(Y;U_1))$. The algorithm proceeds in an iterative way as follows. After building the Markov chain $\hat{X}_i-X-Y-U_{i-1}$, $i\in[2:|\mathcal{X}|-1]$, in which $\hat{X}_i\triangleq \mathds{1}_{\{X=x_i\}}$, apply Algorithm 3 to $(\hat{X}_i,U_{i-1})$ to generate $p_{U_i|U_{i-1}}$, such that $X-Y-U_{i-1}-U_i$ form a Markov chain and $\hat{X}_i\independent U_i$, or equivalently, $X$ is at least $i$-independent of $U_i$. Set $P_i\triangleq (I(X;U_i),I(Y;U_i))$, and let $\mathcal{P}_S\triangleq \{P_i\}_{i=1}^{|\mathcal{X}|-1}$ denote the set of achievable points for the ordered tuple $S$ introduced earlier. Finally, let $\mathcal{P}\triangleq \cup_{S}\mathcal{P}_S\cup\{I(X;Y),H(Y)\}$. A lower bound on $g_\epsilon(X,Y)$ is provided by $uce_{[\mathcal{P}]}(\cdot)$.
\end{proposition}
For a fixed tuple $S$ in Proposition \ref{prpr}, we have
\begin{align}
    I(Y;U_i)&=H(Y)-H(Y,U_{i-1}|U_i)+H(U_{i-1}|Y,U_i)\nonumber\\
    &\geq I(Y;U_{i-1})-H(U_{i-1}|U_i)\label{harf1}\\
    &\geq I(Y;U_{i-1})-1,\ \forall i\in[|\mathcal{X}|-1] \label{harf2},
\end{align}
where (\ref{harf1}) follows from the Markov chain $Y-U_{i-1}-U_i$ and non-negativity of entropy, and (\ref{harf2}) results from the fact that according to Algorithm 3, $U_{i-1}$ conditioned on any realization $u_i$ of $U_i$ has at most two non-zero mass probabilities, and hence, $H(U_{i-1}|U_i)\leq 1$. 

We also have
% \begin{align}
%     I(X;U_i)&=\sum_{x\in\mathcal{X}}p_X(x)D\left(p_{U_i|X}(\cdot|x)||p_{U_i}(\cdot)\right)\nonumber\\
%     &=\sum_{x\in\mathcal{X}\backslash\{x_i\}}p_X(x)D\left(p_{U_i|X}(\cdot|x)||p_{U_i}(\cdot)\right)+p_X(x_i)D\left(p_{U_i|X}(\cdot|x_i)||p_{U_i}(\cdot)\right)\nonumber\\
%     &=\sum_{x\in\mathcal{X}\backslash\{x_i\}}p_X(x)D\left(p_{U_i|X}(\cdot|x)||p_{U_i}(\cdot)\right)\label{bv11}\\
%     &\leq \sum_{x\in\mathcal{X}\backslash\{x_i\}}p_X(x)D\left(p_{U_{i-1}|X}(\cdot|x)||p_{U_{i-1}}(\cdot)\right)\label{bv2}\\
%     &=I(X;U_{i-1})-p_X(x_i)D\left(p_{U_{i-1}|X}(\cdot|x_i)||p_{U_{i-1}}(\cdot)\right)\nonumber\\
%     &\leq I(X;Y)-\sum_{j=1}^ip_X(x_j)D\left(p_{Y|X}(\cdot|x_j)||p_{Y}(\cdot)\right)\label{bv3},\ \forall i\in[|\mathcal{X}|-1],
% \end{align}
\begin{align}
    I(X;U_i)&=I(\hat{X}_i,X;U_i)\label{rt1}\\
    &=I(X;U_i|\hat{X}_i)\label{rt2}\\
    &\leq I(X;U_{i-1}|\hat{X}_i)\label{rt3}\\
    &=I(X;U_{i-1})-I(\hat{X}_i;U_{i-1}),\ \forall i\in[|\mathcal{X}|-1]\label{rt4},
\end{align}
where (\ref{rt1}) follows from having defined $\hat{X}_i$ as a function of $X$, (\ref{rt2}) results from $\hat{X}_i\independent U_i$, and finally, (\ref{rt3}) is from the application of data processing inequality in $\hat{X}_i-X-U_{i-1}-U_i$. 

The procedure in Proposition \ref{prpr} can be computationally complex when $|\mathcal{X}|$ is large, as the total number of ordered $(|\mathcal{X}|-1)$-tuples is $|\mathcal{X}|!$. Therefore, this calls for a simpler scheme when $|\mathcal{X}|$ is large. Let $S=(x_1,\ldots,x_{|\mathcal{X}|-1})$ be a given ordered tuple. In the first iteration of the algorithm, a release random variable $U_1$ is generated, which results in a utility greater than or equal to $(H(Y)-1)^+$ from (\ref{harf2}), since $U_0=Y$, and a privacy-leakage lower than or equal to $I(X;Y)-I(\hat{X}_1;Y)$ from (\ref{rt4}). Since $g_\epsilon(X,Y)$ is concave and non-decreasing in $\epsilon$, if a utility within 1 bit of $H(Y)$ is to be achieved in iteration 1, a heuristic approach is to choose an $x_1\in\mathcal{X}$ which is likely to result in the maximum drop in the privacy-leakage, i.e., $I(\hat{X}_1;Y)$. In other words, on the U-P plane, the algorithm tries to depart from the rightmost point $(I(X;Y),H(Y))$ with the lowest slope. As a result, we set $x_1\triangleq \argmax_{x\in\mathcal{X}}I(\mathds{1}_{\{X=x\}};U_0)$, and the algorithm proceeds to provide $U_1$ in $X-Y-U_1$. Following the same rationale, we set $x_2\triangleq \argmax_{x\in\mathcal{X}\backslash\{x_1\}}I(\mathds{1}_{\{X=x\}};U_1)$, and so on. 
% Therefore, for large $|\mathcal{X}|$, we can restrict the analysis to an ordered tuple $S$, in which $x_i$'s are arranged in an iterative way.

% As a result, 

% according to a non-increasing order of $p_X(x)D(p_{Y|X}(\cdot|x)||p_Y(\cdot))$. In other words, setting $\mathcal{D}_0\triangleq\emptyset$, we can write
% \begin{equation}\label{ordering}
%    x_i\triangleq \argmax_{x\in\mathcal{X}\backslash\mathcal{D}_{i-1}}p(x)D(p_{Y|X}(\cdot|x)||p_Y(\cdot)),\ \mathcal{D}_{i}\triangleq\mathcal{D}_{i-1}\cup\{x_i\},\ \ \forall i\in[|\mathcal{X}|-1]. 
% \end{equation}

% Having in mind that the statistical independence between $X$ and $U$ means that $p_{X|U}(x|\cdot)=p_X(x),\ \forall x\in\mathcal{X}$, we observe that this $U_1$ can be viewed as a first step towards statistical independence. 
% In this construction, we have $X-Y-U$ form a Markov chain and a point with utility of at least $(H(Y)-1)^+$, and privacy leakage of at most $I(X;Y)-p(x^*)D\left(p_{Y|X}(\cdot|x^*)||p_Y(\cdot)\right)$ is achievable in the utility-privacy trade-off curve.  the slope at $(0,0)$ is lower bounded by the slope of the straight line connecting this point to the origin.
The procedure of this simplified scheme is presented in Algorithm \ref{algor4}.
\begin{algorithm}
\caption{Privacy-preserving mapping - a lower bound on $g_\epsilon(X,Y).$}\label{algor4}
\begin{algorithmic}[1]
{\color{black}\Function{Algorithm4}{$p_{XY}$}
\State $U_0\triangleq Y$ 
\State $x_1\triangleq \argmax_{x\in\mathcal{X}}I(\mathds{1}_{\{X=x\}};U_0)$
\State $\mathcal{P}\triangleq\{(I(X;U_0),I(Y;U_0))\}$
\State $i = 1$
\While{$I(X;U_{i-1})\neq 0$}
\State $\hat{X}_i\triangleq\mathds{1}_{\{X=x_i\}}$
\State $p_{U_i|U_{i-1}}\triangleq\textnormal{Algorithm3}(p_{\hat{X}_iU_{i-1}})$ 
\State $\mathcal{P}\triangleq\mathcal{P}\cup\{(I(X;U_i),I(Y;U_i))\}$
\State $i = i + 1$
\State $x_i\triangleq \argmax_{x\in\mathcal{X}\backslash\{x_1,\ldots,x_{i-1}\}}I(\mathds{1}_{\{X=x\}};U_{i-1})$
\EndWhile
\State \Return $uce_{[\mathcal{P}]}(\cdot)$
\EndFunction }
\end{algorithmic}
\end{algorithm}

\section{Numerical results}
In this section, the performance of the proposed privacy-preserving schemes are evaluated. Prior to this evaluation, we note that
\begin{align}
    g_0(X,Y)&\geq g_0^L\triangleq\left(H(Y)-\log\textnormal{rank}(\mathbf{P}_{X|Y})\right)^+\label{lbyek}\\
    G_0(X,Y)&\geq G_0^L\triangleq\left(H(Y)-\min\{H(X),\log\textnormal{rank}(\mathbf{P}_{X|Y})\}\right)^+,\label{lbdo}
\end{align}
where $\mathbf{P}_{X|Y}$ is the matrix form of the conditional pmf $p_{X|Y}$. 

The lower bound in (\ref{lbyek}) is from \cite[Corollary 2]{RG-JSAIT} and (\ref{lbdo}) results from i) $G_0(X,Y)\geq g_0(X,Y)$ along with (\ref{lbyek}) and ii) the fact that according to Lemma \ref{lemcardi}, for an optimal mapping $p^*_{U|XY}$, we have $H(Y|X,U^*)=0$, and hence
\begin{align*}
    G_0(X,Y)&=I(Y;U^*)\\
    &=I(X,Y;U^*)-I(X;U^*|Y)\\
    &=I(Y;U^*|X)-I(X;U^*|Y)\\
    &\geq H(Y|X)-H(Y|X,U^*)-H(X|Y)\\
    &=H(Y|X)-H(X|Y)\\
    &=H(Y)-H(X).
\end{align*}
\subsection{Full data observation}
Consider the following joint pmf on $\mathcal{X}\times\mathcal{Y}=[8]^2$ generated randomly\footnote{Each probability vector is obtain by normalizing an 8-dimensional vector whose elements have been i.i.d. generated according to uniform distribution over the interval $[0,1]$.}.
\begin{equation}
    \mathbf{p}_X=  \begin{bmatrix}
    0.175 \cr
    0.089 \cr
    0.146 \cr
    0.026 \cr
    0.077 \cr
    0.167 \cr
    0.145 \cr
    0.175
  \end{bmatrix},\ \  \mathbf{P}_{Y|X}=   \begin{bmatrix}
    0.130 &  0.233 &  0.159 &  0.045 &  0.185 &  0.158 &  0.039 &  0.051 \cr
    0.007 &  0.061 &  0.072 &  0.117 &  0.046 &  0.054 &  0.067 &  0.065 \cr
    0.168 &  0.251 &  0.217 &  0.106 &  0.034 &  0.107 &  0.219 &  0.160 \cr
    0.185 &  0.011 &  0.008 &  0.154 &  0.141 &  0.147 &  0.066 &  0.123 \cr
    0.134 &  0.099 &  0.100 &  0.169 &  0.271 &  0.188 &  0.212 &  0.091 \cr
    0.150 &  0.016 &  0.087 &  0.180 &  0.096 &  0.202 &  0.063 &  0.216 \cr
    0.147 &  0.035 &  0.175 &  0.066 &  0.165 &  0.115 &  0.242 &  0.152 \cr
    0.078 &  0.293 &  0.182 &  0.162 &  0.063 &  0.029 &  0.091 &  0.143
  \end{bmatrix}.
\end{equation}
Figure \ref{fig_a1} illustrates the lower bounds on $G_\epsilon(X,Y)$ in Proposition \ref{Prop1} and Remark \ref{nonrem}. The filled black circles are the achievable points in Proposition \ref{Prop1} whose upper concave envelope is drawn in solid black line. The empty red circles correspond to the non-algorithmic achievable points in Remark \ref{nonrem} whose upper concave envelope is plotted in solid red line. The top dashed blue line is the line connecting $(0,H(Y|X))$ to $(I(X;Y),H(Y))$, while the bottom dashed blue line corresponds to the line connecting $(0,G_0^L)$, where $G_0^L$ is given in (\ref{lbdo}), to $(I(X;Y),H(Y))$.
\begin{figure}[ht]
    \centering
    \includegraphics{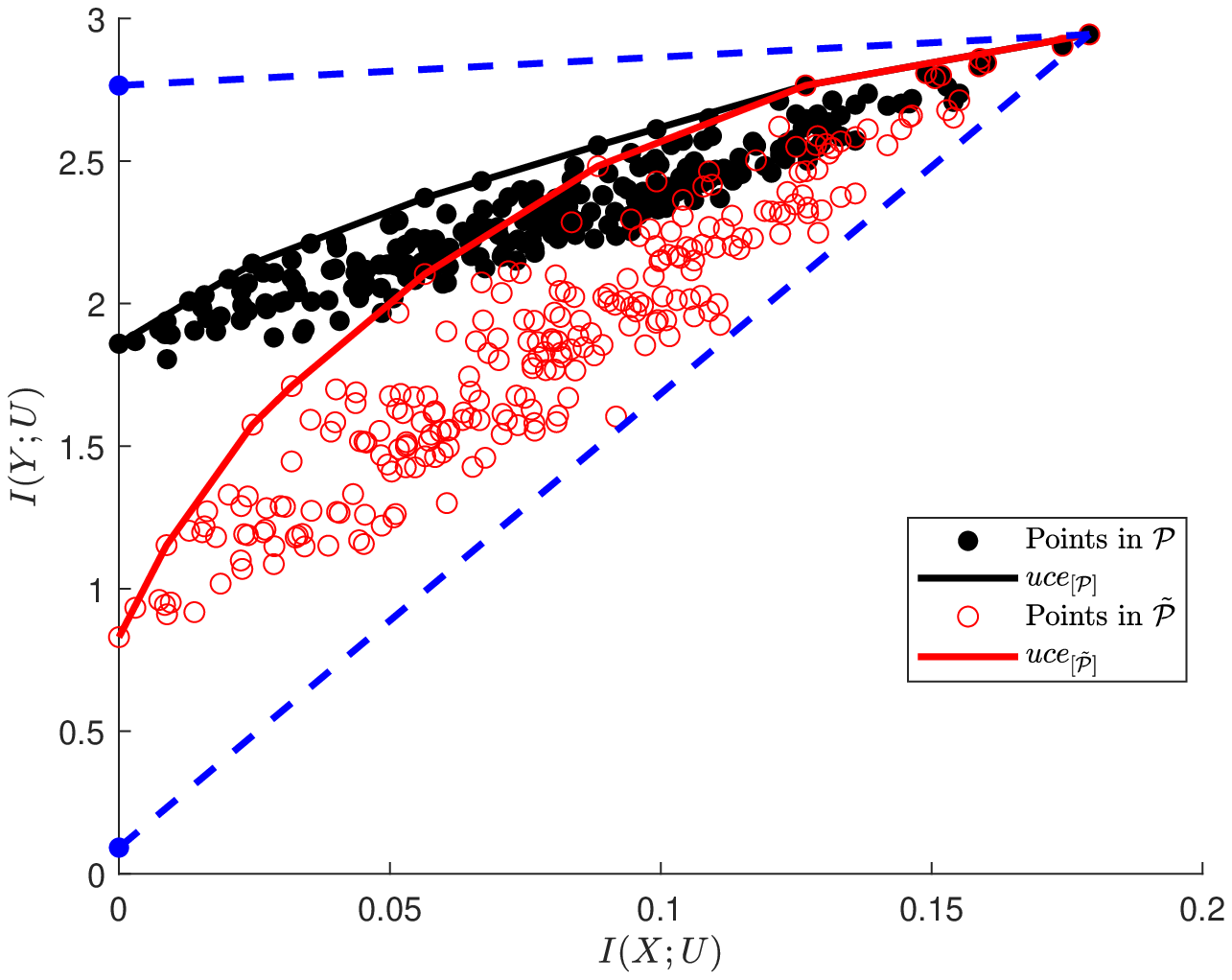}
    \caption{Lower bounds on $G_\epsilon(X,Y)$ in Proposition \ref{Prop1} and Remark \ref{nonrem}.}
    \label{fig_a1}
\end{figure}
Figure \ref{fig_a2} illustrates the lower bounds on $G_\epsilon(X,Y)$ in the simplified version of Proposition \ref{Prop1}, i.e., Algorithm \ref{algorithm3}. The filled black circles are the achievable points in Algorithm \ref{algorithm3} whose upper concave envelope is drawn in solid black line. The empty red circles correspond to the non-algorithmic achievable points in Remark \ref{nonrem} when applied to the procedure in Algorithm \ref{algorithm3} whose upper concave envelope is plotted in solid red line. The dashed blue lines are as mentioned earlier.
\begin{figure}[ht]
    \centering
    \includegraphics{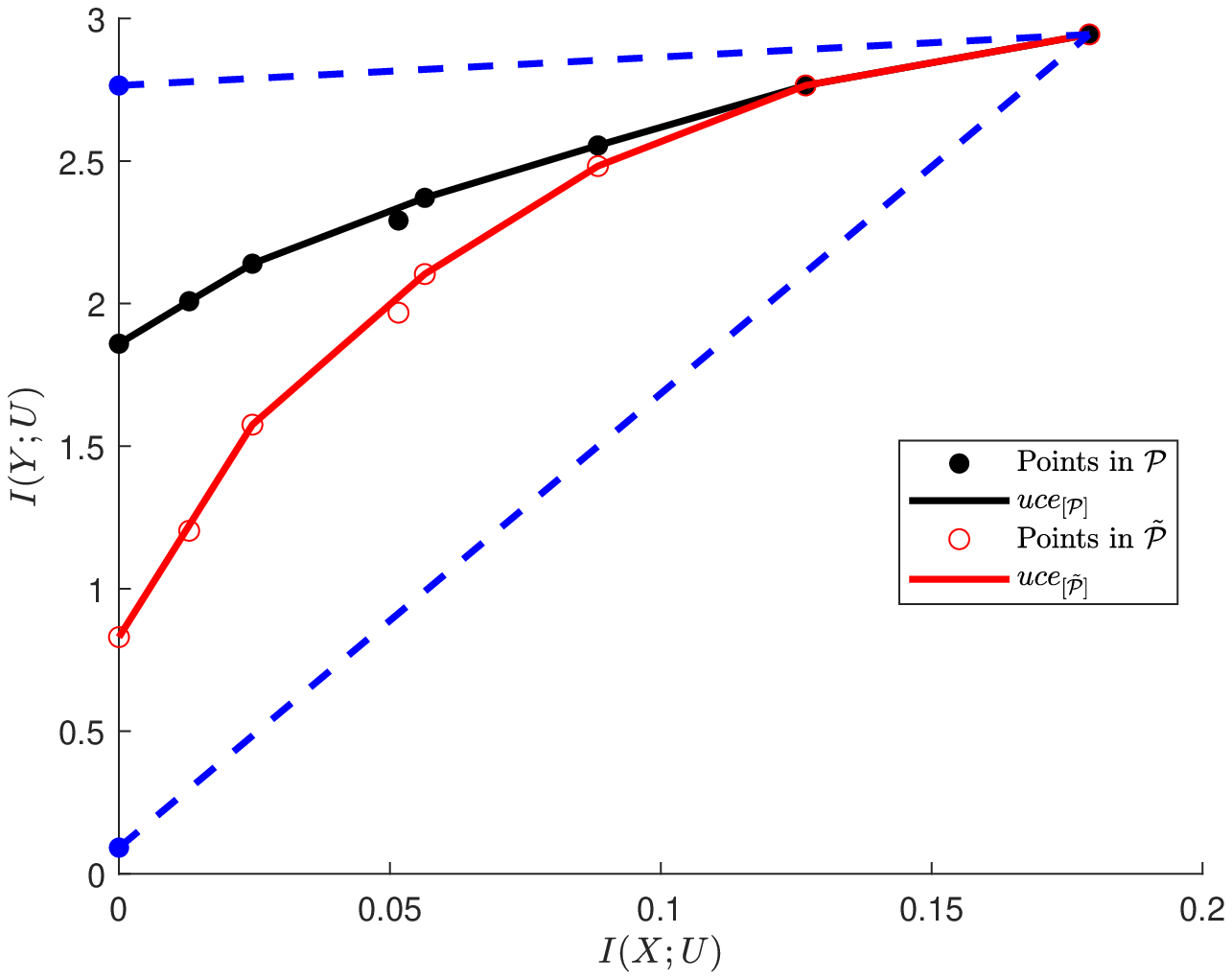}
    \caption{Lower bounds on $G_\epsilon(X,Y)$ in Algorithm \ref{algorithm3}.}
    \label{fig_a2}
\end{figure}
\subsection{Public data observation}
Figure \ref{fig_a3} illustrates the lower bounds on $g_\epsilon(X,Y)$ in Proposition \ref{prpr} and its simplified version, i.e., Algorithm \ref{algor4}. The filled black circles are the achievable points in Proposition \ref{prpr} whose upper concave envelope is drawn in solid black line. The empty red circles correspond to the achievable points in Algorithm \ref{algor4} whose upper concave envelope is plotted in solid red line. The top dashed blue line is as mentioned earlier in the full data observation, while the bottom dashed blue line is the line connecting $(0,g_0^L)$, where $g_0^L$ is given in (\ref{lbyek}), to $(I(X;Y),H(Y))$.
\begin{figure}[ht]
    \centering
    \includegraphics{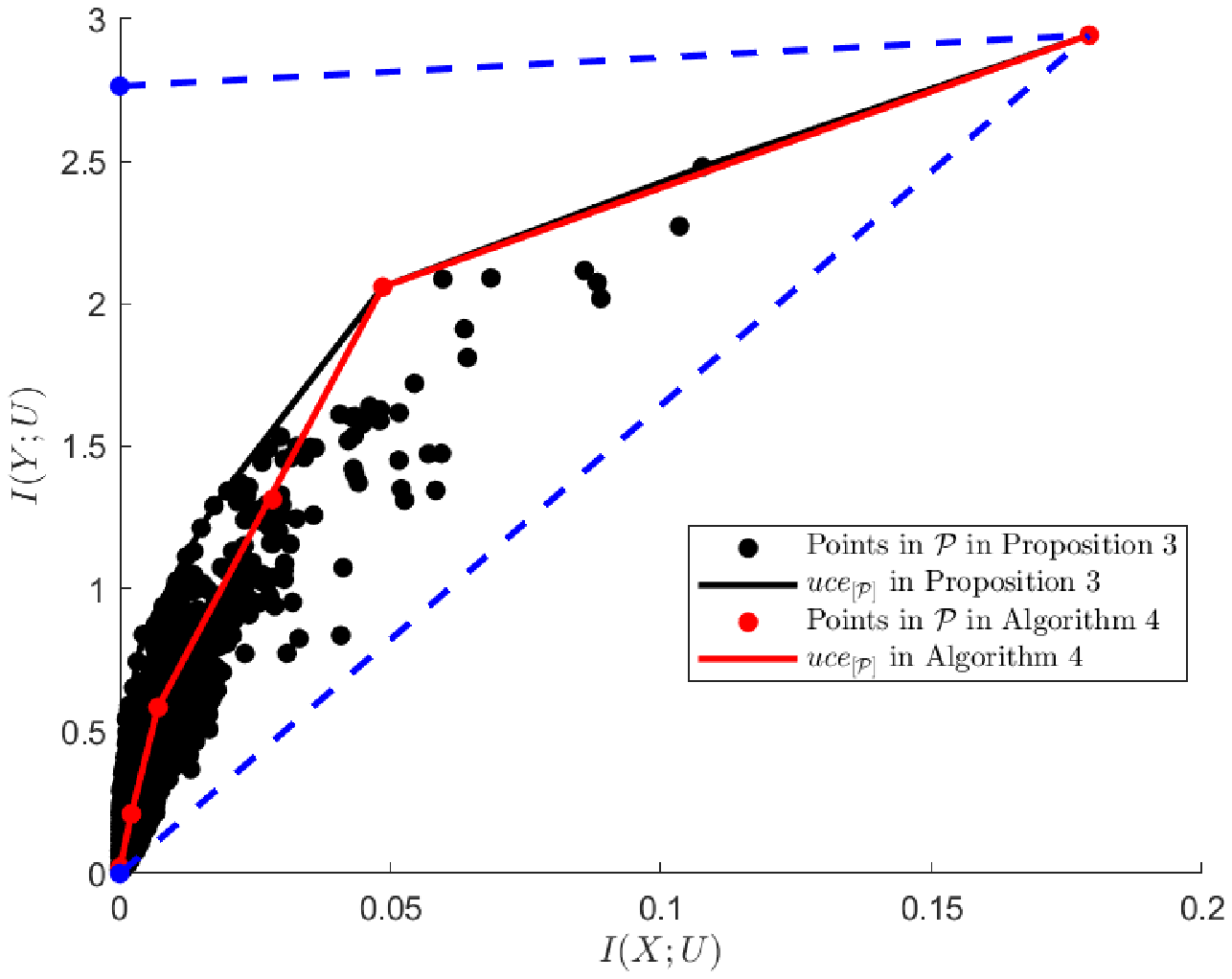}
    \caption{Lower bounds on $g_\epsilon(X,Y)$ in Proposition \ref{prpr} and Algorithm \ref{algor4}.}
    \label{fig_a3}
\end{figure}
\section{Conclusions}
Information-theoretic privacy is considered in this paper, in which a curator, aware of the joint distribution $p_{XY}$, wishes to maximize $I(Y;U)$ subject to $I(X;U)\leq\epsilon$. The optimization is over $p_{U|XY}$ (or $p_{U|Y}$) when curator has access to $(X,Y)$ (or only $Y$), and trade-off is captured by $G_\epsilon(X,Y)$ (or $g_\epsilon(X,Y)$). The problem is investigated from theoretical and practical point of view.
\appendices

\section{}\label{app22}
Fix an arbitrary $x_0\in\mathcal{X}$. Assume that for some $u_0\in\mathcal{U}^*$, we have $p(x_0,y'|u_0)>0$, and $p(x_0,y''|u_0)>0$ with $y'\neq y''$. It is shown that this cannot be optimal by construction. In other words, a mapping $p_{\hat{U}|XY}$ is constructed such that $I(X;\hat{U})=I(X;U^*)$ and $I(Y;\hat{U})>I(Y;U^*)$, which disproves the optimality of $p_{U^*|XY}$. 

Assume the random variable $\hat{U}\in(\mathcal{U}^*\backslash\{u_0\})\cup\{u_0',u_0''\}$ with $u_0',u_0''\not\in\mathcal{U}^*$ such that 
\begin{align*}
    p_{X,Y|\hat{U}}(x_0,y'|u_0'),p_{X,Y|\hat{U}}(x_0,y''|u_0'')&\triangleq p_{X,Y|U}(x_0,y'|u_0)+p_{X,Y|U}(x_0,y''|u_0)\\p_{X,Y|\hat{U}}(x_0,y''|u_0'),p_{X,Y|\hat{U}}(x_0,y'|u_0'')&\triangleq 0,\\
    p_{\hat{U}}(u_0')&\triangleq p_{U}(u_0)\frac{p_{X,Y|U}(x_0,y'|u_0)}{p_{X,Y|U}(x_0,y'|u_0)+p_{X,Y|U}(x_0,y''|u_0)}\nonumber\\
    p_{\hat{U}}(u_0'')&\triangleq p_{U}(u_0)\frac{p_{X,Y|U}(x_0,y''|u_0)}{p_{X,Y|U}(x_0,y'|u_0)+p_{X,Y|U}(x_0,y''|u_0)}\nonumber\\
    p_{X,Y,\hat{U}}(x,y,u)&\triangleq p_{X,Y,U}(x,y,u),\ \forall (x,y,u)\in\mathcal{X}\times\mathcal{Y}\times(\hat{U}\backslash\{u_0',u_0''\})\nonumber\\
    p_{X,Y,\hat{U}}(x,y|u)&\triangleq p_{X,Y,U}(x,y|u_0),\ \forall u\in\{u_0',u_0''\}, \forall (x,y)\neq (x_0,y'),(x_0,y'').
\end{align*}
It can be verified that in this construction, i) the marginal $p_{X,Y}$ is preserved in $(X,Y,\hat{U})$, ii) $H(X|U^*)=H(X|\hat{U})$ (which results from having $p_{X,\hat{U}}(\cdot,u)=p_{X,U^*}(\cdot,u),\ \forall u\in\mathcal{U}^*\backslash\{u_0,u_0',u_0''\}$, and $p_{X|U^*}(\cdot|u_0)=p_{X|\hat{U}}(\cdot|u_0')=p_{X|\hat{U}}(\cdot|u_0'')$, with $p(u_0)=p(u_0')+p(u_0'')$), and iii) $H(Y|\hat{U})<H(Y|U)$ due to strict concavity of entropy. Therefore, we have constructed $p_{\hat{U}|X,Y}$, such that $I(X;\hat{U})=I(X;U^*)$, and $I(Y;U^*)<I(Y;\hat{U})$ which contradicts the attainability of $G_0(X,Y)$ by $p_{U^*|X,Y}$. Hence, by noting that $x_0$ was chosen arbitrarily, we obtain
\begin{equation*}
    |\{y\in\mathcal{Y}|p(x,y|u^*)>0\}|\leq 1,\ \forall (x,u^*) \in\mathcal{X}\times\mathcal{U}^*,
\end{equation*}
which is equivalent to (\ref{chahar}) and (\ref{entfunction}) when $X\independent U$.

\section{}\label{app1}
Let $x^*\triangleq \argmin_xp(x)$ and define $\mathcal{X}_0\triangleq\{x\in\mathcal{X}|f(x)=f(x^*)\}$. Hence, we have
    \begin{equation}\label{lf0}
        p(x^*)\leq\min\left\{\textnormal{Pr}\{X\in\mathcal{X}_0\},1-\textnormal{Pr}\{X\in\mathcal{X}_0\}\right\}.
    \end{equation}
    Define $Z\triangleq\mathds{1}_{\{f(X)=f(x^*)\}}$ as a function of $f(X)$. Since $f(X)$ has at least two realizations, $Z$ is a Bernoulli random variable with parameter $\textnormal{Pr}\{X\in\mathcal{X}_0\}$. Therefore, we can write
    \begin{align}
        H(f(X))&\geq H(Z)\label{lf1}\\
        &=H_b\left(\textnormal{Pr}\{X\in\mathcal{X}_0\}\right)\nonumber\\
        &\geq H_b(p(x^*))\label{lf2},
    \end{align}
where (\ref{lf1}) follows from defining $Z$ as a function of $f(X)$, and (\ref{lf2}) results from (\ref{lf0}).  
\section{}\label{app3}
   If $i=j$ or $i=k$, the proof is complete by the non-negativity of entropy. Also, if any of $p_i,p_j,p_k$ is 0 or 1, the proof is complete, since it results in one of the following trivial possibilities i) $0\leq 0$, ii) $H_b(p_i)\leq H_b(1-p_i)$, or iii) the non-negativity of entropy. Therefore, we assume that the indices $i,j,k$ are all distinct and non of the mass probabilities is 0 or 1. Since the RHS of (\ref{nbv1}) is symmetric with respect to $p_j,p_k$, i.e., it doesn't change if we exchange $p_j$ and $p_k$, without loss of generality, assume that $p_j\leq p_k$. Therefore, we have
   \begin{align}
       H_b(p_i)&=H_b(p_j+p_k)\nonumber\\
       &< H_b(p_k)+H'_b(p_k)p_j\label{nbv2}\\
       &\leq H_b(p_k)+H'_b(p_j)p_j\label{nbv3}\\
       &< H_b(p_k)+H(p_j)\label{nbv4},
   \end{align}
where (\ref{nbv2}) results from Taylor expansion of $H_b(\cdot)$ and its strict concavity, i.e., $H''_b(\cdot)<0$. The latter also results in (\ref{nbv3}) and (\ref{nbv4}), i.e., $H'_b(p_k)\leq H'_b(p_j)$ since $p_j\leq p_k$, and $H'_b(p_j)p_j<H_b(p_j)$.
\bibliography{REFERENCE}

\begin{thebibliography}{10}
\providecommand{\url}[1]{#1}
\csname url@samestyle\endcsname
\providecommand{\newblock}{\relax}
\providecommand{\bibinfo}[2]{#2}
\providecommand{\BIBentrySTDinterwordspacing}{\spaceskip=0pt\relax}
\providecommand{\BIBentryALTinterwordstretchfactor}{4}
\providecommand{\BIBentryALTinterwordspacing}{\spaceskip=\fontdimen2\font plus
\BIBentryALTinterwordstretchfactor\fontdimen3\font minus
  \fontdimen4\font\relax}
\providecommand{\BIBforeignlanguage}[2]{{%
\expandafter\ifx\csname l@#1\endcsname\relax
\typeout{** WARNING: IEEEtran.bst: No hyphenation pattern has been}%
\typeout{** loaded for the language `#1'. Using the pattern for}%
\typeout{** the default language instead.}%
\else
\language=\csname l@#1\endcsname
\fi
#2}}
\providecommand{\BIBdecl}{\relax}
\BIBdecl

\bibitem{RG-JSAIT}
B.~{Rassouli} and D.~{Gündüz}, ``On perfect privacy,'' \emph{IEEE Journal on
  Selected Areas in Information Theory}, vol.~2, no.~1, pp. 177--191, 2021.

\bibitem{Calmon1}
F.~Calmon and N.~Fawaz, ``Privacy against statistical inference,'' in
  \emph{50th Annual Allerton Conference}, Illinois, USA, Oct. 2012, pp.
  1401--1407.

\bibitem{Makhdoumi}
A.~Makhdoumi, S.~Salamatian, N.~Fawaz, and M.~M\'{e}dard, ``From the
  information bottleneck to the privacy funnel,'' in \emph{IEEE Information
  Theory Workshop (ITW)}, 2014, pp. 501--505.

\bibitem{SG19}
S.~{Sreekumar} and D.~{Gündüz}, ``Optimal privacy-utility trade-off under a
  rate constraint,'' in \emph{2019 IEEE International Symposium on Information
  Theory (ISIT)}, 2019, pp. 2159--2163.

\bibitem{info7010015}
\BIBentryALTinterwordspacing
S.~Asoodeh, M.~Diaz, F.~Alajaji, and T.~Linder, ``Information extraction under
  privacy constraints,'' \emph{Information}, vol.~7, no.~1, 2016. [Online].
  Available: \url{https://www.mdpi.com/2078-2489/7/1/15}
\BIBentrySTDinterwordspacing

\bibitem{Issa}
I.~{Issa}, A.~B. {Wagner}, and S.~{Kamath}, ``An operational approach to
  information leakage,'' \emph{IEEE Transactions on Information Theory},
  vol.~66, no.~3, pp. 1625--1657, 2020.

\bibitem{Shkel}
Y.~Y. Shkel, R.~S. Blum, and H.~V. Poor, ``Secrecy by design with applications
  to privacy and compression,'' \emph{IEEE Transactions on Information Theory},
  vol.~67, no.~2, pp. 824--843, 2021.

\bibitem{Zamani}
A.~Zamani, T.~J. Oechtering, and M.~Skoglund, ``Data disclosure with non-zero
  leakage and non-invertible leakage matrix,'' \emph{IEEE Transactions on
  Information Forensics and Security}, vol.~17, pp. 165--179, 2022.

\bibitem{RRG}
B.~Rassouli, F.~E. Rosas, and D.~Gündüz, ``Data disclosure under perfect
  sample privacy,'' \emph{IEEE Transactions on Information Forensics and
  Security}, vol.~15, pp. 2012--2025, 2020.

\bibitem{Liao}
J.~Liao, O.~Kosut, L.~Sankar, and F.~du~Pin~Calmon, ``Tunable measures for
  information leakage and applications to privacy-utility tradeoffs,''
  \emph{IEEE Transactions on Information Theory}, vol.~65, no.~12, pp.
  8043--8066, 2019.

\bibitem{Cuff}
\BIBentryALTinterwordspacing
P.~Cuff and L.~Yu, ``Differential privacy as a mutual information constraint,''
  ser. CCS '16.\hskip 1em plus 0.5em minus 0.4em\relax New York, NY, USA:
  Association for Computing Machinery, 2016, p. 43–54. [Online]. Available:
  \url{https://doi.org/10.1145/2976749.2978308}
\BIBentrySTDinterwordspacing

\bibitem{RGTIFS}
B.~Rassouli and D.~Gündüz, ``Optimal utility-privacy trade-off with total
  variation distance as a privacy measure,'' \emph{IEEE Transactions on
  Information Forensics and Security}, vol.~15, pp. 594--603, 2020.

\bibitem{Wang}
H.~Wang, L.~Vo, F.~P. Calmon, M.~Médard, K.~R. Duffy, and M.~Varia, ``Privacy
  with estimation guarantees,'' \emph{IEEE Transactions on Information Theory},
  vol.~65, no.~12, pp. 8025--8042, 2019.

\bibitem{Tishby}
N.~Tishby, F.~Pereira, and W.~Bialek, ``The information bottleneck method,''
  \emph{arXiv preprint physics/0004057}, 2000.

\bibitem{Calmon2}
F.~Calmon, A.~Makhdoumi, and M.~M\'{e}dard, ``Fundamental limits of perfect
  privacy,'' in \emph{IEEE Int. Symp. Inf. Theory (ISIT)}, 2015, pp.
  1796--1800.

\bibitem{Berger}
T.~Berger and R.~Yeung, ``Multiterminal source encoding with encoder
  breakdown,'' \emph{IEEE Trans. Inf. Theory}, pp. 237--244, 1989.

\bibitem{Shahab1}
S.~Asoodeh, F.~Alajaji, and T.~Linder, ``Notes on information-theoretic
  privacy,'' in \emph{52nd Annual Allerton Conference}, Illinois, USA, Oct.
  2014, pp. 1272--1278.

\bibitem{murty}
\BIBentryALTinterwordspacing
K.~Murty, \emph{Linear Programming}.\hskip 1em plus 0.5em minus 0.4em\relax
  Wiley, 1984. [Online]. Available:
  \url{https://books.google.co.uk/books?id=ibQJvAEACAAJ}
\BIBentrySTDinterwordspacing

\bibitem{arxv4}
\BIBentryALTinterwordspacing
B.~Rassouli and D.~G{\"{u}}nd{\"{u}}z, ``On perfect privacy and maximal
  correlation,'' \emph{CoRR}, vol. abs/1712.08500v4, 2020. [Online]. Available:
  \url{https://arxiv.org/abs/1712.08500v4}
\BIBentrySTDinterwordspacing

\bibitem{Elgamal}
A.~E. Gamal and Y.-H. Kim, \emph{Network Information Theory}.\hskip 1em plus
  0.5em minus 0.4em\relax Cambridge University Press, 2012.

\bibitem{Anantharam}
V.~Anantharam, A.~Gohari, S.~Kamath, and C.~Nair, ``On maximal correlation,
  hypercontractivity, and the data processing inequality studied by {E}rkip and
  {C}over,'' \emph{https://arxiv.org/pdf/1304.6133.pdf}, Apr. 2013.

\end{thebibliography}
\bibliographystyle{IEEEtran}
\end{document}